\documentclass[%
  reprint,
  superscriptaddress,
]{revtex4-2}

\usepackage{amsmath}
\usepackage{amssymb}
\usepackage{graphicx}
\usepackage[bookmarks=false]{hyperref}
\usepackage{bm,bbm}
\hypersetup{colorlinks=true,citecolor=blue,linkcolor=red,%
urlcolor=blue,pdfstartview=FitH,bookmarksopen=true}

\usepackage[T1]{fontenc}
\usepackage[osf,sc]{mathpazo}
\usepackage{dsfont}

\usepackage{amsthm}

\newtheorem{theorem}{Theorem}

\newtheorem{lemma}[theorem]{Lemma}
\newtheorem{corollary}[theorem]{Corollary}
\newtheorem{observation}[theorem]{Observation}

\DeclareMathOperator{\Tr}{Tr}

\DeclareMathOperator{\rank}{rank}
\DeclareMathOperator{\conv}{conv}

\DeclareMathOperator{\id}{id}
\newcommand{\bra}[1]{\langle #1\rvert}
\newcommand{\ket}[1]{\lvert #1\rangle}
\newcommand{\mean}[1]{\langle #1\rangle}

\newcommand{\braket}[2]{\langle #1\vert #2\rangle}
\newcommand{\ketbra}[2]{\vert #1 \rangle \langle #2 \vert}
\newcommand{\abs}[1]{\lvert #1\rvert}
\newcommand{\norm}[1]{\lVert #1\rVert}
\newcommand{\vect}[1]{\bm{#1}}
\newcommand{\md}{\mathrm{d}}

\newcommand{\mi}{\mathrm{i}}

\newcommand{\E}{\mathds{J}}
\newcommand{\I}{\mathds{1}}
\newcommand{\FA}{\mathrm{~\forall~}}
\newcommand{\Sep}{\mathrm{SEP}}

\newcommand{\maxover}[1][]{\underset{#1}{\mathrm{max}}}
\newcommand{\minover}[1][]{\underset{#1}{\mathrm{min}}}
\newcommand{\findover}[1][]{\underset{#1}{\mathrm{find}}}
\newcommand{\subto}{\mathrm{~s.t.}}

\newcommand{\cE}{\mathcal{E}}
\newcommand{\cF}{\mathcal{F}}

\newcommand{\cH}{\mathcal{H}}
\newcommand{\cM}{\mathcal{M}}
\newcommand{\cO}{\mathcal{O}}
\newcommand{\cP}{\mathcal{P}}
\newcommand{\cS}{\mathcal{S}}

\newcommand{\dC}{\mathds{C}}
\newcommand{\dF}{\mathds{F}}
\newcommand{\dR}{\mathds{R}}

\newcommand{\dN}{\mathds{N}}

\newcommand{\tL}{\widetilde{\Lambda}}
\newcommand{\tX}{\widetilde{X}}

\newcommand{\bmid}{\;\big|\;}

\begin{document}

\title{Quantum-Inspired Hierarchy for Rank-Constrained Optimization}

\author{Xiao-Dong Yu}
\affiliation{Department of Physics, Shandong University, Jinan 250100, China}
\affiliation{Naturwissenschaftlich-Technische Fakult\"at, Universit\"at Siegen,
Walter-Flex-Str. 3, D-57068 Siegen, Germany}

\author{Timo Simnacher}
\affiliation{Naturwissenschaftlich-Technische Fakult\"at, Universit\"at Siegen,
Walter-Flex-Str. 3, D-57068 Siegen, Germany}

\author{H. Chau Nguyen}
\affiliation{Naturwissenschaftlich-Technische Fakult\"at, Universit\"at Siegen,
Walter-Flex-Str. 3, D-57068 Siegen, Germany}

\author{Otfried G\"uhne}
\affiliation{Naturwissenschaftlich-Technische Fakult\"at, Universit\"at Siegen,
Walter-Flex-Str. 3, D-57068 Siegen, Germany}

\date{\today}

\begin{abstract}
  Many problems in information theory can be reduced to optimizations  over 
  matrices, where the rank of the matrices is constrained. We establish a link 
  between rank-constrained optimization and the theory of quantum entanglement. 
  More precisely, we prove that a large class of rank-constrained semidefinite 
  programs can be written as a convex optimization over separable quantum 
  states and, consequently, we construct a complete  hierarchy of semidefinite 
  programs for solving the  original problem. This hierarchy not only provides 
  a sequence of certified bounds for the rank-constrained optimization problem, 
  but also gives pretty good and often exact values in practice when the lowest 
  level of the hierarchy is considered. We demonstrate that our approach can be 
  used for relevant problems in quantum information processing, such as the 
  optimization over pure states, the characterization of mixed unitary channels 
  and faithful  entanglement, and quantum contextuality, as well as in 
  classical information theory including the maximum cut problem, 
  pseudo-Boolean optimization, and  the orthonormal representation of graphs.  
  Finally, we show that our ideas can be extended to rank-constrained quadratic 
  and higher-order programming.
\end{abstract}

\maketitle

\section{Introduction}
The mathematical theory of optimization has become a vital tool in various 
branches of science. This is not only due to the fact that some central problems
(e.g., finding the ground state energy of a given Hamiltonian in condensed matter 
physics) are by definition optimization problems, where mathematical methods
can be directly applied. It also turned out that other physical problems, which 
are not directly optimizations, can be reformulated as optimization tasks.

In physics, many efforts have been devoted to so-called semidefinite programs 
(SDPs), which is a class of highly tractable convex optimization problems. In 
quantum information theory, they have been used to characterize quantum 
entanglement \cite{Doherty.etal2002} and quantum correlations 
\cite{Navascues.etal2007}. In condensed matter physics, SDPs are relevant for 
solving ground-state problems \cite{Barthel.Huebener2012}. In conformal field 
theory, they have been employed for bootstrap problems 
\cite{SimmonsDuffin2015}. In fact, SDPs also found widespread applications in 
more general topics beyond physics; examples include the Shannon capacity of 
graphs \cite{Lovasz1979} and global polynomial optimization 
\cite{Lasserre2001,Parrilo2000}.

In many cases, however, one cannot directly formulate an SDP, as some 
non-convex constraints remain. Well-known examples are the characterization of 
quantum correlations for a fixed dimension 
\cite{Navascues.etal2008,Navascues.Vertesi2015}, the determination of the 
faithfulness of quantum entanglement \cite{Guehne.etal2021}, the ground state 
energy in spin glasses \cite{Barahona.etal1988}, and compressed sensing 
tomography \cite{Gross.etal2010}. Interestingly, these non-convex optimization 
problems share a common structure: They can be formulated as SDPs with an extra 
rank constraint.  Apart from these physics examples, rank-constrained 
optimizations are also widely-used in signal processing, model reduction, and 
system identification \cite{Markovsky2019}. All these applications demonstrate 
that to achieve significant progress, it would be highly desirable to develop 
techniques to deal with rank constraints in SDPs.

In this paper, we provide a method to deal with rank constraints based 
on the theory of quantum entanglement. More precisely, we prove that a large 
class of rank-constrained SDPs can be written as a convex optimization over 
separable two-party quantum states. Based on this, a complete hierarchy of SDPs 
can be constructed. In this way, we demonstrate that quantum information theory 
does not only benefit from ideas of optimization theory, but the results 
obtained in this field can also be used to study mathematical problems (like 
the Max-Cut problem) from a fresh perspective. Notably, unlike the widely-used 
local optimization methods \cite{Orsi.etal2006,Sun.Dai2017}, our method gives 
global bounds for the rank-constrained optimization. This makes our method 
particularly useful for certification problems in quantum information, where 
global bounds are usually necessary to establish conclusions with certainty.

In order to demonstrate the usefulness of our method, we first show that the 
optimization over pure quantum states or unitary matrices in quantum 
information can be naturally written as a rank-constrained optimization. This 
provides a complete characterization of faithful entanglement 
\cite{Weilenmann.etal2020,Guehne.etal2021} and of mixed unitary channels 
\cite{Alberti.Uhlmann1982,Lee.Watrous2020}. The second example concerns the  
dimension-bounded orthonormal representations of graphs \cite{Lovasz2019}, 
which is closely related to the existence of quantum contextuality in a given 
measurement configuration \cite{Cabello.etal2014,Ramanathan.Horodecki2014}.  
Finally, we consider the maximum cut (Max-Cut) problem 
\cite{Goemans.Williamson1995} and quadratic optimization over Boolean vectors 
\cite{Boros.Hammer2002}. These problems are not only very important in 
classical information theory, but also frequently encountered in statistical 
physics \cite{Kirkpatrick.etal1983} and complex networks 
\cite{Dorogovtsev.etal2008}. Remarkably, solving these optimization problems 
with noisy intermediate-scale quantum computers has drawn a lot of research 
interest in recent years 
\cite{Lucas2014,Boixo.etal2014,Farhi.etal2014,Farhi.Harrow2019}.  Consequently, 
our methods may be used to compare the operational performance of 
intermediate-scale quantum devices with different classical algorithms.

Our paper is organized as follows. In Sec.~II we explain the core idea of our 
method, first for matrices with complex entries, then for real matrices.  We 
also discuss how symmetries can be used to simplify the resulting sequence of 
SDPs. In Sec.~III we present several examples, where our methods can be 
applied. In Sec.~IV we discuss the more general form of rank-constrained SDPs, 
as well as rank-constrained quadratic and higher-order optimization problems.  
Finally, we conclude and discuss open problems.

\section{Rank-constrained SDP and quantum entanglement}\label{sec:main}

SDPs are widely used in various branches of science, especially in the quantum 
regime. One of the reasons is that density matrices are automatically 
positive semidefinite, so that related optimization problems naturally 
contain some semidefinite constraints. Another important reason that 
SDPs have drawn a lot of interest is that there are efficient algorithms for 
solving them \cite{Boyd.Vandenberghe2004}, moreover, symmetries can be used to 
drastically simplify the SDPs \cite{Vallentin2009,Rosset2018,Aguilar.etal2018}.  
In many cases, however, one cannot directly formulate an SDP, as some 
non-convex constraints remain.  This happens, for example, when the underlying 
quantum states are required to be pure or the quantum system is of bounded 
dimension.  These restrictions will introduce some extra rank constraints, 
which is the main focus of this paper.

The prototype optimization problem we are considering is given by
\begin{equation}
  \begin{aligned}
    &\maxover[\rho]  && \Tr(X\rho)\\
    &\subto && \Lambda(\rho)=Y,~\Tr(\rho)=1,\\
    &       && \rho\ge0,~\rank(\rho)\le k.
  \end{aligned}
  \label{eq:SDPRank}
\end{equation}
Here, $\rho$ and $X$ are $n \times n$ matrices with real ($\dF=\dR$) or complex 
($\dF=\dC$) entries, which are symmetric (respectively Hermitian). $\Lambda$ is 
a map from matrices in $\dF^{n\times n}$ to matrices in $\dF^{m\times m}$ and 
consequently $Y\in \dF^{m\times m}$.
In this way, the constraint $\Lambda(\rho)=Y$ denotes all affine equality 
constraints. While our main results are formulated for the rank-constrained SDP 
in Eq.~\eqref{eq:SDPRank}, we stress that our method can also be extended to 
more general cases with (semidefinite) inequality constraints $\Lambda(\rho)\le 
Y$, without the normalization condition $\Tr(\rho)=1$, or even without the 
positivity constraint $\rho\ge 0$.

\subsection{Optimization over complex matrices}

We start with $\dF=\dC$ for the optimization in Eq.~\eqref{eq:SDPRank}, where 
we can easily apply the results from quantum information.
Let $\cF$ be the feasible region of optimization~\eqref{eq:SDPRank}, i.e.,
\begin{equation}
  \hspace*{-.6ex}
  \cF=\big\{\rho\;\big|\;\Lambda(\rho)=Y,\Tr(\rho)=1,
  \rho\ge0,\rank(\rho)\le k\big\}.
  \label{eq:feasibleRegion}
\end{equation}
With the terminology in quantum information, $\cF$ is a subset 
of quantum states in the quantum system (or Hilbert space) 
$\dC^n$ \cite{Nielsen.Chuang2000}.

Now, we recall the notion of state purification in quantum information 
\cite{Nielsen.Chuang2000}. Let $\cH_1=\dC^n$ and $\cH_2=\dC^k$ be two quantum 
systems (Hilbert spaces). Then, a quantum state $\rho$ in $\cH_1$ satisfies 
that $\rank(\rho)\le k$ if and only if there exists a pure state 
$\ket{\varphi}\in\cH_1\otimes\cH_2$ such that 
$\Tr_2(\ket{\varphi}\bra{\varphi})=\rho$, where $\Tr_2(\cdot)$ is the partial 
trace operation on quantum system $\cH_2$. Thus, $\cF$ can be written as 
$\Tr_2(\cP)$, where
\begin{equation}
  \cP=\big\{\ket{\varphi}\bra{\varphi} \bmid
    \tL(\ket{\varphi}\bra{\varphi})=Y,~\braket{\varphi}{\varphi}=1\big\}
  \label{eq:defP}
\end{equation}
with $\tL(\cdot)=\Lambda[\Tr_2(\cdot)]$. Let $\conv(\cP)$ be the convex hull of 
$\cP$, i.e., all states of the form $\sum_i p_i\ket{\varphi_i}\bra{\varphi_i}$, 
where the $p_i$ form a probability distribution and 
$\ket{\varphi_i}\bra{\varphi_i}\in\cP$. By noting that the maximum value of 
a linear function can always be achieved at extreme points, the optimization in 
Eq.~\eqref{eq:SDPRank} is equivalent to
\begin{equation}
  \max_{\rho\in\cF}\Tr(X\rho)=\max_{\Phi\in\conv(\cP)}\Tr(\tX\Phi),
  \label{eq:purify}
\end{equation}
where $\tX=X\otimes\I_k$ with $\I_k$ being the identity operator on $\cH_2$.

\begin{figure}[t!]
  \centering
  \includegraphics[width=0.8\linewidth]{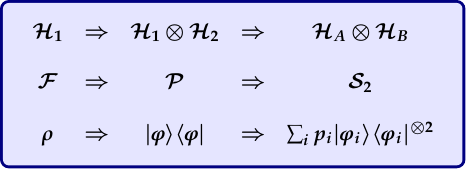}
  \caption{An illustration of the relations between the feasible region $\cF$, 
  the purification $\cP$, and the two-party extension $\cS_2$.  $\ket{\varphi}$ 
is a purification of $\rho$, $\cH_A=\cH_B=\cH_1\otimes\cH_2$, and 
$\ket{\varphi_i}$ are states in $\cP$.}
  \label{fig:relations}
\end{figure}

Equation~\eqref{eq:purify} implies that if we can fully characterize 
$\conv(\cP)$, optimization \eqref{eq:SDPRank} is solved. To this end, we need 
to introduce the notion of separable states that has been widely studied in 
quantum information \cite{Horodecki.etal2009,Guehne.Toth2009}.  More 
specifically, we let $\cH_A=\cH_B=\cH_1\otimes\cH_2=\dC^n\otimes\dC^k$ and 
define the separability cone $\Sep$ on $\cH_A\otimes\cH_B$ as
\begin{equation}
  \Sep=\conv\big\{M_A\otimes N_B\bmid M_A\ge 0,~N_B\ge 0\big\}.
  \label{eq:SEP}
\end{equation}
Physically, $\Sep$ is the set of all unnormalized separable quantum states 
(besides the zero matrix).  $\Sep$ is a proper convex cone, and its dual cone 
is given by
\begin{equation}
  \Sep^*=\big\{W_{AB}\bmid \Tr(W_{AB}\Phi_{AB})\ge 0~~\FA\Phi_{AB}\in\Sep\big\},
  \label{eq:SEPDual}
\end{equation}
which, in  the language of quantum information, corresponds to the set of 
entanglement witnesses and all positive 
semidefinite matrices.

Then, we consider the two-party extension of the purified feasible states,
\begin{equation}
  \cS_2=\conv\Big(\Big\{\ket{\varphi}\bra{\varphi}_A
      \otimes\ket{\varphi}\bra{\varphi}_B \;\Big|\;
  \ket{\varphi}\bra{\varphi}\in\cP\Big\}\Big),
  \label{eq:defS2}
\end{equation}
where $\ket{\varphi}_A$ and $\ket{\varphi}_B$ are the same state but belong to 
$\cH_A$ and $\cH_B$, respectively; see Fig.~\ref{fig:relations}. One can easily 
check that
\begin{equation}
  \Tr_B(\cS_2)=\conv(\cP).
  \label{eq:S22S}
\end{equation}
The benefit of introducing the two-party extension is that we can fully 
characterize $\cS_2$ with the separability cone $\Sep$, and hence $\conv(\cP)$ 
is also fully characterized.

The first necessary condition for $\Phi_{AB}\in\cS_2$ is that it is separable 
with respect to the bipartition $(A|B)$, i.e.,
\begin{equation}
  \Phi_{AB}\in\Sep,~\Tr(\Phi_{AB})=1.
  \label{eq:S2Sep}
\end{equation}
Second, $\Phi_{AB}\in\cS_2$ implies that it is within the 
symmetric subspace of $\cH_A\otimes\cH_B$. Mathematically, this can be written 
as
\begin{equation}
  V_{AB}\Phi_{AB}=\Phi_{AB},
  \label{eq:S2Symm}
\end{equation}
where $V_{AB}$ is the swap operator between $\cH_A$ and 
$\cH_B$, i.e., 
$V_{AB}\ket{\psi_1}_{A}\ket{\psi_2}_B=\ket{\psi_2}_{A}\ket{\psi_1}_B$ for any 
$\ket{\psi_1},\ket{\psi_2}\in\dC^n\otimes\dC^k$. The final necessary condition 
needed arises from Eq.~\eqref{eq:defP}, i.e., 
$\tL(\ket{\varphi}\bra{\varphi})=Y$ for $\ket{\varphi}\bra{\varphi}\in\cP$.  
Then, Eq.~\eqref{eq:defS2} implies that
\begin{equation}
  \tL_A\otimes\id_B(\Phi_{AB})=Y\otimes\Tr_A(\Phi_{AB})
  \label{eq:S2Marginal}
\end{equation}
for all $\Phi_{AB}\in\cS_2$, where $\tL_A(\cdot)$ is the map 
$\tL(\cdot)=\Lambda[\Tr_2(\cdot)]$ acting on system $\cH_A$ only, and $\id_B$ 
is the identity map on $\cH_B$. Hereafter, we also use a similar convention for 
matrices, e.g., $\tX_A$ denotes the matrix $\tX$ on system $\cH_A$.

Surprisingly, the conditions in 
Eqs.~(\ref{eq:S2Sep}, \ref{eq:S2Symm}, \ref{eq:S2Marginal}) are also sufficient 
for $\Phi_{AB}\in\cS_2$. The basic idea for the proof is that
Eqs.~(\ref{eq:S2Sep}, \ref{eq:S2Symm}) imply that $\Phi_{AB}$ is a separable 
state in the symmetric subspace, such that it always admits the form 
\cite{Toth.Guehne2009}
\begin{equation}
  \Phi_{AB}=\sum_ip_i\ket{\varphi_i}\bra{\varphi_i}_A
  \otimes\ket{\varphi_i}\bra{\varphi_i}_B.
  \label{eq:symmPhi}
\end{equation}
Then, Eq.~\eqref{eq:S2Marginal} implies that 
$\tL(\ket{\varphi_i}\bra{\varphi_i})=Y$ for all $i$; see 
Appendix~\ref{app:sepOpt} for the proof. With the full characterization of 
$\cS_2$ from Eqs.~(\ref{eq:S2Sep}, \ref{eq:S2Symm}, \ref{eq:S2Marginal}), we 
can directly rewrite the rank-constrained SDP in Eq.~\eqref{eq:SDPRank}. The
result is a so-called conic program, as one constraint is defined by the cone
of separable states.

\begin{theorem}
  For $\dF=\dC$, the rank-constrained SDP in Eq.~\eqref{eq:SDPRank} is 
  equivalent to the conic program
  \begin{align}
    &\maxover[\Phi_{AB}]&& \Tr(\tX_A\otimes\I_B\Phi_{AB})\label{eq:sepOpt}\\
    &\subto && 
    \Phi_{AB}\in\Sep,~\Tr(\Phi_{AB})=1,
    ~V_{AB}\Phi_{AB}=\Phi_{AB},\notag\\
    &       &&\tL_A\otimes\id_B(\Phi_{AB})
    =Y\otimes\Tr_A(\Phi_{AB}).\notag
  \end{align}
  \label{thm:sepOpt}
\end{theorem}

This conic program cannot be directly solved because the characterization of 
the separability cone $\Sep$ is still an NP-hard problem \cite{Gurvits2003}.  
Actually, this is expected, because the rank-constrained SDP is, in general, 
also NP-hard.  However, in quantum information theory many outer relaxations of 
the separability cone $\Sep$ are known.  For example, the positive partial 
transpose (PPT) criterion provides a pretty good approximation for 
low-dimensional quantum systems \cite{Peres1996,Horodecki.etal1996}. More 
generally, inspired by the symmetric extension criterion 
\cite{Doherty.etal2002,Werner1989b,Doherty.Wehner2012}, which says that the 
two-party reduced states of all $N$-party symmetric states are asymptotically 
separable when taking $N\to+\infty$, we obtain a complete hierarchy for 
rank-constrained optimization in Eq.~\eqref{eq:SDPRank}.

To express the hierarchy, we need to introduce the notion of symmetric 
subspaces for multiple parties. We label the $N$ parties as $A,B,\dots,Z$ and 
$\cH_A=\cH_B=\dots=\cH_Z=\cH_1\otimes\cH_2=\dC^n\otimes\dC^k$. For any 
$\cH^{\otimes N}:=\cH_A\otimes\cH_B\otimes\cdots\otimes\cH_Z$, the symmetric 
subspace is defined as
\begin{equation}
  \left\{\ket{\Psi}\in\mathcal{H}^{\otimes N} \;\Big|\;
  V_{\sigma}\ket{\Psi}=\ket{\Psi} ~~\FA \sigma\in S_N\right\},
  \label{eq:symmetricVector}
\end{equation}
where $S_N$ is the permutation group over $N$ symbols and $V_{\sigma}$ are the 
corresponding operators on the $N$ parties $A,B,\dots,Z$. Let $P_N^+$ denote 
the orthogonal projector onto the symmetric subspace of $\mathcal{H}^{\otimes 
N}$, then $P_N^+$ can be explicitly written as
\begin{equation}
  P_N^+=\frac{1}{N!}\sum_{\sigma\in S_N}V_\sigma.
  \label{eq:symmetricProjector}
\end{equation}
Hereafter, without ambiguity, we will also use $P_N^+$ to denote the 
corresponding symmetric subspace. For example, a state $\Phi_{AB\cdots Z}$ is 
within the symmetric space, i.e., $\Phi_{AB\cdots 
Z}=\sum_ip_i\ket{\Psi_i}\bra{\Psi_i}$ for $\ket{\Psi_i}\in P_N^+$, if and only 
if $P_N^+\Phi_{AB\cdots Z}P_N^+=\Phi_{AB\cdots Z}$.

\begin{figure}[t!]
  \centering
  \includegraphics[width=.45\textwidth]{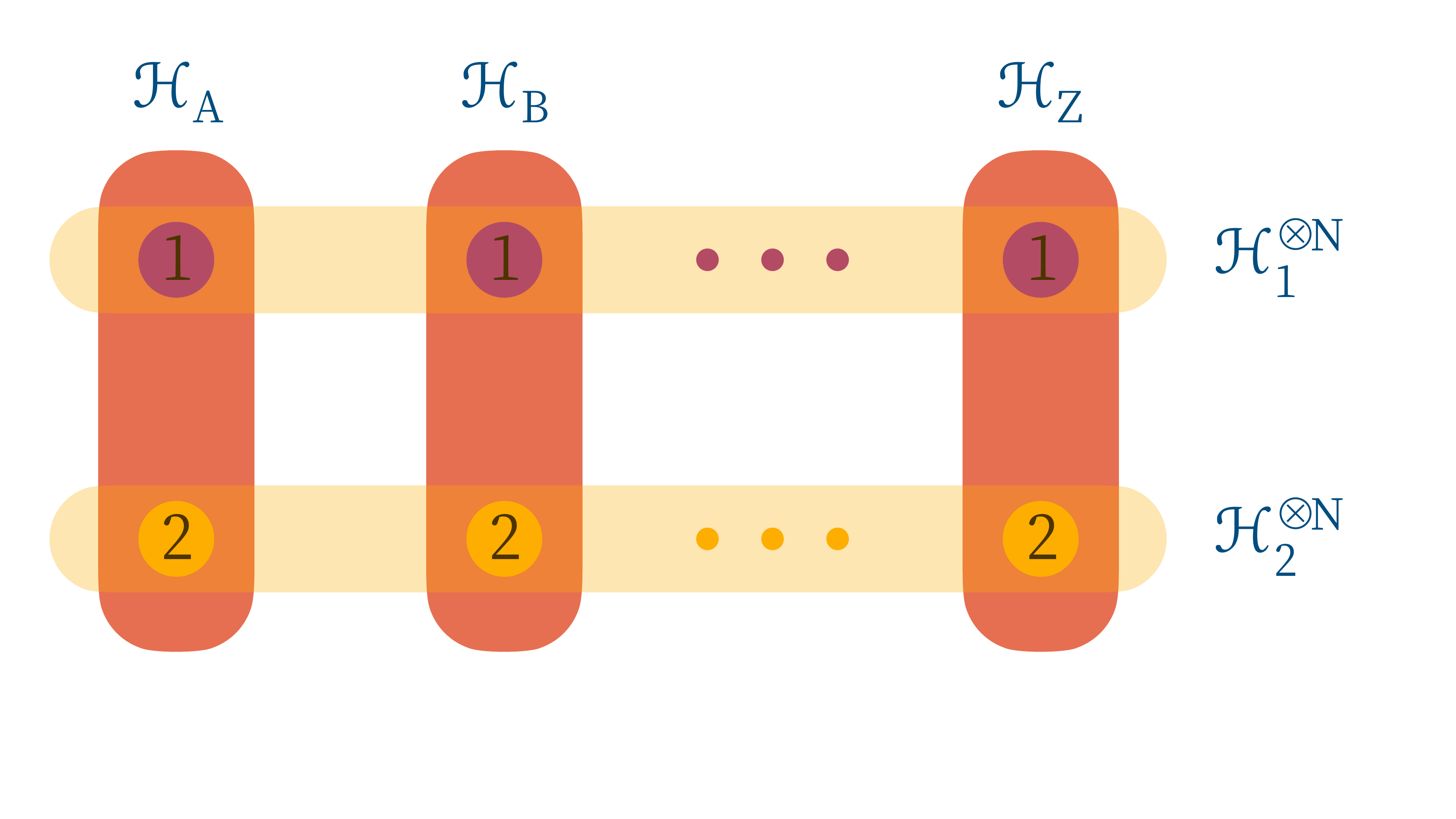}
  \caption{An illustration of the $N$-party extension $\Phi_{AB\cdots Z}$.  
    Here $\cH_1=\dC^n$ is the $n$-dimensional Hilbert space on which the 
    rank-constrained optimization is defined; $\cH_2=\dC^k$ is the 
    $k$-dimension auxiliary Hilbert space that is used for purifying the 
    rank-$k$ (more precisely, rank no larger than $k$) states in $\cH_1=\dC^n$.  
    Sometimes, we also denote $\cH_1^{\otimes N}$ as 
    $\cH_{A_1}\otimes\cH_{B_1}\otimes\cdots\otimes\cH_{Z_1}$ in order to 
    distinguish the Hilbert spaces $\cH_1$ for different parties (similarly for 
  $\cH_2^{\otimes N}=\cH_{A_2}\otimes\cH_{B_2}\otimes\cdots\otimes\cH_{Z_2}$).}
  \label{fig:Ncopy}
\end{figure}

Now we are ready to state the complete hierarchy for rank-constrained 
optimization; see Appendix~\ref{app:symExtOpt} for the proof.
\begin{theorem}
  For $\dF=\dC$, let $\xi$ be the solution of the rank-constrained SDP in 
  Eq.~\eqref{eq:SDPRank}. Then, for any $N$, $\xi$ is upper bounded by the 
  solution $\xi_N$ of the SDP hierarchy
  \begin{alignat}{2}
    &\maxover[\Phi_{AB\cdots Z}]  && \Tr(\tX_A\otimes\I_{B\cdots 
    Z}\Phi_{AB\cdots Z})\notag\\
    &\subto && \Phi_{AB\cdots Z}\ge 0,~\Tr(\Phi_{AB\cdots Z})=1, 
    \label{eq:symExtOpt}\\
    &       && P_N^+\Phi_{AB\cdots Z}P_N^+=\Phi_{AB\cdots Z},\notag\\
    &       && \tL_A\otimes\id_{B\cdots Z}(\Phi_{AB\cdots Z})\notag
    =Y\otimes\Tr_A(\Phi_{AB\cdots Z}).
  \end{alignat}
  Furthermore, the SDP hierarchy is complete in the sense that 
  $\xi_{N+1}\le\xi_N$ and $\lim_{N\to+\infty}\xi_N=\xi$.
  \label{thm:symExtOpt}
\end{theorem}
In addition, any criterion for the full separability of $\Phi_{AB\cdots Z}$ can 
be added to the optimization in Eq.~\eqref{eq:symExtOpt}, which can give 
a better upper bound for the optimization in Eq.~\eqref{eq:SDPRank}.  For 
example, the PPT criterion, more precisely, PPT with respect to all 
bipartitions, can also be added as additional constraints, which can give 
better upper bounds $\xi_N^T$, i.e.,  $\xi\le\xi_{N+1}^T\le\xi_N^T\le\xi_N$ and 
$\lim_{N\to+\infty}\xi_N^T=\xi$.  Furthermore, it is sometimes convenient to 
denote the solution of the SDP by relaxing the rank constraint in 
Eq.~\eqref{eq:SDPRank} as $\xi_1$, then we have $\xi_2\le \xi_1$.

Let us estimate the complexity of the SDP hierarchy in 
Theorem~\ref{thm:symExtOpt}.  For the $N$-th level of the hierarchy, the 
dimension of the matrix reads $\dim(\cH^{\otimes N})=(nk)^N$, but it can be 
further reduced by taking advantage of the fact that $\Phi_{AB\cdots Z}$ is 
within the symmetric subspace, which has the dimension
\begin{equation}
  \dim\left(P_N^+\right)
  =\binom{nk+N-1}{N}=\binom{nk+N-1}{nk-1}.
  \label{eq:dimsym}
\end{equation}
By noting that $k\le n$, Eq.~\eqref{eq:dimsym} implies that, for fixed 
dimension $n$, the complexity of the SDP grows polynomially with the level of 
the hierarchy $N$, and for a fixed level of hierarchy $N$, the complexity of 
the SDP also grows polynomially with the dimension $n$. Similar results also 
hold when considering the PPT criterion, because the partial transpose of 
$\Phi_{AB\cdots Z}$ with respect to any bipartition is within the tensor 
product of two symmetric subspaces $P_k^+\otimes P_{N-k}^+$ for some $k$ 
\cite{Doherty.etal2004}.

\subsection{Optimization over real matrices}

We move on to consider the $\dF=\dR$ case, which is more important in classical 
information theory. One can easily verify that Theorem~\ref{thm:sepOpt} can be 
directly generalized to the $\dF=\dR$ case, if the decomposition in 
Eq.~\eqref{eq:symmPhi} satisfies that 
$\ket{\varphi_i}\bra{\varphi_i}\in\dR^{nk\times nk}$. The obvious way to 
guarantee this is to define the set of separable states over $\dR$. This will, 
however, make the known separability criteria developed in entanglement theory
not directly applicable.

Thus, we employ a different method. We still use the separability cone $\Sep$ 
with respect to the complex numbers, more precisely,
\begin{equation}
  \Phi_{AB}\in\Sep\cap\dR^{nk\times nk},
  \label{eq:SEPR}
\end{equation}
where $\Sep$ is still defined as in Eq.~\eqref{eq:SEP}.  
Equations~(\ref{eq:symmPhi},\,\ref{eq:SEPR}) are not sufficient for 
guaranteeing that $\ket{\varphi_i}\bra{\varphi_i}\in\dR^{nk\times nk}$
\footnote{An explicit counterexample is the (unnormalized) state 
$\Phi_{AB}=\I_{AB}+V_{AB}$.}, however, only a small modification is 
needed.  For pure states one has
$\ket{\varphi_i}\bra{\varphi_i}^T=\ket{\varphi_i^*}\bra{\varphi_i^*}$,
where $(\cdot)^T$ denotes the transpose and $\ket{(\cdot)^*}$ denotes complex 
conjugation with respect to a fixed basis.  Hence, a necessary condition for 
$\ket{\varphi_i}\bra{\varphi_i}\in\dR^{nk\times nk}$ is \begin{equation}
  \Phi_{AB}^{T_A}=\Phi_{AB},
  \label{eq:PTInv}
\end{equation}
where $(\cdot)^{T_A}$ denotes the partial transpose on party $A$. That is, the 
state $\Phi_{AB}$ is invariant under partial transposition.

Interestingly, due to the symmetry and separability of $\Phi_{AB}$, 
Eq.~\eqref{eq:PTInv} is also sufficient for guaranteeing that 
$\ket{\varphi_i}\bra{\varphi_i}\in\dR^{nk\times nk}$; see 
Appendix~\ref{app:sepOpt} for the proof. Hence, we arrive at the following 
theorem for rank-constrained optimization over real matrices.

\begin{theorem}
  For $\dF=\dR$, the rank-constrained SDP in Eq.~\eqref{eq:SDPRank} is 
  equivalent to the conic program
  \begin{equation}
    \begin{aligned}
      &\maxover[\Phi_{AB}]  && \Tr(\tX_A\otimes\I_B\Phi_{AB})\\
      &\subto && \Phi_{AB}\in\Sep,~\Tr(\Phi_{AB})=1,\\
      &       && V_{AB}\Phi_{AB}=\Phi_{AB},~\Phi_{AB}^{T_A}=\Phi_{AB},\\
      &       && \tL_A\otimes\id_B(\Phi_{AB})
		 =Y\otimes\Tr_A(\Phi_{AB}).
    \end{aligned}
    \label{eq:sepOptR}
  \end{equation}
  \label{thm:sepOptR}
\end{theorem}

Similarly to Theorem~\ref{thm:symExtOpt}, we can also construct a 
complete hierarchy with the multi-party extension method for the 
real case.
\begin{theorem}
  For $\dF=\dR$, let $\xi$ be the solution of the rank-constrained SDP in 
  Eq.~\eqref{eq:SDPRank}. Then, for any $N$, $\xi$ is upper bounded by the 
  solution $\xi_N$ of the SDP hierarchy
  \begin{alignat}{2}
    &\maxover[\Phi_{AB\cdots Z}]  && \Tr(\tX_A\otimes\I_{B\cdots 
    Z}\Phi_{AB\cdots Z})\notag\\
    &~\subto && \Phi_{AB\cdots Z}\ge 0,~\Tr(\Phi_{AB\cdots Z})=1,
    \label{eq:symExtOptR}\\
    &       && P_N^+\Phi_{AB\cdots Z}P_N^+=\Phi_{AB\cdots Z},
    ~\Phi_{AB\cdots Z}^{T_A}=\Phi_{AB\cdots Z},\notag\\
    &       && \tL_A\otimes\id_{B\cdots Z}(\Phi_{AB\cdots Z})
    =Y\otimes\Tr_A(\Phi_{AB\cdots Z}).\notag
  \end{alignat}
  Furthermore, the SDP hierarchy is complete, i.e., $\xi_{N+1}\le\xi_N$ and 
  $\lim_{N\to+\infty}\xi_N=\xi$.
  \label{thm:symExtOptR}
\end{theorem}

We emphasize that all variables involved in Eqs.~\eqref{eq:sepOptR} and 
\eqref{eq:symExtOptR} are taken as real matrices. In addition, due to the 
permutation symmetry induced by $P_N^+\Phi_{AB\cdots Z}P_N^+=\Phi_{AB\cdots 
Z}$, $\Phi_{AB\cdots Z}^{T_A}=\Phi_{AB\cdots Z}$ already ensures the 
partial-transpose-invariance with respect to all bipartitions. This also makes 
the PPT criterion as an additional separability condition redundant for the 
hierarchy in Eq.~\eqref{eq:symExtOptR}.

\subsection{Inherent symmetry for the hierarchy}\label{ssec:globalsym}

Before proceeding further, we briefly describe an inherent symmetry in 
Eqs.~(\ref{eq:sepOpt},\,\ref{eq:symExtOpt},\,%
\ref{eq:sepOptR},\,\ref{eq:symExtOptR}), which is particularly useful for the 
practical implementation. In a convex optimization problem, if a group action 
$G$ does not change the objective function and feasible region, then the 
variables can be assumed to be $G$-invariant.  Specifically, if the SDP, 
$\max_{\Phi\in\cS}\Tr(\Phi X)$, satisfies that $g\cS g^\dagger\subset\cS$ and 
$gXg^\dagger=X$ for all $g\in G$, then we can add an extra $G$-invariant 
constraint that $g\Phi g^\dagger=\Phi$ for all $g\in G$.

For the hierarchy in the complex case in Theorems~\ref{thm:sepOpt} and 
\ref{thm:symExtOpt}, regardless of the actual forms of $X$, $\Lambda$, and $Y$, 
there is an inherent $U^{\otimes N}$ symmetry on 
$\cH_{A_2}\otimes\cH_{B_2}\otimes\cdots\otimes\cH_{Z_2}=(\dC^k)^{\otimes N}$ 
for all $U\in \mathrm{SU}(k)$, i.e., on the $N$ auxiliary Hilbert spaces 
$\cH_2^{\otimes N}$ shown in Fig.~\ref{fig:Ncopy}. Hence, $\Phi_{AB\cdots Z}$ 
can be restricted to those satisfying
\begin{equation}
  (\I_n\otimes U)^{\otimes N}\Phi_{AB\cdots Z}
  (\I_n\otimes U^\dagger)^{\otimes N}
  =\Phi_{AB\cdots Z}.
  \label{eq:symmetry}
\end{equation}
This implies that $\Phi_{AB\cdots Z}$ is generated by the symmetric group 
$S_N$ in $\cH_{A_2}\otimes\cH_{B_2}\otimes\cdots\otimes\cH_{Z_2}=(\dC^k)^{\otimes N}$ 
by the Schur-Weyl duality \cite{Goodman.Wallach2009}.

We take the case $N=2$ as an example to illustrate this point. Under the 
restriction in Eq.~\eqref{eq:symmetry}, $\Phi_{AB}$ admits the form
\begin{equation}
  \Phi_{AB}=\Phi_I\otimes\I_{A_2B_2}+\Phi_V\otimes V_{A_2B_2},
\end{equation}
where $\Phi_I$ and $\Phi_V$ are operators on $\cH_{A_1B_1}$, and $\I_{A_2B_2}$ 
and $V_{A_2B_2}$ are the identity and swap operators on 
$\cH_{A_2}\otimes\cH_{B_2}=\dC^k\otimes\dC^k$, respectively. By taking 
advantage of the relations
\begin{equation}
  \begin{aligned}
    &V_{AB}=V_{A_1B_1}\otimes V_{A_2B_2},~
    V_{A_2B_2}^{T_{A_2}}=k\ket{\phi_k^+}\bra{\phi_k^+},\\
    &\Tr_{A_2}(\I_{A_2B_2})=k\I_{B_2},~
    \Tr_{A_2}(V_{A_2B_2})=\I_{B_2},\\
    &V_{A_2B_2}=P_{A_2B_2}^+-P_{A_2B_2}^-,~
    \I_{A_2B_2}=P_{A_2B_2}^++P_{A_2B_2}^-
  \end{aligned}
\end{equation}
where $\ket{\phi_k^+}=\frac{1}{\sqrt{k}}\sum_{\alpha=1}^k
\ket{\alpha}_{A_2}\ket{\alpha}_{B_2}$ is a maximally entangled state and 
$P_{A_2B_2}^+$ and $P_{A_2B_2}^+$ are projectors onto the symmetric and 
antisymmetric subspaces, respectively, $\xi_2^T$ can be simplified to
\begin{align}
  &\maxover[\Phi_I,\Phi_V]  && 
  \Tr\big[X_{A_1}\otimes\I_{B_1}(k^2\Phi_I+k\Phi_V)\big]\notag\\
  &\subto && \Phi_V=V_{A_1B_1}\Phi_I,~\Phi_I+\Phi_V\ge 0,\notag\\
  &       && \Phi_I-\Phi_V\ge 0,~\Phi_I^{T_{A_1}}
  +k\Phi_V^{T_{A_1}}\ge 0,\label{eq:s2T}\\
  &       && \Phi_I^{T_{A_1}}\ge 0,~k^2\Tr(\Phi_I)+k\Tr(\Phi_V)=1,\notag\\
  &       && \Lambda_{A_1}\otimes\id_{B_1}(k\Phi_I+\Phi_V)
  =Y\otimes\Tr_{A_1}(k\Phi_I+\Phi_V).\notag
\end{align}
A significant improvement in Eq.~\eqref{eq:s2T} is that the dimension of the 
variables is $\dC^{n^2\otimes n^2}$, which no longer depends on the rank $k$.

For the hierarchy in the real case in Theorems~\ref{thm:sepOptR} and 
\ref{thm:symExtOptR}, we consider the symmetry $Q^{\otimes N}$ for $Q\in 
\mathrm{O}(k)$, which would also simplify the structure of $\Phi_{AB\cdots Z}$ 
in $\cH_{A_2}\otimes\cH_{B_2}\otimes\cdots\otimes\cH_{Z_2}=(\dR^k)^{\otimes 
N}$.  The $\mathrm{O}(k)$ symmetry can reduce $\Phi_{AB\cdots Z}$ to the Brauer 
algebra $B_N(k)$ in 
$\cH_{A_2}\otimes\cH_{B_2}\otimes\cdots\otimes\cH_{Z_2}=(\dR^k)^{\otimes N}$ 
\cite{Goodman.Wallach2009}, which is more complicated than the $\mathrm{SU}(k)$ 
symmetry.

For the $N=2$ case, the Brauer algebra $B_2(k)$ is the linear span of 
$\{\I_{A_2B_2},V_{A_2B_2},k\ket{\phi_k^+}\bra{\phi_k^+}\}$, which gives the 
symmetrized $\Phi_{AB}$ of the form
\begin{equation}
  \hspace*{-1ex}
  \Phi_{AB}=\Phi_I\otimes\I_{A_2B_2}+\Phi_V\otimes V_{A_2B_2}+\Phi_\phi\otimes 
  k\ket{\phi_k^+}\bra{\phi_k^+},
  \hspace*{-0.6ex}
\end{equation}
where $\Phi_I$, $\Phi_V$, and $\Phi_\phi$ are operators on $\cH_{A_1B_1}$.
Correspondingly, $\xi_2$ can be simplified to
\begin{alignat}{2}
  & \maxover[\Phi_I,\Phi_V,\Phi_\phi] &&
  \Tr\big[X_{A_1}\otimes\I_{B_1}(k^2\Phi_I+k\Phi_V+k\Phi_\phi)\big]\notag\\
  &~~~\subto&&\Phi_V=V_{A_1B_1}\Phi_I,~\Phi_\phi=\Phi_V^{T_{A_1}},
  ~\Phi_I^{T_{A_1}}=\Phi_I,\notag\\
  &       && V_{A_1B_1}\Phi_\phi=\Phi_\phi,~\Phi_I+\Phi_V\ge 0,\notag\\
  &       && \Phi_I-\Phi_V\ge 0,~\Phi_I+\Phi_V+k\Phi_\phi\ge
	     0,\label{eq:s2Tr}\\
  &       && k^2\Tr(\Phi_I)+k\Tr(\Phi_V)+k\Tr(\Phi_\phi)=1,\notag\\
  &       && \Lambda_{A_1}\otimes\id_{B_1}(k\Phi_I+\Phi_V+\Phi_\phi)\notag\\
  &       && \qquad=Y\otimes\Tr_{A_1}(k\Phi_I+\Phi_V+\Phi_\phi).\notag
\end{alignat}

Curiously, in the SDPs in Eqs.~\eqref{eq:s2T}~and~\eqref{eq:s2Tr}, the rank 
constraint $k$ appears as a parameter that, in principle, can take on 
non-integer values. In Appendix~\ref{app:contrank}, we show that $k$ can 
indeed, in some sense, be considered a continuous rank, and that this is useful 
for handling numerical errors.

\section{Examples}

In this section, we show that our method can be widely used in quantum and 
classical information theory. As illustrations, we investigate the examples of 
the optimization over pure states and unitary channels, the characterization of 
faithful entanglement, and quantum contextuality as problems in quantum 
information theory. Concerning classical information theory, we study as 
examples the Max-Cut problem, pseudo-Boolean optimization, and the minimum 
dimension of the orthonormal representation of graphs.

\subsection{Optimization over pure quantum states and unitary channels}

A direct application of our method in quantum information theory is the optimization 
over pure states. For example, we consider the optimization problem from 
incomplete information
\begin{equation}
  \begin{aligned}
    &\maxover[\ket{\varphi}]  && \bra{\varphi}X\ket{\varphi}\\
    &\subto && \bra{\varphi}M_i\ket{\varphi}=m_i,
  \end{aligned}
  \label{eq:SDPPure}
\end{equation}
where the $M_i$ are the performed measurements and the $m_i$ are the corresponding 
measurement results. This can be viewed as a refined problem of compressed 
sensing tomography \cite{Gross.etal2010}, in which the feasibility problem is 
considered. The optimization in Eq.~\eqref{eq:SDPPure} is obviously 
a rank-constrained SDP,
\begin{equation}
  \begin{aligned}
    &\maxover[\rho]  && \Tr(X\rho)\\
    &\subto && \Tr(M_i\rho)=m_i,\,\Tr(\rho)=1,\\
    &       && \rho\ge0,~\rank(\rho)=1.
  \end{aligned}
  \label{eq:SDPPureRank}
\end{equation}
Thus, Theorem~\ref{thm:sepOpt} gives the equivalent conic program
\begin{alignat}{2}
    \notag
    &\maxover[\Phi_{AB}]~~&& \Tr(X_A\otimes\I_B\Phi_{AB})\\
    \label{eq:sepOptPure}
    &\subto && \Phi_{AB}\in\Sep,\,\Tr(\Phi_{AB})=1,\,
	       V_{AB}\Phi_{AB}=\Phi_{AB},\\
    \notag
    &       && \Tr_A(M_i\otimes\I_B\Phi_{AB})=m_i\Tr_A(\Phi_{AB}),
\end{alignat}
from which a complete SDP hierarchy can be constructed using 
Theorem~\ref{thm:symExtOpt}. Similarly, we can also consider the optimization 
over low-rank quantum states.

Because of the Choi-Jamio{\l}kowski duality \cite{Watrous2018}, the results in 
Eqs.~(\ref{eq:SDPPureRank}, \ref{eq:sepOptPure}) can also be used for the 
optimization over unitary (and low-Kraus-rank) channels.  As an example, we 
show that our method provides a complete characterization of the mixed-unitary 
channels, which was recently proved to be an NP-hard problem 
\cite{Lee.Watrous2020}.

A channel $\Lambda$ is called mixed-unitary if there exists a positive integer 
$m$, a probability distribution $(p_1,p_2,\dots,p_m)$, and unitary operators 
$U_1,U_2,\dots,U_m$ such that
\begin{equation}
  \Lambda(\rho)=\sum_{i=1}^m p_i U_i\rho U_i^\dagger.
  \label{eq:mixedUnitary}
\end{equation}
According to the Choi-Jamio{\l}kowski duality, a channel is mixed-unitary if and 
only if the corresponding Choi state defined as
\begin{equation}
  J(\Lambda)=\id_1\otimes\Lambda_2(\ket{\phi^+}\bra{\phi^+}),
  \label{eq:Choi}
\end{equation}
with $\ket{\phi^+}=\frac{1}{\sqrt{n}}\sum_{\alpha=1}^n\ket{\alpha\alpha}$, is 
a mixture of maximally entangled states, i.e.,
$J(\Lambda)=\sum_{i=1}^m p_i\ket{\phi_i}\bra{\phi_i}$,
where the $\ket{\phi_i}$ are maximally entangled states, i.e.,
$\Tr_1(\ket{\phi_i}\bra{\phi_i})=\I_n/n$.
Thus, characterizing the mixed-unitary channels is equivalent to characterizing 
the mixture of maximally entangled states,
\begin{equation}
  \cM=\conv\left\{\ket{\phi}\bra{\phi}\,\Big|\,
  \Tr_1(\ket{\phi}\bra{\phi})
  =\frac{\I_n}{n}\right\}.
  \label{eq:maxEntMixed}
\end{equation}
According to Eq.~\eqref{eq:S22S}, $\Lambda$ is mixed-unitary, i.e.,  
$J(\Lambda)\in\cM$, is equivalent to the feasibility problem
\begin{align}
  \notag
  &\findover  && \Phi_{AB}\in\Sep\\
  \label{eq:sepFeasMixedUnitary}
  &\subto && \Tr_B(\Phi_{AB})=J(\Lambda),~V_{AB}\Phi_{AB}=\Phi_{AB},\\
  \notag
  &       && \Tr_{A_1}\otimes\id_{A_2}\otimes\id_B(\Phi_{AB})
  =\frac{\I_n}{n}\otimes\Tr_A(\Phi_{AB}),\\
  \notag
  &       && \id_{A_1}\otimes\Tr_{A_2}\otimes\id_B(\Phi_{AB})
  =\frac{\I_n}{n}\otimes\Tr_A(\Phi_{AB}),
\end{align}
where the last constraint follows from $\Tr_2(\ket{\phi}\bra{\phi})=\I_n/n$ 
according to Eq.~\eqref{eq:maxEntMixed}. This constraint is redundant for 
Eq.~\eqref{eq:sepFeasMixedUnitary}, but it may help when the SDP relaxations 
are considered.

A further application comes from entanglement theory. Following 
Ref.~\cite{Guehne.etal2021}, the optimization over $\cM$ also provides 
a complete characterization of faithful entanglement 
\cite{Weilenmann.etal2020}, i.e., the entangled states that are detectable by 
fidelity-based witnesses. In Ref.~\cite{Guehne.etal2021}, the authors prove 
that a state $\rho\in\dC^n\otimes\dC^n$ is faithful if and only if 
$\xi:=\max_{\sigma\in\cM}\Tr(\sigma\rho)>1/n$.  According to 
Theorem~\ref{thm:sepOpt}, the solution $\xi$ also equals the conic program
\begin{align}
  \notag
  &\maxover  && \Tr(\rho_A\otimes\I_B\Phi_{AB})\\
  \label{eq:unfaithful}
  &\subto && \Phi_{AB}\in\Sep,~V_{AB}\Phi_{AB}=\Phi_{AB},\\
  \notag
  &       && \Tr_{A_1}\otimes\id_{A_2}\otimes\id_B(\Phi_{AB})
  =\frac{\I_n}{n}\otimes\Tr_A(\Phi_{AB}),\\
  \notag
  &       && \id_{A_1}\otimes\Tr_{A_2}\otimes\id_B(\Phi_{AB})
  =\frac{\I_n}{n}\otimes\Tr_A(\Phi_{AB}),
\end{align}
where $\cH_A=\cH_B=\dC^n\otimes\dC^n$. By taking advantage of the complete 
hierarchy, if for some $N$ there is $\xi_N\le 1/n$ or $\xi_N^T\le 1/n$, then 
$\rho$ is unfaithful. In practice, it is already enough to take $\xi_2^T$ for 
verifying the unfaithfulness of some states that are not detectable by any of
the known methods \cite{Guehne.etal2021,Weilenmann.etal2020}. An explicit 
example for $n=4$ is
\begin{equation}
  \rho=\frac{p}{16}\I_4\otimes\I_4+
  \frac{1-p}{2}\ket{x}\bra{x}+\frac{1-p}{2}\ket{y}\bra{y},
\end{equation}
where
\begin{equation}
  \begin{aligned}
    \ket{x}&=\frac{1}{\sqrt{10}}\sum_{\alpha=1}^4
    \sqrt{\alpha}\ket{\alpha\alpha},\\
    \ket{y}&=\frac{1}{\sqrt{10}}\sum_{\alpha=1}^4
    \beta_4^\alpha\sqrt{5-\alpha}\ket{\alpha\alpha},
  \end{aligned}
\end{equation}
with $p=23/40$ and $\beta_4=(1+\mi)/\sqrt{2}$. For this state, the SDP 
relaxation of Eq.~\eqref{eq:unfaithful} gives the upper bound 
$\xi^T_2=0.24888<1/4$, which matches with the lower bound from gradient search 
and is strictly better than the best known upper bound $\xi_1=0.25063>1/4$ from 
Ref.~\cite{Guehne.etal2021}.

\subsection{Gram matrix and orthonormal representation}

Let $\ket{a_i}\in \dF^k$ ($\dF=\dC$ or $\dF=\dR$) for $i=1,2,\dots,n$ be 
a sequence of vectors, then the Gram matrix defined as 
$\Gamma=[\braket{a_i}{a_j}]_{i,j=1}^n$ satisfies $\Gamma\ge 0$ and 
$\rank(\Gamma)\le k$. The converse is also true in the sense that if 
an $n\times n$ matrix in $\dF^{n\times n}$ satisfies $\Gamma\ge 0$ 
and $\rank(\Gamma)\le k$, then there exist $\ket{a_i}\in \dF^k$ for $i=1,2,\dots,n$, 
such that 
$\Gamma_{ij}=\braket{a_i}{a_j}$ \cite{Lovasz2019}.  This correspondence can 
trigger many applications of the rank-constrained optimization.  For example, 
it can be used to bound the minimum dimension of the orthonormal representation 
of graphs.

In graph theory, a graph is a pictorial representation of a set of objects 
(vertices) where some pairs of objects are connected by links (edges).  
Formally, a graph $G$ is denoted by a pair $(V,E)$, where $V$ is the set of 
vertices, and $E$ is the set of edges that are paired vertices. For a graph 
$G=(V,E)$, an orthonormal representation is a set of normalized vectors 
$\big\{\ket{a_i}\in \dF^k\bmid i\in V\big\}$, such that $\braket{a_i}{a_j}=0$ 
if $\{i,j\}\notin E$ \cite{Lovasz2019}. The minimum dimension problem is to 
find the smallest number $k$ such that an orthonormal representation exists.  
This is not only an important quantity in classical information theory 
\cite{Lovasz2019}, but also widely used in quantum information theory. For 
example, it is a crucial quantity in quantum contextuality theory 
\cite{Cabello.etal2014,Ramanathan.Horodecki2014}, and can be directly used for 
contextuality-based dimension witness \cite{Ray.etal2021}.  Note that in 
quantum contextuality, the definition of orthonormal representations is 
slightly different, where the adjacent instead of the nonadjacent vertices are 
required to be orthogonal to each other, i.e., $\braket{a_i}{a_j}=0$ if 
$\{i,j\}\in E$.  In the following, we use the standard definition in graph 
theory.  All results can be trivially adapted to the alternative definition by 
considering the complement graph.

\begin{figure}
  \centering
  \includegraphics[width=.3\textwidth]{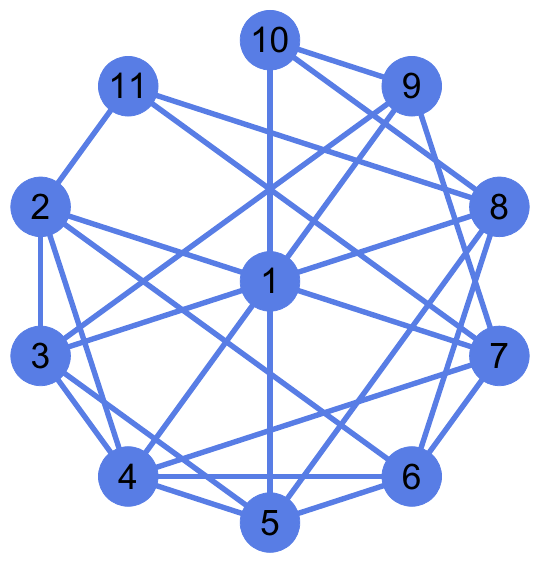}
  \caption{%
    For this $11$-vertex graph, one obtains that $\vartheta(G)=4$ \cite{Note2} 
    and hence, a lower bound of $4$ for the minimal dimension.  In contrast, 
    our PPT relaxation of Eq.~\eqref{eq:MINOR} can already exclude both real 
    and complex orthonormal representations in dimension $4$.
  }
  \label{fig:G11}
\end{figure}
\footnotetext{The numerical error can be shown to be smaller than $10^{-100}$ 
using the standard primal and dual problem of the Lov\'asz 
$\vartheta$-function's SDP characterization \cite{Lovasz1979} and the 
arbitrary-precision SDP solver SDPA-GMP \cite{Nakata2010}}

By taking advantage of the Gram matrix, the problem of the minimum dimension of 
the orthonormal representation \cite{Lovasz2019} can be expressed as
\begin{equation}
  \begin{aligned}
    &\minover[\Gamma]  && k\\
    &\subto && \Delta(\Gamma)=\I_n,~\Gamma_{ij}=0~\FA \{i,j\}\notin E,\\
    &       && \Gamma\ge0,~\rank(\Gamma)\le k,
  \end{aligned}
  \label{eq:OR}
\end{equation}
where $G=(V,E)$ is a graph with $\abs{V}=n$ vertices, $E$ is the set of edges, 
$\Delta(\cdot)$ denotes the map of eliminating all off-diagonal elements of 
a matrix (completely dephasing map), and $\Gamma\in\dR^{n\times n}$ or 
$\Gamma\in\dC^{n\times n}$ corresponds to the real or complex representation.  
Let $W$ be the adjacency matrix of $G$, i.e., $W_{ij}=1$ if $\{i,j\}\in E$ and 
$W_{ij}=0$ otherwise, then the first two constraints in Eq.~\eqref{eq:OR} can 
also be written as $(1-W_{ij})\Gamma_{ij}=\delta_{ij}$, i.e.,
\begin{equation}
  \Lambda(\Gamma):=(\E_n-W)\odot\Gamma=\I_n,
\end{equation}
where $\E_n$ is the $n\times n$ matrix with all elements being one and $[X\odot 
Y]_{ij}=X_{ij}Y_{ij}$ is the Hadamard product of matrices.  Then, the existence 
of a $k$-dimensional orthonormal representation is equivalent to the 
feasibility problem
\begin{align}
  \notag
  &\findover  && \Phi_{AB}\\
  \label{eq:MINOR}
  &\subto && \Phi_{AB}\in\Sep,~\Tr(\Phi_{AB})=n,\\
  \notag
  &       && V_{AB}\Phi_{AB}=\Phi_{AB},
  ~\left(\Phi_{AB}^{T_A}=\Phi_{AB}\right),\\
  \notag
  &       && \tL_A\otimes\id_B(\Phi_{AB})
  =\frac{1}{n}\I_n\otimes\Tr_A(\Phi_{AB}),
\end{align}
where $\cH_A=\cH_B=\cH_1\otimes\cH_2=\dC^n\otimes\dC^k~(\dR^n\otimes\dR^k)$, 
$\tL(\cdot)=\Lambda[\Tr_2(\cdot)]$, and the extra constraint 
$\Phi_{AB}^{T_A}=\Phi_{AB}$ is for the case that $\dF=\dR$ only.
Note that the inherent symmetry presented in Sec.~\ref{ssec:globalsym} can be 
used for simplifying the SDP relaxations.

The Lov\'asz $\vartheta$-function defined by
\begin{equation}
  \vartheta(G) = \min_{\{\ket{a_i}\}_{i \in V},\ket{c}} \max_{i \in V} 
  \frac{1}{|\braket{c}{a_i}|^2},
  \label{eq:Lovasz}
\end{equation}
where the $\ket{a_i}$ form an orthonormal representation and $\ket{c}$ is 
a unit vector, is probably the best-known way to obtain a lower bound on the 
minimal dimension of orthonormal representations. We note that the value of the 
Lov\'asz $\vartheta$-function is independent of whether the orthonormal 
representation is real or complex \cite{Lovasz2019}. For any $k$-dimensional 
orthonormal representation $\ket{a_i}$, ${\ket{a_i}\otimes\ket{a_i^*}}$ also 
form an orthonormal representation and, with $\ket{c} = \frac{1}{\sqrt{k}} 
\sum_{\alpha=1}^k\ket{\alpha}\otimes\ket{\alpha}$, the bound $k\ge\vartheta(G)$ 
is readily obtained from Eq.~\eqref{eq:Lovasz}. Our method can provide a better 
bound even for small graphs; see Fig.~\ref{fig:G11}.

\subsection{Max-Cut problem}

\begin{figure}[t]
  \centering
  \includegraphics[width=.3\textwidth]{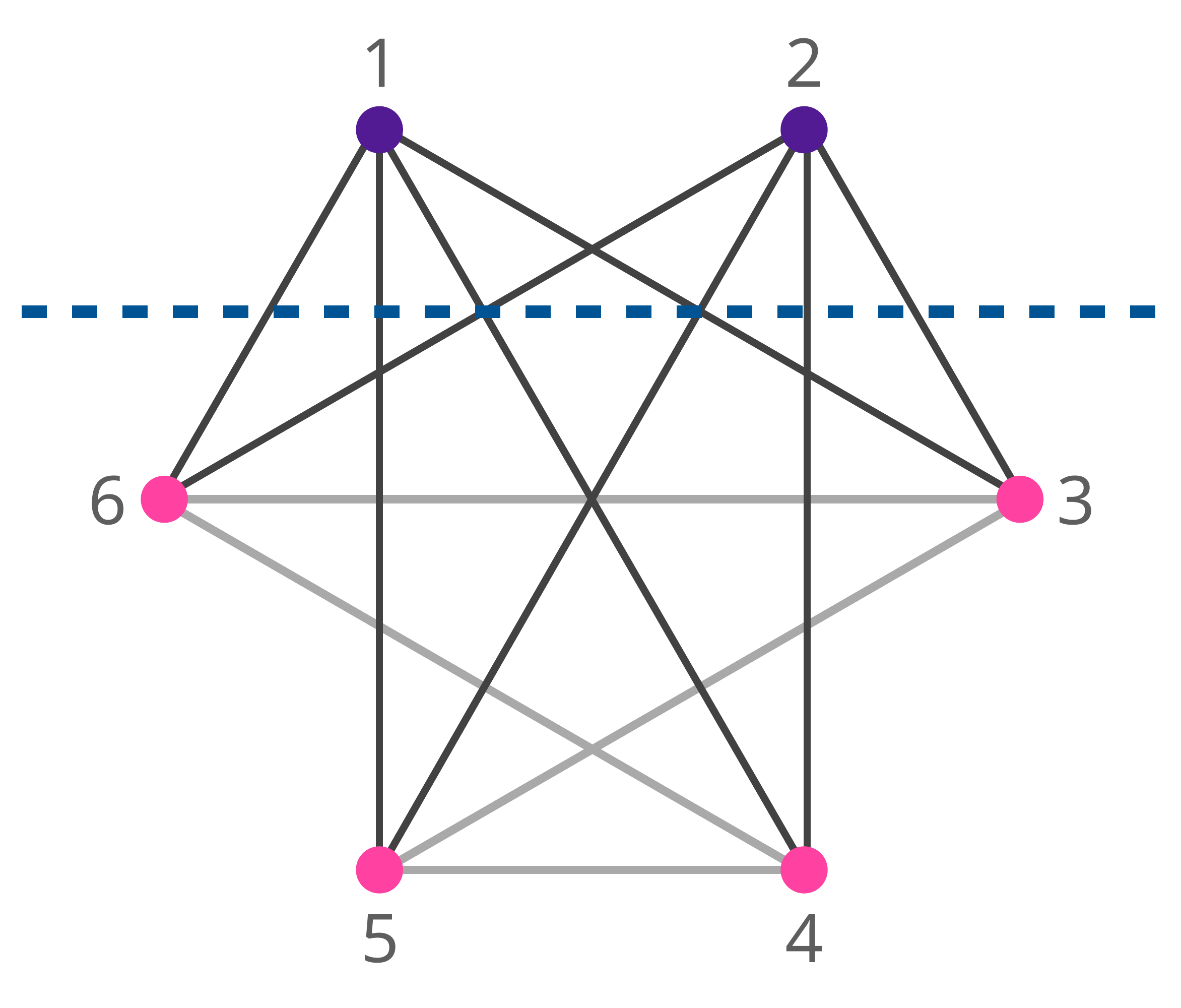}
  \caption{The Max-Cut problem for a graph is to find a cut, i.e., 
    a bipartition, such that the number of edges that cross the cut is 
    maximized. For the graph shown in the figure, the Max-Cut is $8$ (achieved 
    by the cut $1,2$ versus $3,4,5,6$), which matches our SDP relaxation 
  $\xi_2=8$, while the usual Goemans-Williamson method gives only the upper bound 
$\xi_1=9$.}
  \label{fig:maxcut}
\end{figure}

The Max-Cut problem is among the best-known rank-constrained optimization 
problems \cite{Goemans.Williamson1995} and also draws a lot of interest in 
quantum computing \cite{Hamerly.etal2019,Crooks2018}. Given a graph $G=(V,E)$, 
the Max-Cut problem is to find a cut, i.e., a bipartition of the vertices 
$(S,S^c)$, where $S^c=V\setminus S$, that maximizes the number of edges between 
$S$ and $S^c$; see Fig.~\ref{fig:maxcut}. A significant breakthrough for the 
Max-Cut problem was the work by Goemans  and Williamson 
\cite{Goemans.Williamson1995}, in which they showed that the Max-Cut problem 
can be written as the rank-constrained optimization
\begin{equation}
  \begin{aligned}
    &\maxover[\rho]  && \frac{1}{4}\Tr[W(\E_n-\rho)]\\
    &\subto && \Delta(\rho)=\I_n,\,\rho\ge0,\,\rank(\rho)=1,
  \end{aligned}
  \label{eq:MAXCUT}
\end{equation}
where $n=\abs{V}$, $\rho\in\dR^{n\times n}$, $\E_n$ is the $n\times n$ matrix 
with all elements being one, and $W$ is the adjacency matrix of $G$. To see why 
the Max-Cut problem is equivalent to Eq.~\eqref{eq:MAXCUT}, we denote a cut 
with the binary vector $\vect{x}\in\{-1,1\}^n$ such that $x_i=1$ if $i\in S$ 
and $x_i=-1$ if $i\in S^c$ and let $\rho=\vect{x}\vect{x}^T$, then the number 
of edges between $S$ and $S^c$ is $\frac{1}{4}\sum_{(i,j)\in E}(1-x_ix_j)$, 
which is equal to the objective function in Eq.~\eqref{eq:MAXCUT}. Furthermore, 
the set of all cuts $\rho=\vect{x}\vect{x}^T$ can be fully characterized by the 
constraints in Eq.~\eqref{eq:MAXCUT}.
The idea of the Goemans-Williamson approximation is to remove the rank constraint 
in Eq.~\eqref{eq:MAXCUT} and solve the resulting SDP relaxation, which gives an 
upper bound $\xi_1$ for the Max-Cut problem.

In the following, we show how our method can give a better estimate than the 
Goemans-Williamson approximation. By noting that we can add a redundant 
constraint $\Tr(\rho)=n$, Theorem~\ref{thm:sepOptR} implies that the Max-Cut 
problem is equivalent to the conic program
\begin{equation}
  \begin{aligned}
    &\maxover[\Phi_{AB}]  && 
    \frac{1}{4}\Tr[W\E_n]-\frac{1}{4}\Tr[W_A\otimes\I_B\Phi_{AB}]\\
    &\subto && \Phi_{AB}\in\Sep,~\Tr(\Phi_{AB})=n,\\
    &       && V_{AB}\Phi_{AB}=\Phi_{AB},~\Phi_{AB}^{T_A}=\Phi_{AB},\\
    &       && \Delta_A\otimes\id_B(\Phi_{AB})
	       =\frac{1}{n}\I_n\otimes\Tr_A(\Phi_{AB}),
  \end{aligned}
  \label{eq:MAXCUTSEP}
\end{equation}
where $\cH_A=\cH_B=\dR^n$. Correspondingly, a complete hierarchy of SDPs can be 
constructed from Theorem~\ref{thm:symExtOptR}.

We have tested the SDP relaxation $\xi_2$ (replacing $\Phi_{AB}\in\Sep$ with 
$\Phi_{AB}\ge 0$) with some random graphs (randomly generated adjacency 
matrices). Let us discuss the largest two graphs that we have tested; the 
details for these two graphs are given in the Supplemental Material. For 
a $64$-vertex graph with $419$ edges, the Goemans-Williamson method gives the 
upper bound $\lfloor\xi_1\rfloor=299$; instead $\xi_2=287$.  For the 
$72$-vertex graph with $475$ edges, the Goemans-Williamson method gives the 
upper bound $\lfloor\xi_1\rfloor=335$; instead $\xi_2=321$.  Furthermore, the 
optimal $\Phi_{AB}$ also shows that the upper bound $\xi_2$ in these two cases 
are achievable.  Hence, $\xi_2$ gives exactly the solution to the Max-Cut 
problem in these examples
\footnote{To obtain $\xi_2$, we used the SCS solver (unparallel version) on 
  a single cluster node.  The solver took $1.33\times 10^5$ and $3.67\times 
  10^5$ seconds for solving the corresponding SDPs of the $64$-vertex and 
  $72$-vertex graphs, respectively.
}.
Actually, for all the graphs that we have tested, $\xi_2$ already gives the 
exact solution of the Max-Cut problem.  Finally, we would like to mention that, 
although our method gives a much better bound, it is more costly than the 
Goemans-Williamson method. For example, the size of the matrix grows 
quadratically on the number of vertices for $\xi_2$, compared to only growing 
linearly for the Goemans-Williamson method.

\subsection{Pseudo-Boolean optimization}

Similar to the Max-Cut problem, we can apply the method to general 
optimization of a real-valued function over Boolean variables. These so-called 
pseudo-Boolean optimization problems  find wide applications in, for 
example, statistical mechanics, computer science, discrete mathematics, and 
economics (see Ref.~\cite{Boros.Hammer2002} and references therein). As 
a demonstration, we consider the quadratic pseudo-Boolean optimization
\begin{equation}
  \begin{aligned}
    &\maxover[\vect{x}]  && \vect{x}^TQ\vect{x}+\vect{c}^T\vect{x}\\
    &\subto && x_i=\pm 1,
  \end{aligned}
  \label{eq:QSBO}
\end{equation}
where $Q\in\dR^{(n-1)\times(n-1)}$, $\vect{c}\in\dR^{n-1}$, and 
$\vect{x}^T=[x_1,x_2,\dots,x_{n-1}]$; higher-order cases can be obtained by 
reducing to quadratic forms \cite{Boros.Hammer2002} or applying the results in 
Sec.~\ref{ssec:quadratic}. Notably, performing quadratic pseudo-Boolean 
optimization problems with noisy intermediate-scale quantum computers has drawn 
a lot of research interest 
\cite{Boixo.etal2014,Lucas2014,Farhi.etal2014,Farhi.Harrow2019}.  So, the 
following method may be used for characterizing benchmarks of such devices.

The quadratic pseudo-Boolean optimization problem can also be written as 
a rank-constrained optimization.  The basic idea is to write $\rho$ as an 
$n\times n$ matrix
\begin{equation}
  \rho=
  \begin{bmatrix}
    \vect{x}\vect{x}^T & \vect{x}\\
    \vect{x}^T & 1
  \end{bmatrix}
  =
  \begin{bmatrix}
    \vect{x}\\ 1
  \end{bmatrix}
  \begin{bmatrix}
    \vect{x}^T & 1
  \end{bmatrix}.
\end{equation}
Furthermore, we define $L$ as
\begin{equation}
  L=
  \begin{bmatrix}
    Q & \frac{1}{2}\vect{c}\\
    \frac{1}{2}\vect{c}^T & 0
  \end{bmatrix}.
\end{equation}
Then the optimization problem in Eq.~\eqref{eq:QSBO} can be written as the 
rank-constrained SDP
\begin{equation}
  \begin{aligned}
    &\maxover[\rho]  && \Tr(L\rho)\\
    &\subto && \Delta(\rho)=\I_n,~\rho\ge0,~\rank(\rho)=1,
  \end{aligned}
  \label{eq:LeastSquareSDP}
\end{equation}
which is of a similar form as in Eq.~\eqref{eq:MAXCUT}. By 
Theorem~\ref{thm:sepOptR}, the quadratic pseudo-Boolean optimization problem is 
equivalent to the conic program in Eq.~\eqref{eq:MAXCUTSEP} with the objective 
function replaced by $\Tr(L_A\otimes\I_B\Phi_{AB})$.

To illustrate the performance of our method, we consider the Boolean least 
squares optimization, i.e.,
\begin{equation}
  \begin{aligned}
    &\minover[\vect{x}]  && \norm{A\vect{x}-\vect{b}}_2^2\\
    &\subto && x_i=\pm 1.
  \end{aligned}
  \label{eq:BLS}
\end{equation}
We have tested our SDP relaxation $\xi_2$, compared to the widely-used SDP 
relaxation $\xi_1$ \cite{Nesterov1998}, for $1000$ random matrices 
$A\in\dR^{40\times30}$ and vectors $\vect{b}\in\dR^{40}$ with elements 
independently normally distributed. For this size, the optimal value $\xi$ can 
still be obtained by brute force. In most cases, the optimum is reached by 
$\xi_2$ while there is a significant gap between the optimal value $\xi$ and 
$\xi_1$. More precisely, for the $1000$ random samples, the average ratio 
$\mean{\xi_2/\xi}=99.93\%$, in contrast to $\mean{\xi_1/\xi}=49.32\%$. Note 
that as the minimization is considered in Eq.~\eqref{eq:BLS}, $\xi_N$ provide 
lower bounds for $\xi$ instead of upper bounds.

In passing, we note that the sum-of-square (SOS) hierarchy 
\cite{Lasserre2001,Parrilo2000} can also be used for some examples in this 
paper, such as the pseudo-Boolean optimization 
\cite{Lovasz.Schrijver1991,Fleming.etal2019,Erdogdu.etal2017,Bandeira.Kunisky2019}.  
However, our method in its generality is not subsumed by the SOS hierarchy.

On the one hand, there are two relatively separated steps in our approach to 
the rank-constrained optimization. The first step consists of mapping 
a rank-constrained problem to an entanglement problem. It is this mapping that 
allows the symmetric extension to be applied in the second step. The 
implication of this mapping is actually broader: if new optimization algorithms 
over separable states are proposed, either classical algorithms or quantum 
algorithms, they can be directly used for rank-constrained optimization 
problems. In fact, exploring the full application of entanglement theory with 
many methods other than the symmetric extension, such as other widely used 
method for entanglement witness, is still yet to be exploited. On the other 
hand, our approach also points towards a formulation of an optimization problem 
in a more physical language. In the SOS hierarchy, the variables are treated as 
individual scalars. Instead, in our method they are treated as a single vector 
in a Hilbert space. This is important from a physicist's point of view, because 
this treatment makes it easier for a physicist to study the global properties 
of the physical system. One example is the utilization of not only the discrete 
symmetries \cite{Gatermann.Parrilo2004} but also the continuous symmetries 
\cite{Yu.etal2021}.

\section{More general results on rank-constrained 
optimization}\label{sec:general}

In this section, we consider extensions of the problem in
Eq.~(\ref{eq:SDPRank}) and general cases of rank-constrained 
optimization. For simplicity, we only consider the optimization over complex 
matrices. All results can be similarly applied to the optimization 
over real matrices by adding the partial-transpose-invariance 
constraint $\Phi_{AB}^{T_A}=\Phi_{AB}$ or $\Phi_{AB\cdots 
Z}^{T_A}=\Phi_{AB\cdots Z}$.

\subsection{Inequality constraints}\label{ssec:ineqcons}

We start from the rank-constrained SDP with inequality constraints
\begin{equation}
  \begin{aligned}
    &\maxover[\rho]  && \Tr(X\rho)\\
    &\subto && \Lambda(\rho)\le Y,~\Tr(\rho)=1,\\
    &       && \rho\ge0,~\rank(\rho)\le k,
  \end{aligned}
  \label{eq:SDPRankIneq}
\end{equation}
where $\Lambda$ is a Hermiticity-preserving map \cite{Watrous2018}. Similar to 
Eqs.~(\ref{eq:feasibleRegion},\,\ref{eq:defP}), we can still define the 
feasible region $\cF$ in $\dC^{n\times n}$ and its purification $\cP$ in 
$\dC^{nk\times nk}$ as
\begin{align}
  \hspace*{-1ex}
  \cF&=\big\{\rho\bmid\Lambda(\rho)\le Y,\Tr(\rho)=1,
  \rho\ge0,\rank(\rho)\le k\big\},
  \hspace*{-.6ex}\\
  \hspace*{-1ex}
  \cP&=\big\{\ket{\varphi}\bra{\varphi}\bmid
    \tL(\ket{\varphi}\bra{\varphi})\le
  Y,~\braket{\varphi}{\varphi}=1\big\},
  \label{eq:PIneq}
\end{align}
where $\tL(\cdot)=\Lambda[\Tr_2(\cdot)]$. Again, we denote the solution of 
Eq.~\eqref{eq:SDPRankIneq} as $\xi$. In this case, the proof of Theorem~\ref{thm:sepOpt} 
does not work, because although the constraints
\begin{equation}
  \begin{aligned}
    &\Phi_{AB}\in\Sep,~\Tr(\Phi_{AB})=1,~V_{AB}\Phi_{AB}=\Phi_{AB},\\
    &\tL\otimes\id_B(\Phi_{AB})\le Y\otimes\Tr_A(\Phi_{AB})
  \end{aligned}
  \label{eq:consIneq}
\end{equation}
still provide a necessary condition for $\Tr_B(\Phi_{AB})\in\cS:=\conv(\cP)$, 
they are no longer sufficient. This is because, contrary to the equality case, 
the pure states in the decomposition of $\Phi_{AB}$ can no longer be guaranteed 
to be in $\cP$ for the inequality case, and hence the proof from 
Appendix~\ref{app:sepOpt} does not work in this case. However, the complete 
hierarchy analogously to Eq.~\eqref{eq:symExtOpt} still provides the exact 
solution $\xi$.
\begin{theorem}
  For $\dF=\dC$, let $\xi$ be the solution of the rank-constrained SDP in 
  Eq.~\eqref{eq:SDPRankIneq}. Then, for any $N$, $\xi$ is upper bounded by the 
  solution $\xi_N$ of the SDP hierarchy
  \begin{alignat}{2}
    &\maxover[\Phi_{AB\cdots Z}]  && \Tr(\tX_A\otimes\I_{B\cdots Z}\Phi_{AB\cdots 
    Z})\notag\\
    &~\subto && \Phi_{AB\cdots Z}\ge 0,~\Tr(\Phi_{AB\cdots Z})=1,
    \label{eq:symExtOptIneq}\\
    &       && P_N^+\Phi_{AB\cdots Z}P_N^+=\Phi_{AB\cdots Z},
    \notag\\
    &       && \tL\otimes\id_{B\cdots Z}(\Phi_{AB\cdots Z}) \le 
    Y\otimes\Tr_A(\Phi_{AB\cdots Z}).\notag
  \end{alignat}
   Furthermore, the SDP hierarchy is complete, 
  i.e., $\xi_{N+1}\le\xi_N$ and $\lim_{N\to+\infty}\xi_N=\xi$.
  \label{thm:sepOptIneq}
\end{theorem}

Similarly, any criterion for the full separability of $\Phi_{AB\cdots Z}$ or 
the unnormalized state $Y\otimes\Tr_A(\Phi_{AB\cdots Z})-\tL\otimes\id_{B\cdots 
Z}(\Phi_{AB\cdots Z})$, such as the PPT criterion, can be added to the 
optimization in Eq.~\eqref{eq:symExtOptIneq}, which can give a better upper 
bound for the optimization in Eq.~\eqref{eq:SDPRankIneq}.

For simplicity, we only present the intuition of the proof of 
Theorem~\ref{thm:sepOptIneq} here; see Appendix~\ref{app:ineqcons} for 
a rigorous proof.
The property $\xi_{N+1}\le\xi_N$ follows from the hierarchical  property that 
if $\Phi_{AB\cdots ZZ'}$ is within the feasible region of level $N+1$, then 
$\Phi_{AB\cdots Z}=\Tr_{Z'}(\Phi_{AB\cdots ZZ'})$ is within the feasible region 
of level $N$.

For the convergence property, we consider a separable variant of the 
optimization in Eq.~\eqref{eq:symExtOptIneq} by replacing $\Phi_{AB\cdots Z}\ge 
0$ with $\Phi_{AB\cdots Z}\in\Sep$, and denote the corresponding solutions as 
$\widetilde\xi_N$, i.e., add a tilde to distinguish the solution with the 
separability constraint from the original $\xi_N$. Then, the quantum de Finetti 
theorem \cite{Christandl.etal2007,Caves.etal2002} implies that
\begin{equation}
  \lim_{N\to+\infty}\widetilde\xi_N=\lim_{N\to+\infty}\xi_N.
  \label{eq:tildexi}
\end{equation}
Now, we assume that the $\widetilde\xi_N$ are achieved by the separable states
\begin{equation}
  \widetilde\Phi_{AB\cdots Z}=\int_\psi f_N(\psi)
  \ket{\psi}\bra{\psi}^{\otimes N}\md\psi,
  \label{eq:symmPhiFin}
\end{equation}
where the $f_N(\psi)\md\psi$ are $N$-dependent probability distributions, and 
$\md\psi$ denotes the normalized uniform distribution. As the set of 
probability distributions on a compact set is also compact in the weak topology 
\cite{Bogachev2007}, we can take $f_\infty(\psi)\md\psi$ as a limit point of 
$f_N(\psi)\md\psi$. Thus, we get an $N$-independent probability distribution 
$f_\infty(\psi)\md\psi$. Let
\begin{equation}
  \widetilde\Phi^\infty_{AB\cdots Z}=\int_\psi f_\infty(\psi)
  \ket{\psi}\bra{\psi}^{\otimes N}\md\psi,
  \label{eq:symmPhin}
\end{equation}
which satisfies all the constraints in Eq.~\eqref{eq:symExtOptIneq} for 
arbitrary $N$ by the hierarchical property, and moreover
\begin{equation}
  \lim_{N\to+\infty}\widetilde\xi_N
  =\lim_{N\to+\infty}\Tr\big(\tX_A\widetilde\Phi_A^N\big)
  =\Tr\big(\tX_A\widetilde\Phi_A^\infty\big),
  \label{eq:deFinettiConvegence}
\end{equation}
where
\begin{align}
  \widetilde\Phi_A^N&=\Tr_{B\cdots Z}\big(\widetilde\Phi_{AB\cdots Z}\big)
  =\int_\psi f_N(\psi)\ket{\psi}\bra{\psi}\md\psi,\\
  \widetilde\Phi_A^\infty&=\Tr_{B\cdots Z}
  \big(\widetilde\Phi^\infty_{AB\cdots Z}\big)
  =\int_\psi f_\infty(\psi)\ket{\psi}\bra{\psi}\md\psi.
\end{align}
By Eq.~\eqref{eq:tildexi}, to prove that $\lim_{N\to+\infty}\xi_N=\xi$, we only 
need to show that $\widetilde\Phi_A^\infty\in\conv(\cP)$.  To this end, it is 
sufficient to show that $Y_\varphi:=\tL(\ket{\varphi}\bra{\varphi})\le Y$ 
whenever $f_\infty(\varphi)\ne 0$.  By substituting Eq.~\eqref{eq:symmPhin} 
into the last constraint in Eq.~\eqref{eq:symExtOptIneq}, we get that for 
arbitrary $N$
\begin{equation}
  \int_\psi f_\infty(\psi)(Y-Y_\psi)\otimes\ket{\psi}
  \bra{\psi}^{\otimes N}\md\psi\ge 0,
  \label{eq:ineqExt}
\end{equation}
which implies that
\begin{equation}
  \frac{\int_\psi 
  f_\infty(\psi)(Y-Y_\psi)\abs{\braket{\varphi}{\psi}}^{2N}\md\psi}
  {\int_\psi\abs{\braket{\varphi}{\psi}}^{2N}\md\psi}
  \ge 0
  \label{eq:ineqInf}
\end{equation}
for any $\ket{\varphi}$ and $N$. Note that for the complement of any 
$\varepsilon$-ball $B^c_\varphi(\varepsilon):=\big\{\ket{\psi}\bra{\psi}\bmid 
\abs{\braket{\varphi}{\psi}}^2\le 1-\varepsilon\big\}$
with $\varepsilon>0$, we have
\begin{equation}
  \lim_{N\to+\infty}\frac{
    \int_{B^c_\varphi(\varepsilon)}
  \abs{\braket{\varphi}{\psi}}^{2N}\md\psi}{
    \int_\psi
  \abs{\braket{\varphi}{\psi}}^{2N}\md\psi}=0,
\end{equation}
because the numerator decreases exponentially to zero, but the denominator 
$\int_\psi\abs{\braket{\varphi}{\psi}}^{2N}\md\psi=1/\dim(P_N^+)$ decreases 
polynomially according to Eq.~\eqref{eq:dimsym} and the relation
$\int_\psi\ket{\psi}\bra{\psi}^{\otimes N}\md\psi=P_N^+/\dim(P_N^+)$ 
\cite{Watrous2018}. Hence,
\begin{equation}
  \lim_{N\to+\infty}\frac{\abs{\braket{\varphi}{\psi}}^{2N}}
  {\int_\psi\abs{\braket{\varphi}{\psi}}^{2N}\md\psi}
  =\delta(\psi-\varphi),
  \label{eq:deltameasure}
\end{equation}
where $\delta(\cdot)$ is the Dirac-delta function. Then, in the limit 
$N\to+\infty$, Eq.~\eqref{eq:ineqInf} gives that
\begin{equation}
  \int_\psi f_\infty(\psi)(Y-Y_\psi)\delta(\psi-\varphi)\md\psi
  =f_\infty(\varphi)(Y-Y_\varphi)\ge 0,
  \label{eq:ineqCons}
\end{equation}
and hence, $Y_\varphi\le Y$ when $f_\infty(\varphi)\ne 0$.

\subsection{Non-positive-semidefinite variables}\label{ssec:nonpsdvar}

Second, we study the rank-constrained optimization for 
non-positive-semidefinite and even non-square matrices. Consider the 
rank-constrained optimization
\begin{equation}
  \begin{aligned}
    &\maxover[\omega]  && \Tr(X\omega)+\Tr(X^\dagger\omega^\dagger)\\
    &\subto && \Lambda(\omega)=Y,~\rank(\omega)\le k,
  \end{aligned}
  \label{eq:NonPSDRank}
\end{equation}
where $\omega\in\dC^{m\times n}$, and the form of the objective function is 
chosen such that it is real-valued. Here, we impose an extra assumption that 
the optimal value can be attained on bounded $\omega$, i.e., we consider the 
optimization
\begin{equation}
  \begin{aligned}
    &\maxover[\omega]  && \Tr(X\omega)+\Tr(X^\dagger\omega^\dagger)\\
    &\subto && \Lambda(\omega)=Y,~\norm{\omega}\le R,~\rank(\omega)\le k,
  \end{aligned}
  \label{eq:NonPSDRankBounded}
\end{equation}
where $\norm{\omega}=\Tr(\sqrt{\omega\omega^\dagger})$ is the trace norm of 
$\omega$, and $R$ is a suitably chosen bound depending on the actual problem.  
Especially, by taking $R\to+\infty$, Eq.~\eqref{eq:NonPSDRankBounded} turns to 
Eq.~\eqref{eq:NonPSDRank}. The key observation for solving 
Eq.~\eqref{eq:NonPSDRankBounded} is the following lemma; see 
Appendix~\ref{app:PSDembedding} for the proof.

\begin{lemma}
  A matrix $\omega\in\dF^{m\times n}$ ($\dF=\dC$ or $\dF=\dR$) satisfies that 
  $\rank(\omega)\le k$ and $\norm{\omega}\le R$ if and only if there exists 
  $A\in\dF^{m\times m}$ and $B\in\dF^{n\times n}$ such that
  \begin{equation}
    \Omega:=
    \begin{bmatrix}
      A & \omega\\
      \omega^\dagger & B
    \end{bmatrix}
    \label{eq:Omega}
  \end{equation}
  satisfies that $\Omega\ge 0$, $\Tr(\Omega)=2R$, and $\rank(\Omega)\le k$.
  \label{lem:PSDembedding}
\end{lemma}

By taking advantage of Lemma~\ref{lem:PSDembedding}, the optimization in 
Eq.~\eqref{eq:NonPSDRankBounded} can be written as
\begin{equation}
  \begin{aligned}
    &\maxover[\Omega]  && \Tr(L\Omega)\\
    &\subto && \Lambda\circ P(\Omega)=Y,
	       ~\Tr(\Omega)=2R,\\
    &       && \Omega\ge 0,~\rank(\Omega)\le k,
  \end{aligned}
  \label{eq:NonPSDRankEmbed}
\end{equation}
where
\begin{equation}
  \Omega=
  \begin{bmatrix}
    A & \omega\\
    \omega^\dagger & B
  \end{bmatrix},~
  L=
  \begin{bmatrix}
    0 & X^\dagger\\
    X & 0
  \end{bmatrix},~
  P(\Omega)=\omega.
\end{equation}
Then, after normalization, Eq.~\eqref{eq:NonPSDRankEmbed} is of the simple form 
of the rank-constrained SDP as in Eq.~\eqref{eq:SDPRank}. Thus, all the methods 
developed in Sec.~\ref{sec:main} are directly applicable.

Furthermore, by applying the technique from Sec.~\ref{ssec:ineqcons}, it is 
also possible to consider element-wise inequality constraints of the form 
$\Lambda(\omega)\preceq Y$ for the optimization in Eq.~\eqref{eq:NonPSDRank}, 
where $\preceq$ denotes the element-wise comparison.

\subsection{Unnormalized variables}\label{ssec:unnormvar}

Third, we consider the rank-constrained optimization without the normalization 
constraint. We consider the general rank-constrained SDP
\begin{equation}
  \begin{aligned}
    &\maxover[\rho]  && \Tr(X\rho)\\
    &\subto && \Lambda(\rho)=Y,~M(\rho)\le Z,\\
    &       && \rho\ge0,~\rank(\rho)\le k,
  \end{aligned}
  \label{eq:SDPRankUnnorm}
\end{equation}
in which both the equality constraint ($\Lambda(\rho)=Y$) and the inequality 
constraint ($M(\rho)\le Z$) are involved.

The first method we can try is to find a matrix $C$ such that 
$W:=\Lambda^*(C)>0$, where $\Lambda^*$ is the dual or adjoint map of $\Lambda$ 
\cite{Watrous2018}.  If this is possible, we can add a redundant 
normalization-like constraint
\begin{equation}
  \Tr(W\rho)=w,
\end{equation}
which follows from $\Lambda(\rho)=Y$, where $w=\Tr(CY)$. The 
strictly-positive-definite property of $W$ implies that $w>0$; otherwise the 
problem is trivial ($\rho=0$). Then, by applying the transformation 
$\widetilde{\rho}=w^{-1}\sqrt{W}\rho\sqrt{W}$, the general rank-constrained SDP 
is transformed to the form with the normalization condition over 
$\widetilde{\rho}$.  Thus, the methods in Secs.~\ref{sec:main} and 
\ref{ssec:ineqcons} are directly applicable.

In general, we can combine the techniques of the inequality constraint and the 
non-positive-semidefinite variable to tackle the problem. Again, we impose an 
extra assumption that the optimization can be attained on bounded $\rho$, i.e., 
we consider the optimization
\begin{equation}
  \begin{aligned}
    &\maxover[\rho]  && \Tr(X\rho)\\
    &\subto && \Lambda(\rho)=Y,~M(\rho)\le Z,~\Tr(\rho)\le R,\\
    &       && \rho\ge0,~\rank(\rho)\le k,
  \end{aligned}
  \label{eq:SDPRankUnnormBounded}
\end{equation}
where $R$ is a suitably chosen bound depending on the actual problem. By taking 
advantage of Lemma~\ref{lem:PSDembedding}, the optimization in 
Eq.~\eqref{eq:SDPRankUnnormBounded} can be written as
\begin{align}
  &\maxover[\Omega]  && \Tr(L\Omega) \label{eq:SDPRankUnnormBoundedEmbed}\\
  &\subto && \Lambda\circ P(\Omega)=Y,
  ~M\circ P(\Omega)\le Z,~P(\Omega)\ge 0,\notag\\
  &       && \Tr(\Omega)=2R,~\Omega\ge 0,~\rank(\Omega)\le k,\notag
\end{align}
where
\begin{equation}
  \Omega=
  \begin{bmatrix}
    A & \rho\\
    \rho & B
  \end{bmatrix},~
  L=\frac{1}{2}
  \begin{bmatrix}
    0 & X\\
    X & 0
  \end{bmatrix},~
  P(\Omega)=\rho.
\end{equation}
Then, Eq.~\eqref{eq:SDPRankUnnormBoundedEmbed} is a rank-constrained SDP with 
the normalization constraint. By applying the methods from Sec.~\ref{sec:main} 
and Sec.~\ref{ssec:ineqcons}, a complete SDP hierarchy can be constructed.

\subsection{Quadratic optimization and beyond}
\label{ssec:quadratic}

Last, we show that our method can also be used for (rank-constrained) quadratic 
and higher-order optimization. The key observation is that quadratic functions 
over $\rho$ can be written as linear functions over $\rho\otimes\rho$. For 
example, we can rewrite
\begin{equation}
  \begin{aligned}
    &\Tr(X\rho Y\rho)=\frac{1}{2}\Tr[\{V,X\otimes Y\}(\rho\otimes\rho)],\\
    &\Tr(X\rho)\Tr(Y\rho)=\Tr[(X\otimes Y)(\rho\otimes\rho)],
  \end{aligned}
  \label{eq:quadForm}
\end{equation}
where $V$ is the swap operator, and the anti-commutator $\{\cdot,\cdot\}$ is 
taken to ensure the Hermiticity. Thus, without loss of generality, we consider 
the rank-constrained quadratic optimization
\begin{equation}
  \begin{aligned}
    &\maxover[\rho]  && \Tr[X_{A_1B_1}(\rho_{A_1}\otimes\rho_{B_1})]\\
    &\subto && \Lambda(\rho)=Y,~\Tr(\rho)=1,\\
    &       && \rho\ge0,~\rank(\rho)\le k,
  \end{aligned}
  \label{eq:quadratic}
\end{equation}
where $\cH_{A_1}=\cH_{B_1}=\dC^n$, $\rho_{A_1}$ and $\rho_{B_1}$ denote the 
same state $\rho$ on $\cH_{A_1}$ and $\cH_{B_1}$, respectively, and 
$X_{A_1B_1}$ is some Hermitian matrix on $\cH_{A_1}\otimes\cH_{A_2}$. The 
generalization to the general cases as in the previous subsections is obvious.

\begin{figure}
  \centering
  \includegraphics[width=0.925\linewidth]{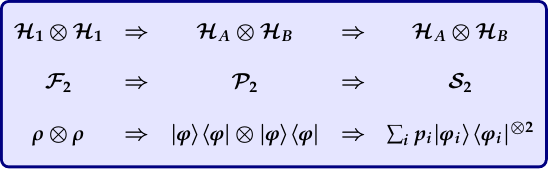}
  \caption{An illustration of the relations between the two-party feasible 
  region $\cF_2$, the two-party purification $\cP_2$, and the two-party 
extension $\cS_2$.}
  \label{fig:quadratic}
\end{figure}

To solve Eq.~\eqref{eq:quadratic}, we consider 
$\cF_2:=\big\{\rho\otimes\rho\bmid\rho\in\cF\big\}$ and 
$\cP_2:=\big\{\ket{\varphi}\bra{\varphi}\otimes\ket{\varphi}\bra{\varphi}
\bmid\ket{\varphi}\bra{\varphi}\in\cP\big\}$; see Fig.~\ref{fig:quadratic}.  
From the definitions in 
Eqs.~(\ref{eq:feasibleRegion},\,\ref{eq:defP},\,\ref{eq:defS2}), we have
\begin{equation}
  \Tr_{A_2B_2}(\cS_2)=\conv\big[\Tr_{A_2B_2}(\cP_2)\big]=\conv(\cF_2).
  \label{eq:cF2}
\end{equation}
As the set $\cS_2$ is already fully characterized by Theorem~\ref{thm:sepOpt}, 
the rank-constrained quadratic optimization in Eq.~\eqref{eq:quadratic} is 
equivalent to the conic program
\begin{alignat}{2}
  \notag
  &\maxover[\Phi_{AB}]~~&& \Tr[\tX_{AB}\Phi_{AB}]\\
  \label{eq:varOpt}
  &\subto && \Phi_{AB}\in\Sep,\,\Tr(\Phi_{AB})=1,\,
  V_{AB}\Phi_{AB}=\Phi_{AB},\\
  \notag
  &       && \tL_A\otimes\id_B(\Phi_{AB})
  =Y\otimes\Tr_A(\Phi_{AB}),
\end{alignat}
where $\tX_{AB}=X_{A_1B_1}\otimes\I_{A_2B_2}$. Accordingly, a complete 
hierarchy can be constructed similarly as in Theorem~\ref{thm:symExtOpt}.

We conclude this section with a few remarks. First, taking $k=n$ (i.e., taking 
the rank bound to be the dimension of $\rho$) corresponds to the quadratic 
programming without rank constraint. Second, this method can be used for 
various uncertainty relations in quantum information, in which the minimization 
of the variance is automatically a quadratic program. Finally, the above 
procedure can be easily generalized to higher-order programming.  The main idea 
is that all the results in Sec.~\ref{sec:main} can be directly generalized to 
fully characterize
\begin{equation}
  \cS_N:=\conv\Big(\Big\{\ket{\varphi}\bra{\varphi}^{\otimes N}
  \;\Big|\;\ket{\varphi}\bra{\varphi}\in\cP\Big\}\Big),
  \label{eq:cSN}
\end{equation}
and $\cS_N$ satisfies that $\Tr_{A_2B_2\cdots Z_2}(\cS_N)=\conv(\cF_N)$, where 
$\cF_N:=\big\{\rho^{\otimes N}\bmid\rho\in\cF\big\}$; see 
Appendix~\ref{app:sepOpt} for more details. Thus, the (rank-constrained) 
higher-order optimization over $\rho^{\otimes N}$ is fully characterizable with 
$\cS_N$.

\section{Conclusion}

We have introduced a method to map SDPs with  rank constraints to optimizations 
over separable quantum states. This allowed us to construct a complete 
hierarchy of SDPs for the original rank-constrained SDP. We studied various 
examples and demonstrated the practical viability of our approach. Finally, we 
discussed several extensions to more general problems.

For further research, there are several interesting directions.  First, 
concerning the presented method, a careful study of possible large-scale 
implementations, including the exploitation of possible symmetries, is 
desirable. This may finally shed new light on some of the examples presented 
here. Second, another promising method for solving the convex optimization 
problems in Theorems~\ref{thm:sepOpt} and \ref{thm:sepOptR} is to consider the 
dual conic programs, which correspond to the optimization over entanglement 
witnesses.  The benefit of this method will be that any feasible witness 
operator can provide a certified upper bound for the optimization problem.  
Third, on a broader perspective, it would be interesting to study other SDPs 
with additional constraints. An example is conditions in a product form, which 
frequently occur in quantum information due to the tensor product structure of 
the underlying Hilbert spaces. Finding SDP hierarchies for such problems will 
be very useful for the progress of this field.

\begin{acknowledgments}
  We would like to thank Matthias Kleinmann and Nikolai Wyderka for 
  discussions. To formulate and solve the SDPs, we used SDPA \cite{SDPA7} as 
  well as CVXPY \cite{Diamond.Boyd2016,Agrawal.etal2018} with SCS 
  \cite{ODonoghue.etal2016} and MOSEK \cite{Mosek} solvers. This work was 
  supported by the Deutsche Forschungsgemeinschaft (DFG, German Research 
  Foundation, project numbers 447948357 and 440958198), the Sino-German Center 
  for Research Promotion (Project M-0294), the ERC (Consolidator Grant 
  683107/TempoQ), and the House of Young Talents Siegen.
\end{acknowledgments}

\appendix

\newtheorem{manualtheoreminner}{Theorem}
\newenvironment{manualtheorem}[1]{%
  \renewcommand\themanualtheoreminner{#1}%
\manualtheoreminner}{\endmanualtheoreminner}

\newtheorem{manuallemmainner}{Lemma}
\newenvironment{manuallemma}[1]{%
  \renewcommand\themanuallemmainner{#1}%
\manuallemmainner}{\endmanuallemmainner}

\section{Characterization of $\cS_2$ and $\cS_N$}\label{app:sepOpt}

First, we prove that the conditions in 
Eqs.~(\ref{eq:S2Sep},\ref{eq:S2Symm},\ref{eq:S2Marginal}) are also sufficient 
for $\Phi_{AB}\in\cS_2$. The constraints in 
Eqs.~(\ref{eq:S2Sep},\,\ref{eq:S2Symm}) imply that $\Phi_{AB}$ is a separable 
state in the symmetric subspace, which always admits the form 
\cite{Toth.Guehne2009}
\begin{equation}
  \Phi_{AB}=\sum_ip_i\ket{\varphi_i}\bra{\varphi_i}_A
  \otimes\ket{\varphi_i}\bra{\varphi_i}_B,
  \label{eq:symmPhiA}
\end{equation}
where the $p_i$ form a probability distribution and the $\ket{\varphi_i}$ are 
normalized. Hereafter, without loss of generality, we assume that all $p_i$ 
are strictly positive. From Eqs.~(\ref{eq:defP},\,\ref{eq:defS2}), to show that 
$\Phi_{AB}\in\cS_2$ we only need to show that
\begin{equation}
  \tL(\ket{\varphi_i}\bra{\varphi_i})=Y
  \label{eq:marginal}
\end{equation}
for all $\ket{\varphi_i}$. To this end, we introduce an auxiliary map
\begin{equation}
  \cE(\cdot)=
  \tL(\cdot)-\Tr(\cdot)Y.
  \label{eq:defE}
\end{equation}
Thus, the last constraint in Eq.~\eqref{eq:sepOpt} is equivalent to 
$\cE_A\otimes \id_B(\Phi_{AB})=0$, which implies that
\begin{equation}
  \cE_A\otimes\cE_B^\dagger(\Phi_{AB})=0,
  \label{eq:cEnull}
\end{equation}
where $\cE^\dagger$ is the linear map satisfying 
$\cE^\dagger(X)=[\cE(X)]^\dagger$ for any Hermitian matrix $X$, and the 
subscripts $A,B$ in $\cE_A,\cE^\dagger_B$ indicate that the maps operate on 
systems $\cH_A$ and $\cH_B$, respectively.
We note that $\cE^\dagger$ is not the dual map of $\cE$. Then, 
Eqs.~(\ref{eq:symmPhiA},\,\ref{eq:cEnull}) imply that
\begin{equation}
  \sum_ip_iE_i\otimes E_i^\dagger=0,
  \label{eq:Ei}
\end{equation}
where $E_i=\cE(\ket{\varphi_i}\bra{\varphi_i})$. Let $V$ be the swap operator 
acting on the same space as $E_i\otimes E_i^\dagger$, then the relations 
$\Tr[V(E_i\otimes E_i^\dagger)]=\Tr(E_iE_i^\dagger)$ imply that
\begin{equation}
  \Tr\left[V\left(\sum_ip_iE_i\otimes E_i^\dagger\right)\right]
  =\sum_ip_i\Tr(E_iE_i^\dagger)=0.
  \label{eq:Einull}
\end{equation}
Furthermore, as $\Tr(E_iE_i^\dagger)>0$ unless $E_i=0$, we obtain
\begin{equation}
  \cE(\ket{\varphi_i}\bra{\varphi_i})=E_i=0
  \label{eq:Epinull}
\end{equation}
for all $\ket{\varphi_i}$. Then, Eq.~\eqref{eq:marginal} follows directly from 
the definition of $\cE$ in Eq.~\eqref{eq:defE}, and hence $\Phi_{AB}\in\cS_2$.  
This proves Theorem~\ref{thm:sepOpt}.

Second, we show that Eq.~(\ref{eq:symmPhi}) and $\Phi_{AB}^{T_A}=\Phi_{AB}$ 
imply that $\ket{\varphi_i}\bra{\varphi_i}\in\dR^{nk\times nk}$ for all $i$.  
From the form of $\Phi_{AB}$ in Eq.~\eqref{eq:symmPhiA}, we obtain
\begin{equation}
  \Phi_{AB}^{T_A}=\sum_ip_i\ket{\varphi_i^*}\bra{\varphi_i^*}_A
  \otimes\ket{\varphi_i}\bra{\varphi_i}_B,
  \label{eq:symmPhiPT}
\end{equation}
where $\ket{\varphi_i^*}$ denote the complex conjugate of $\ket{\varphi_i}$.  
Then, the fact that $\Phi_{AB}^{T_A}=\Phi_{AB}$ is a separable state within the 
symmetric subspace implies that
\begin{equation}
  \ket{\varphi_i}\bra{\varphi_i}=\ket{\varphi_i^*}\bra{\varphi_i^*},
\end{equation}
i.e., $\ket{\varphi_i}\bra{\varphi_i}\in\dR^{nk\times nk}$ for all $i$. This 
proves Theorem~\ref{thm:sepOptR}. Notably, this argument can be directly 
generalized to multi-party states, which provides a simple proof for the result 
in Ref.~\cite{Chen.etal2019}.

Last, we show that the method for characterizing $\cS_2$ can be directly 
generalized to characterizing $\cS_N$. Recall that $\cS_N$ is defined as
\begin{equation}
  \cS_N:=\conv\big(\big\{\ket{\varphi}\bra{\varphi}^{\otimes 
  N}\bmid\ket{\varphi}\bra{\varphi}\in\cP\big\}\big),
  \label{eq:cSNA}
\end{equation}
where $\cP$ is defined as
\begin{equation}
  \cP=\big\{\ket{\varphi}\bra{\varphi}\bmid
    \tL(\ket{\varphi}\bra{\varphi})=Y,~\braket{\varphi}{\varphi}=1
  \big\}.
  \label{eq:defPA}
\end{equation}
Then, we show that $\Phi_{ABC\cdots Z}\in\cS_N$ if and only if
\begin{align}
  \label{eq:SNSep}
  &\Phi_{ABC\cdots Z}\in\Sep,~\Tr(\Phi_{ABC\cdots Z})=1,\\
  \label{eq:SNSymm}
  &P_N^+\Phi_{ABC\cdots Z}P_N^+=\Phi_{ABC\cdots Z},\\
  \label{eq:SNMarginal}
  &\tL_A\otimes\id_{BC\cdots Z}(\Phi_{ABC\cdots Z})
  =Y\otimes\Tr_A(\Phi_{ABC\cdots Z}).
\end{align}
Similarly to case of $\cS_2$, the constraints in Eqs.~\eqref{eq:SNSep} and 
\eqref{eq:SNSymm} imply that $\Phi_{ABC\cdots Z}$ is a separable state in the 
symmetric subspace, which always admits the form \cite{Toth.Guehne2009}
\begin{equation}
  \Phi_{ABC\cdots Z}=\sum_ip_i\ket{\varphi_i}\bra{\varphi_i}^{\otimes N},
  \label{eq:symmPhiNA}
\end{equation}
where the $p_i$ form a probability distribution and the $\ket{\varphi_i}$ are 
normalized. Thus, Eq.~\eqref{eq:SNMarginal} implies that
\begin{equation}
  \cE_A\otimes\cE_B^\dagger\otimes\Tr_{C\cdots Z}(\Phi_{ABC\cdots Z})
  =\sum_ip_iE_i\otimes E_i^\dagger=0,
  \label{eq:cEnullN}
\end{equation}
where $E_i=\cE(\ket{\varphi_i}\bra{\varphi_i})$. Then, 
$\ket{\varphi_i}\bra{\varphi_i}\in\cP$ follows from 
Eqs.~(\ref{eq:Einull},\,\ref{eq:Epinull}), which prove that $\Phi_{ABC\cdots 
Z}\in\cS_N$.  Similarly, in the case of $\dF=\dR$, we only need to add the 
partial-transpose-invariant constraint
\begin{equation}
  \Phi_{ABC\cdots Z}^{T_A}=\Phi_{ABC\cdots Z}.
\end{equation}

\section{Proof of Theorem~\ref{thm:symExtOpt}}\label{app:symExtOpt}

\begin{manualtheorem}{\ref{thm:symExtOpt}}
  For $\dF=\dC$, let $\xi$ be the solution of the rank-constrained SDP in 
  Eq.~\eqref{eq:SDPRank}. Then, for any $N$, $\xi$ is upper bounded by the 
  solution $\xi_N$ of the following SDP hierarchy
  \begin{alignat}{2}
    &\maxover[\Phi_{AB\cdots Z}]  && \Tr(\tX_A\otimes\I_{B\cdots 
    Z}\Phi_{AB\cdots Z})\notag\\
    \label{eq:symExtOptA}
    &\subto && \Phi_{AB\cdots Z}\ge 0,~\Tr(\Phi_{AB\cdots Z})=1,\\
    &       && P_N^+\Phi_{AB\cdots Z}P_N^+=\Phi_{AB\cdots Z},\notag\\
    &       && \tL_A\otimes\id_{B\cdots Z}(\Phi_{AB\cdots Z})\notag
    =Y\otimes\Tr_A(\Phi_{AB\cdots Z}).
  \end{alignat}
  Furthermore, the SDP hierarchy is complete in the sense that 
  $\xi_{N+1}\le\xi_N$ and $\lim_{N\to+\infty}\xi_N=\xi$.
\end{manualtheorem}

The proof is similar to the proof of Theorem~2 in Ref.~\cite{Yu.etal2021}. For 
completeness, we also present it here. To prove Theorem~\ref{thm:symExtOpt}, we 
take advantage of the following lemma, which can be viewed as a special case of 
the quantum de Finetti theorem \cite{Christandl.etal2007}; see also related 
results in Refs.~\cite{Navascues.etal2009b,Berta.etal2018}.

\begin{lemma}
  Let $\rho_N$ be an $N$-party quantum state in the symmetric subspace $P_N^+$, 
  then for all $\ell<N$ there exists an $\ell$-party quantum state
  \begin{equation}
    \sigma_\ell=\sum_\mu p_\mu\ket{\varphi_\mu}\bra{\varphi_\mu}^{\otimes\ell},
    \label{eq:ksep}
  \end{equation}
  i.e., a fully separable state in $P_\ell^+$, such that
  \begin{equation}
    \norm{\Tr_{N-\ell}(\rho_N)-\sigma_\ell}
    \le\frac{4\ell D}{N},
    \label{eq:deFinetti}
  \end{equation}
  where $\norm{\cdot}$ is the trace norm and $D$ is the local dimension.
  \label{lem:deFinetti}
\end{lemma}

The part that $\xi$ is upper bounded by $\xi_N$ for any $N$ is obvious.  Hence, 
we only need to prove that $\xi_{N+1}\le\xi_N$ and 
$\lim_{N\to+\infty}\xi_N=\xi$.

We first show that $\xi_{N+1}\le\xi_N$. This follows from 
the fact that if a multi-party quantum state is within the symmetric 
subspace, so are the reduced states. Mathematically, we have the relation
\begin{equation}
  (P_N^+\otimes\I_{nk})P_{N+1}^+=P_{N+1}^+.
\end{equation}
Suppose that there exists an $(N+1)$-party extension $\Phi_{AB\cdots ZZ'}$ 
satisfying all constraints that achieves the maximum $\xi_{N+1}$ in 
Theorem~\ref{thm:symExtOpt}. Then, the constraint
$P_{N+1}^+\Phi_{AB\cdots ZZ'}P_{N+1}^+=\Phi_{AB\cdots ZZ'}$ implies that
\begin{equation}
  \begin{aligned}
    &(P_N^+\otimes\I_{nk})\Phi_{AB\cdots ZZ'}(P_N^+\otimes\I_{nk})\\
    =&P_N^+\otimes\I_{nk}P_{N+1}\Phi_{AB\cdots ZZ'}P_{N+1}P_N^+\otimes\I_{nk}\\
    =&P_{N+1}\Phi_{AB\cdots ZZ'}P_{N+1}\\
    =&\Phi_{AB\cdots ZZ'}.
  \end{aligned}
  \label{eq:subsysSymm}
\end{equation}
Thus, one can easily verify that the reduced state $\Tr_{Z'}(\Phi_{AB\cdots 
ZZ'})$ is an $N$-party extension satisfying all the constraints in 
Theorem~\ref{thm:symExtOpt} with objective value 
$\xi_{N+1}$.  From this, the result $\xi_{N+1}\le\xi_N$ follows.

Next, we prove the convergence part, i.e., $\lim_{N\to+\infty}\xi_N=\xi$.  
Suppose that the solution $\xi_N$ of the $N$-party extension in 
Theorem~\ref{thm:symExtOpt} is achieved by the quantum state $\Phi_{AB\cdots 
Z}$.  Let $\Phi^N_{AB}=\Tr_{C\cdots Z}(\Phi_{ABC\cdots Z})$, then $\Phi^N_{AB}$ 
satisfies that
\begin{equation}
  \begin{aligned}
    &\Tr(\tX_A\otimes\I_B\Phi_{AB}^N)=\xi_N,~\Tr(\Phi^N_{AB})=1,\\
    &\tL_A\otimes\id_B(\Phi^N_{AB})=Y\otimes\Tr_A(\Phi^N_{AB}).
  \end{aligned}
  \label{eq:reducedPhi}
\end{equation}
Further, Lemma~\ref{lem:deFinetti} implies that there exist separable states
$\widetilde{\Phi}^N_{AB}$ such that
\begin{align}
  &V_{AB}\widetilde{\Phi}^N_{AB}=\widetilde{\Phi}^N_{AB},
  \label{eq:PhideFinetti1}\\
  &\norm{\Phi^N_{AB}-\widetilde{\Phi}^N_{AB}}\le\frac{8nk}{N}.
  \label{eq:PhideFinetti2}
\end{align}
As the set of quantum states for any fixed dimension is compact, we can choose 
a convergent subsequence $\Phi^{N_i}_{AB}$ of the sequence $\Phi^{N}_{AB}$.  
Thus, Eq.~\eqref{eq:PhideFinetti2} implies that
\begin{equation}
  \Phi_{AB}:=\lim_{i\to+\infty}\Phi^{N_i}_{AB}
  =\lim_{i\to+\infty}\widetilde{\Phi}^{N_i}_{AB}.
  \label{eq:convergentPhi}
\end{equation}
As all $\widetilde{\Phi}^{N_i}_{AB}$ are separable and the set of separable 
states is closed, thus 
$\Phi_{AB}=\lim_{i\to+\infty}\widetilde{\Phi}^{N_i}_{AB}$ is separable.  
Further, as all the functions on $\Phi^N_{AB}$ or $\widetilde{\Phi}^N_{AB}$ in 
Eqs.~(\ref{eq:reducedPhi},\,\ref{eq:PhideFinetti1}) are continuous, 
Eq.~\eqref{eq:convergentPhi} implies that $\Phi_{AB}$ satisfies all the 
constraints in Eq.~\eqref{eq:sepOpt}. In other words, $\Phi_{AB}$ is a feasible 
point of program~\eqref{eq:sepOpt}, thus 
$\Tr(\widetilde{X}_A\otimes\I_B\Phi_{AB})=\lim_{N\to+\infty}\xi_N\le\xi$.  
Together with the fact that $\xi_N\ge\xi$, we then have 
$\lim_{N\to+\infty}\xi_N=\xi$.

At last, we would like to note that the above proof also gives an estimate of 
the convergence rate of the SDP hierarchy.  Without loss of generality, we 
assume some bound conditions on $\Lambda,X,Y$ such that their operations do not 
increase the trace norm, e.g., the diamond norm of $\Lambda(\cdot)-\Tr(\cdot)Y$ 
and the spectrum norm of $X$ are no greater than one.  Then, each level of the 
SDP hierarchy provides an $\cO(nk/N)$ approximate solution to the convex 
optimization in Theorem~\ref{thm:sepOpt}.  More precisely, for each $N$ there 
exists $\varepsilon=\cO(nk/N)$ such that
\begin{equation}
  \xi\le\xi_N\le\Tr[(\tX_A\otimes\I_B)\Phi_{AB}]+\varepsilon,
  \label{eq:complexitySep1}
\end{equation}
where $\Phi_{AB}$ is $\varepsilon$-close to the feasible region in the sense 
that
\begin{align}
  \label{eq:complexitySep2}
  &\Phi_{AB}\in\Sep,~\Tr(\Phi_{AB})=1,
  ~V_{AB}\Phi_{AB}=\Phi_{AB},\\
  \label{eq:complexitySep3}
  &\norm{\tL_A\otimes\id_B(\Phi_{AB})
  -Y\otimes\Tr_A(\Phi_{AB})}\le\varepsilon.
\end{align}
This can be proved by taking $\Phi_{AB}$ as $\widetilde{\Phi}^N_{AB}$ in above 
proof and the approximate rate $\varepsilon=\cO(nk/N)$ results from 
Eq.~\eqref{eq:PhideFinetti2}.

\section{Continuous rank from inherent symmetry}\label{app:contrank}

In the following, we consider the second level of the hierarchy, i.e., $N=2$.  
As described in Sec.~\ref{ssec:globalsym}, the corresponding SDPs can be 
simplified to Eqs.~\eqref{eq:s2T}~and~\eqref{eq:s2Tr} for the complex and real 
cases, respectively.  Here, the parameter $k$, that constrains the rank in the 
rank-constrained SDP in Eq.~\eqref{eq:SDPRank}, in principle, does not need to 
be an integer. Indeed, the following observations show that it is not 
unreasonable to consider this, in some sense, continuous rank.
\begin{observation}\label{obs:contrank}
  A feasible point $\Phi_{A_1B_1} = k^2 \Phi_I + k \Phi_V$ of the SDP in 
  Eq.~\eqref{eq:s2T} with parameter $k \ge 1$ is also a feasible point 
  $\Phi_{A_1B_1} = k'^2 \Phi_I' + k' \Phi_V' $ of the SDP with parameter $k' 
  \ge k$.
\end{observation}
\begin{proof}
  The observation is trivial when $k'=k$. In the following, we assume that 
  $k'>k\ge 1$. From the relations that $\Phi_V=V\Phi_I$, $\Phi'_V=V\Phi'_I$, 
  and $\Phi_{A_1B_1}=k^2\Phi_I+k\Phi_V=k'^2\Phi'_I+k'\Phi_V'$, one can obtain 
  that
  \begin{equation}
    \begin{aligned}
      k'^2\Phi'_I+k'\Phi_V'&=k^2\Phi_I+k\Phi_V,\\
      k'\Phi'_I+k'^2\Phi_V'&=k\Phi_I+k^2\Phi_V,
    \end{aligned}
  \end{equation}
  which further imply that
  \begin{equation}
    \Phi_I' = \frac{k(kk'-1)}{k'(k'^2-1)}\Phi_I
    +\frac{k(k'-k)}{k'(k'^2-1)}\Phi_V.
  \end{equation}
  Thus, we can express $\Phi_I'$ and $\Phi_V'$ in terms of $\Phi_I$ and 
  $\Phi_V$. The feasibility follows from the feasibility of $\Phi_{A_1B_1} 
  = k^2 \Phi_I + k \Phi_V$
  and
  \begin{equation}
    \begin{aligned}
      &\Phi_I' \pm \Phi_V' = \frac{k(k\pm 1)}{k'(k'\pm 1)}
      \Big(\Phi_I \pm \Phi_V\Big), \\
      &\big(\Phi_I'\big)^{T_{A_1}} = \frac{k'-k}{k'(k'^2-1)} 
      \left(\Phi_I^{T_{A_1}} + k \Phi_V^{T_{A_1}}\right)+ \frac{k^2-1}{k'^2-1} 
      \Phi_I^{T_{A_1}},\\
      &\big(\Phi_I'\big)^{T_{A_1}} + k'\big(\Phi_V'\big)^{T_{A_1}} 
      = \frac{k}{k'} \left(\Phi_I^{T_{A_1}} + k \Phi_V^{T_{A_1}}\right),
    \end{aligned}
  \end{equation}
  since all coefficients are nonnegative. The linear constraints are obviously 
  satisfied as we consider $\Phi_{A_1B_1}' = \Phi_{A_1B_1}$.
\end{proof}
A similar statement also holds in the real case.
\begin{observation}
  A feasible point $\Phi_{A_1B_1} = k^2 \Phi_I + k \Phi_V + k \Phi_\phi$ of the 
  SDP in Eq.~\eqref{eq:s2Tr} with parameter $k \ge 1$ is also a feasible point 
  $\Phi_{A_1B_1} = k'^2 \Phi_I' + k' \Phi_V' + k' \Phi_\phi'$ of the SDP with 
  parameter $k' \ge k$.
\end{observation}

\begin{proof}
  In this case, from  $\Phi'_V=V\Phi'_I$, $\Phi'_\phi=(\Phi'_V)^{T_{A_1}}$, and
  $k'^2\Phi_I'+k'\Phi_V'+k'\Phi_\phi'=k^2\Phi_I+k\Phi_V+k\Phi_\phi$, we obtain 
  that
  \begin{equation}
    \begin{aligned}
      \Phi'_I=&\frac{k(kk'+k-2)}{k'(k'+2)(k'-1)}\Phi_I\\
      +&\frac{k(k'-k)}{k'(k'+2)(k'-1)}\Phi_V\\
      +&\frac{k(k'-k)}{k'(k'+2)(k'-1)}\Phi_\phi.
    \end{aligned}
  \end{equation}
  Analogous to the proof of Observation~\ref{obs:contrank}, it is 
  straightforward to verify that the coefficients in the following equalities 
  are all nonnegative,
  \begin{equation}
    \begin{aligned}
      &\Phi'_I+\Phi'_V=\frac{(k+2)(k-1)}{(k'+2)(k'-1)}
      \Big( \Phi_I+\Phi_V \Big),\\
      &\qquad+\frac{2(k'-k)}{k'(k'+2)(k'-1)}
      \Big( \Phi_I+\Phi_V+k\Phi_\phi\Big)\\
      &\Phi'_I-\Phi'_V=\frac{k(k-1)}{k'(k'-1)}
      \Big( \Phi_I-\Phi_V \Big),\\
      &\Phi'_I+\Phi'_V+k'\Phi'_\phi=\frac{k}{k'}
      \Big(\Phi_I+\Phi_V+k\Phi_\phi \Big).\\
    \end{aligned}
  \end{equation}
  It is also obvious that $(\Phi'_I)^{T_{A_1}}=\Phi'_I$ and 
  $V\Phi'_\phi=\Phi'_\phi$. Hence, the feasibility follows.
\end{proof}

Thus, the set of feasible points grows monotonically with continuous
$k$, and hence, the same is true for the objective value. Apart from the interpretation
as a continuous rank, this also helps in preventing invalid conclusions because
of numerical errors, since parameters $k$ can be sampled in a region around the considered
rank.

\section{Proof of Theorem~\ref{thm:sepOptIneq}}\label{app:ineqcons}

\begin{manualtheorem}{\ref{thm:sepOptIneq}}
  For $\dF=\dC$, let $\xi$ be the solution of the rank-constrained SDP in 
  Eq.~\eqref{eq:SDPRankIneq}. Then, for any $N$, $\xi$ is upper bounded by the 
  solution $\xi_N$ of the following SDP hierarchy
  \begin{alignat}{2}
    &\maxover[\Phi_{AB\cdots Z}]  && \Tr(\tX_A\otimes\I_{B\cdots Z}\Phi_{AB\cdots 
    Z})\notag\\
    \label{eq:symExtOptIneqA}
    &~\subto && \Phi_{AB\cdots Z}\ge 0,~\Tr(\Phi_{AB\cdots Z})=1,\\
    &       && P_N^+\Phi_{AB\cdots Z}P_N^+=\Phi_{AB\cdots Z},
    \notag\\
    &       && \tL\otimes\id_{B\cdots Z}(\Phi_{AB\cdots Z}) \le 
    Y\otimes\Tr_A(\Phi_{AB\cdots Z}).\notag
  \end{alignat}
   Furthermore, the SDP hierarchy is complete, 
  i.e., $\xi_{N+1}\le\xi_N$ and $\lim_{N\to+\infty}\xi_N=\xi$.
\end{manualtheorem}

We will denote the set of pure states embedded in the space of Hermitian 
operators by $\Omega (d) = \{\ketbra{\varphi}{\varphi} \mid 
\braket{\varphi}{\varphi}=1\}$. In this appendix, by means of this embedding, 
we also identify the symbol $\varphi$ as the operator 
$\ketbra{\varphi}{\varphi}$, and analogously for $\psi$. Note that $\Omega(d)$ 
is a compact metric space.  Topologically it is also a separable space, i.e., 
it has a countable dense subset.

Let us start with restating the quantum de Finetti theorem for infinite 
sequences explicitly. Although the results are known in the literature, we 
repeat here a simple proof for completeness. This proof is based on 
the proof of the quantum de Finetti theorem for finite sequences in 
Lemma~\ref{lem:deFinetti}~\cite[Theorem~II.8]{Christandl.etal2007}.

\begin{lemma}[quantum de Finetti Theorem for bosonic sequences]
Let $\{\rho_n\}_{n=0}^{\infty}$ be a sequence of bosonic extensions of density 
operators over $[\dC^d]^{\otimes n}$, i.e., $\rho_0=1$, $\rho_k 
= \mathrm{Tr}_{n-k} [\rho_n]$ for all $n \ge k \ge 0$, and the $\rho_n$ are in the 
symmetric subspace of $[\dC^d]^{\otimes n}$. Then, there exists a Borel 
probability measure $\mu$ over $\Omega(d)$ such that
\begin{equation}
\rho_n = \int\md \mu (\varphi) \ketbra{\varphi}{\varphi}^{\otimes n}
\end{equation}
for all $n=0,1,2,\ldots$
\label{th:definetti1}
\end{lemma}

\begin{proof}
According to Lemma~\ref{lem:deFinetti}, for a given $N$ and $n$, there is 
a Borel probability measure (in fact, in this case the measure can be discrete) 
$\mu_{(N,n)}$ over $\Omega(d)$ such that
\begin{equation}
\norm{\rho_n - \int\md \mu_{(N,n)} (\varphi) \ketbra{\varphi}{\varphi}^{\otimes 
n}}  \le \frac{4 d n}{N}.  \end{equation}
Consider the sequence of Borel probability measures 
$\{\mu_{(N,n)}\}_{N=n+1}^{\infty}$.  Since the space of all Borel probability 
measures over $\Omega(d)$ is sequentially compact in the weak 
topology~\cite[Theorem 8.9.3]{Bogachev2007}, this sequence has at least one 
limit point $\mu_n$. Thus, we have that
\begin{equation}
\rho_n = \int \md \mu_n (\varphi) \ketbra{\varphi}{\varphi}^{\otimes n}.  
\end{equation}
Let $\mu$ in turn be a limit point of the sequence $\{\mu_n\}_{n=1}^{\infty}$, then one finds 
\begin{equation}
\rho_n = \int \md \mu (\varphi) \ketbra{\varphi}{\varphi}^{\otimes n}
\end{equation}
for all $n=0,1,2,\ldots$
\end{proof}

The measure that appears in Lemma~\ref{th:definetti1} can also be shown to be 
unique -- a fact that is needed for our subsequent argument. This is 
a consequence of the following lemma.

\begin{lemma}
The self-adjoint commutative algebra of continuous functions over $\Omega (d)$ 
generated by functions of the form $\Tr[\ketbra{\psi}{\psi}\,\cdot\,]: \Omega 
(d) \to \dR$ for all $\psi \in \Omega(d)$ and the constant functions is dense 
in the space $C[\Omega(d)]$ of continuous functions on $\Omega(d)$.
\label{th:dense}
\end{lemma}

\begin{proof}
The lemma is a direct consequence of Stone's theorem~\cite[Theorem 
4.3.4]{Pedersen1989}, which states that a self-adjoint subalgebra of 
$C[\Omega(d)]$ containing the constants and separating points in $\Omega(d)$ is 
uniformly dense in $C[\Omega(d)]$. Indeed, the algebra contains the constants 
by construction. It also separates points in $\Omega$ since for any two 
distinguished points $\ketbra{\varphi_1}{\varphi_1},  
\ketbra{\varphi_2}{\varphi_2} \in C[\Omega(d)]$, the function $\Tr 
[\ketbra{\varphi_1}{\varphi_1}\,\cdot\,]$ separates them.
\end{proof}

\begin{corollary}
\label{th:tensor_moments}
Let $\mu$ be a Borel measure on $\Omega(d)$, then the set of tensor moments
\begin{equation}
\int \md \mu (\varphi) \ketbra{\varphi}{\varphi}^{\otimes n}
\label{eq:tensor_moments}
\end{equation}
for $n \in \dN$ specify $\mu$ uniquely. As a consequence, the measure $\mu$ 
that appears in Lemma~\ref{th:definetti1} is unique.
\end{corollary}

\begin{proof}
The first step in the proof is to use the Riesz representation, or 
Riesz--Markov theorem~\cite[Theorem 7.10.4]{Bogachev2007}, which establishes 
a one-to-one correspondence between a Borel measure $\mu$ over the compact 
metric space $\Omega(d)$ and a linear functional over the continuous functions 
$L: C[\Omega(d)] \to \dR$ by $L(g) = \int \md \mu (\varphi) g(\varphi)$ for $g 
\in C[\Omega(d)]$. Now the set of tensor moments in Eq.~\eqref{eq:tensor_moments} in 
fact specifies the values of the linear functionals  $L$ on the whole 
subalgebra generated by the constants and functions of the form 
$\Tr[\ketbra{\psi}{\psi}\,\cdot\,]:\Omega (d)\to\dR$ for all 
$\ketbra{\psi}{\psi} \in \Omega(d)$.  Since this subalgebra is dense in 
$C[\Omega(d)]$ by Lemma~\ref{th:dense}, these values uniquely specify the 
linear functional $L$, and thus the Borel measure $\mu$.
\end{proof}

\begin{corollary}
Let $\mu$ be a Borel probability measure over $\Omega(d)$ and $f$ be 
a continuous function over $\Omega(d)$ such that
\begin{equation}
\int \md \mu(\varphi) f(\varphi) \ketbra{\varphi}{\varphi}^{\otimes n} \ge 0
\label{eq:positive_moments}
\end{equation}
for all $n \in \dN$. Then, $f$ is nonnegative almost everywhere with respect to 
$\mu$.
\label{th:positive_moments}
\end{corollary}

\begin{proof}
We consider two cases:

(i) Suppose $\int \md \mu (\varphi) f(\varphi) = 0$, then the operators $\int 
\md \mu(\varphi) f(\varphi) \ketbra{\varphi}{\varphi}^{\otimes n}$ in 
Eq.~\eqref{eq:positive_moments} are traceless. The positivity condition implies 
that, in this case, equality holds in Eq.~\eqref{eq:positive_moments} for all 
$n\in\dN$. By Corollary~\ref{th:tensor_moments}, the measure $f(\varphi) \md 
\mu(\varphi)$ vanishes, and therefore, $f(\varphi)=0$ almost everywhere with respect 
to $\mu$.

(ii) Suppose $\int \md \mu (\varphi) f(\varphi) >0$. Consider the sequence of 
states $\rho_n= \int \md \mu(\varphi) f(\varphi) 
\ketbra{\varphi}{\varphi}^{\otimes n}/ \int \md \mu(\varphi) f(\varphi)$, which 
clearly satisfy the condition of the quantum de Finetti theorem in 
Lemma~\ref{th:definetti1}. Therefore, there is a unique Borel probability 
measure $\nu$ over $\Omega(d)$ such that
\begin{equation}
\rho_n = \int \md \nu (\varphi) \ketbra{\varphi}{\varphi}^{\otimes n}.
\end{equation}
Thus, Corollary~\ref{th:tensor_moments} implies that one can identify $d \mu 
(\varphi) f(\varphi)/\int \md \mu(\varphi) f(\varphi)= \md \nu (\varphi)$.  In 
particular, $f(\varphi) \ge 0$ almost everywhere with respect to $\mu$.
\end{proof}

To extend this result to operator-valued functions, we need the following 
lemma.
\begin{lemma}
Let $X$ be a separable topological space, $(Y,\mu)$ be a measure space. Let 
$\phi: X \times Y \to \dR$ be a function such that $\phi(\,\cdot\,,y)$ is 
continuous for any $y \in Y$. If for any fixed $x \in X$, $\phi(x, y ) \ge 0$ 
for almost all $y \in Y$, then for almost all $y \in Y$, it holds that 
$\phi(x,y) \ge 0$ for all $x$.
\label{th:topo_measure}
\end{lemma}

\begin{proof}
Because $X$ is a separable space, it has a countable dense subset 
$\{x_i\}_{i\in\dN}$.  Hence, the continuity of $\phi(\,\cdot\,,y)$ implies that 
for a given $y$, $\phi(x,y) \ge 0$ for all $x$ is equivalent to $\phi(x_i,y) 
\ge 0$ for all $i\in\dN$. Let $A=\{y \in Y: \mbox{$\phi(x_i,y) < 0$ for some 
$i$}\}$ and $B_i = \{y \in Y: \phi(x_i,y) < 0\}$, then it is clear that $A 
= \cup_{i=1}^{\infty} B_i$.  Therefore $\mu (A) \le \sum_{i=1}^{\infty} \mu 
(B_i) = 0$ by subadditivity.  Thus $\mu(A)=0$, and it follows that for almost 
all $y \in Y$, we have that $\phi(x,y) \ge 0$ for all $x$.
\end{proof}

\begin{corollary}
Let $\mu$ be a Borel probability measure over $\Omega(d)$ and $F$ be 
a continuous function over $\Omega(d)$ with Hermitian operators 
acting on $\dC^D$ as values such that
\begin{equation}
\int \md \mu(\varphi) F(\varphi) \otimes \ketbra{\varphi}{\varphi}^{\otimes N} 
\ge 0
\label{eq:tensor_positive_moments}
\end{equation}
for all $N \in \dN$. Then, $F$ is almost everywhere positive semidefinite with respect to 
$\mu$.
\label{th:tensor_positive_moments}
\end{corollary}

\begin{proof}
Let $\ket{\psi}$ be a state in $\dC^D$, then it follows 
from Eq.~\eqref{eq:tensor_positive_moments} that \begin{equation}
\int \md \mu(\varphi) \bra{\psi} F(\varphi) \ket{\psi} 
\ketbra{\varphi}{\varphi}^{\otimes N} \ge 0
\end{equation}
for all $N$. Therefore, according to Corollary~\ref{th:positive_moments}, we 
deduce that $\bra{\psi} F(\varphi) \ket{\psi} \ge 0$ for almost all $\varphi$ 
(with respect to $\mu$). Now consider $\bra{\psi} F(\varphi) 
\ket{\psi}=\Tr[\ketbra{\psi}{\psi}F(\varphi)]$ as a function over $\Omega(D) 
\times \Omega(d)$ and apply Lemma~\ref{th:topo_measure}. Then, we find that for 
almost all $\varphi$ (with respect to $\mu$), $\bra{\psi} F(\varphi) \ket{\psi} 
\ge 0$  for all $\psi$.  In other words, $F(\varphi) \ge 0$ almost everywhere 
with respect to $\mu$.
\end{proof}

With all these preparations, let us prove the completeness of the hierarchy for 
the SDP with inequality constraint in Theorem~\ref{thm:sepOptIneq}.

The argument given in the main text can be broken into two steps. In the first 
step, one shows that there exists a Borel probability measure $\mu$ on pure 
states $\ketbra{\varphi}{\varphi} \in \Omega(d)$ such that
the state
\begin{equation}
\Phi_{ABC \ldots Z} = \int \md \mu (\varphi) \ketbra{\varphi}{\varphi}^{\otimes 
N}
\end{equation}
satisfies all constraints in Eq.~\eqref{eq:symExtOptIneqA} and $\lim_{N \to 
+\infty} \xi_N = \Tr (\widetilde{X}_A \otimes \Phi_A)$. The proof for this step 
given in the main text is essentially complete. One only has to keep in mind 
that, in principle, the measure $\mu$ that arises is of a general probabilistic 
Borel measure and may not correspond to a probability density function.

In the second step, one shows that almost all $\varphi$ (with respect to the 
measure $\mu$) belong to $\mathcal{P}$ as defined in 
equation~\eqref{eq:PIneq}. In the main text, an intuitive argument is given 
under the assumption that $\mu$ can be written as a well-behaved distribution.  
With all the above mathematical preparations, we can now remove this 
assumption.

We need to show that $\tL(\ketbra{\varphi}{\varphi}) \le Y$ almost everywhere.  
This is the case since we have \begin{equation}
\int \md \mu(\varphi) [Y - \tL(\ketbra{\varphi}{\varphi})] \otimes 
\ketbra{\varphi}{\varphi}^{N} \ge 0
\end{equation}
for all $N \in \dN$. Thus, according to Corollary~\ref{th:tensor_positive_moments}, 
$Y-\tL(\ketbra{\varphi}{\varphi})\ge0$ holds almost everywhere.

\section{Proof of Lemma~\ref{lem:PSDembedding}}\label{app:PSDembedding}

\begin{manuallemma}{\ref{lem:PSDembedding}}
  A matrix $\omega\in\dF^{m\times n}$ ($\dF=\dC$ or $\dF=\dR$) satisfies that 
  $\rank(\omega)\le k$ and $\norm{\omega}\le R$ if and only if there exists 
  $A\in\dF^{m\times m}$ and $B\in\dF^{n\times n}$ such that
  \begin{equation}
    \Omega:=
    \begin{bmatrix}
      A & \omega\\
      \omega^\dagger & B
    \end{bmatrix}
    \label{eq:OmegaA}
  \end{equation}
  satisfies that $\Omega\ge 0$, $\Tr(\Omega)=2R$, and $\rank(\Omega)\le k$.
\end{manuallemma}

For the proof of Lemma~\ref{lem:PSDembedding}, we take advantage of the 
following observations resulting from the inequality between the arithmetic and 
geometric means:\\
Observation~(i): For any $a,b,x\ge 0$ satisfying $ab\ge x^2$, we have that 
$a+b\ge 2x$.\\
Observation~(ii): For any $y\ge x\ge 0$, there exist $a,b\ge 0$ such that 
$ab=x^2$ and $a+b=2y$.

We first prove the ``if'' part. The rank statement is obvious because the rank 
of a submatrix is no larger than that of the whole matrix, i.e., 
$\rank(\omega)\le\rank(\Omega)\le k$. Now, we show that $\Omega\ge 0$ and 
$\Tr(\Omega)=2R$ imply that $\norm{\omega}\le R$. Consider the singular value 
decomposition of $\omega$
\begin{equation}
  \omega=U^\dagger DV,
  \label{eq:SVD}
\end{equation}
where $U$ and $V$ are unitary matrices, $D_{ii}\ge 0$, and $D_{ij}=0$ for $i\ne 
j$.  Furthermore, we have $\norm{\omega}=\sum_{i=1}^{\min\{m,n\}}D_{ii}$. Let
\begin{equation}
  \widetilde{\Omega}
  =
  \begin{bmatrix}
    U & 0\\
    0 & V
  \end{bmatrix}
  \Omega
  \begin{bmatrix}
    U^\dagger & 0\\
    0 & V^\dagger
  \end{bmatrix}
  =
  \begin{bmatrix}
    UAU^\dagger & D\\
    D^T & VBV^\dagger
  \end{bmatrix}.
  \label{eq:OmegaTilde}
\end{equation}
Then, $\Omega\ge 0$ implies that
\begin{equation}
  (UAU^\dagger)_{ii}(VBV^\dagger)_{ii}\ge D_{ii}^2,
\end{equation}
for $i=1,2,\dots,\min\{m,n\}$. Thus, Observation (i) implies that
\begin{equation}
  (UAU^\dagger)_{ii}+(VBV^\dagger)_{ii}\ge 2D_{ii},
\end{equation}
for $i=1,2,\dots,\min\{m,n\}$, whose summation gives that
\begin{equation}
  \begin{aligned}
    \Tr(\Omega)=&\sum_{i=1}^m
    (UAU^\dagger)_{ii}+\sum_{i=1}^n(VBV^\dagger)_{ii}\\
    \ge&\sum_{i=1}^{\min\{m,n\}}
    \left[(UAU^\dagger)_{ii}+(VBV^\dagger)_{ii}\right]\\
    \ge&2\sum_{i=1}^{\min\{m,n\}}D_{ii}.
  \end{aligned}
\end{equation}
Hence, $\Tr(\Omega)=2R$ gives that 
$\norm{\omega}=\sum_{i=1}^{\min\{m,n\}}D_{ii}\le R$.

To prove the ``only if'' part, we again consider the decomposition in 
Eq.~\eqref{eq:OmegaTilde}. Then, $\rank(\omega)\le k$ implies that $D_{ii}=0$ 
when $i>k$. One can easily verify that $\Omega$ satisfies that $\Omega\ge 0$ 
and $\rank(\Omega)=\rank(\widetilde{\Omega})\le k$ when
\begin{equation}
  \begin{aligned}
    &(UAU^\dagger)_{ij}=0
    ~~\text{ for }i\ne j \text{ and } i=j>k,\\
    &(VBV^\dagger)_{ij}=0
    ~~\text{ for }i\ne j \text{ and } i=j>k,\\
    &(UAU^\dagger)_{ii}\ge 0,~(VBV^\dagger)_{ii}\ge 0
    ~~\text{ for } i=1,2,\dots k,\\
    &(UAU^\dagger)_{ii}(VBV^\dagger)_{ii}=D_{ii}^2
    ~~\text{ for } i=1,2,\dots k.\\
  \end{aligned}
\end{equation}
Then, Observation~(ii) and the bound constraint 
$\norm{\omega}=\sum_{i=1}^kD_{ii}\le R$ imply that we can choose suitable 
$(UAU^\dagger)_{ii}$ and $(VBV^\dagger)_{ii}$ for $i=1,2,\dots,k$ such that 
$\Tr(\Omega)=\Tr(\widetilde{\Omega})
=\sum_{i=1}^k[(UAU^\dagger)_{ii}+(VBV^\dagger)_{ii}]=2R$.

\bibliography{QuantumInf}

\begin{thebibliography}{71}%
\makeatletter
\providecommand \@ifxundefined [1]{%
 \@ifx{#1\undefined}
}%
\providecommand \@ifnum [1]{%
 \ifnum #1\expandafter \@firstoftwo
 \else \expandafter \@secondoftwo
 \fi
}%
\providecommand \@ifx [1]{%
 \ifx #1\expandafter \@firstoftwo
 \else \expandafter \@secondoftwo
 \fi
}%
\providecommand \natexlab [1]{#1}%
\providecommand \enquote  [1]{``#1''}%
\providecommand \bibnamefont  [1]{#1}%
\providecommand \bibfnamefont [1]{#1}%
\providecommand \citenamefont [1]{#1}%
\providecommand \href@noop [0]{\@secondoftwo}%
\providecommand \href [0]{\begingroup \@sanitize@url \@href}%
\providecommand \@href[1]{\@@startlink{#1}\@@href}%
\providecommand \@@href[1]{\endgroup#1\@@endlink}%
\providecommand \@sanitize@url [0]{\catcode `\\12\catcode `\$12\catcode
  `\&12\catcode `\#12\catcode `\^12\catcode `\_12\catcode `\%12\relax}%
\providecommand \@@startlink[1]{}%
\providecommand \@@endlink[0]{}%
\providecommand \url  [0]{\begingroup\@sanitize@url \@url }%
\providecommand \@url [1]{\endgroup\@href {#1}{\urlprefix }}%
\providecommand \urlprefix  [0]{URL }%
\providecommand \Eprint [0]{\href }%
\providecommand \doibase [0]{https://doi.org/}%
\providecommand \selectlanguage [0]{\@gobble}%
\providecommand \bibinfo  [0]{\@secondoftwo}%
\providecommand \bibfield  [0]{\@secondoftwo}%
\providecommand \translation [1]{[#1]}%
\providecommand \BibitemOpen [0]{}%
\providecommand \bibitemStop [0]{}%
\providecommand \bibitemNoStop [0]{.\EOS\space}%
\providecommand \EOS [0]{\spacefactor3000\relax}%
\providecommand \BibitemShut  [1]{\csname bibitem#1\endcsname}%
\let\auto@bib@innerbib\@empty
\bibitem [{\citenamefont {Doherty}\ \emph {et~al.}(2002)\citenamefont
  {Doherty}, \citenamefont {Parrilo},\ and\ \citenamefont
  {Spedalieri}}]{Doherty.etal2002}%
  \BibitemOpen
  \bibfield  {author} {\bibinfo {author} {\bibfnamefont {A.~C.}\ \bibnamefont
  {Doherty}}, \bibinfo {author} {\bibfnamefont {P.~A.}\ \bibnamefont
  {Parrilo}},\ and\ \bibinfo {author} {\bibfnamefont {F.~M.}\ \bibnamefont
  {Spedalieri}},\ }\bibfield  {title} {\bibinfo {title} {Distinguishing
  separable and entangled states},\ }\href
  {https://doi.org/10.1103/PhysRevLett.88.187904} {\bibfield  {journal}
  {\bibinfo  {journal} {Phys. Rev. Lett.}\ }\textbf {\bibinfo {volume} {88}},\
  \bibinfo {pages} {187904} (\bibinfo {year} {2002})}\BibitemShut {NoStop}%
\bibitem [{\citenamefont {Navascu\'es}\ \emph {et~al.}(2007)\citenamefont
  {Navascu\'es}, \citenamefont {Pironio},\ and\ \citenamefont
  {Ac\'{\i}n}}]{Navascues.etal2007}%
  \BibitemOpen
  \bibfield  {author} {\bibinfo {author} {\bibfnamefont {M.}~\bibnamefont
  {Navascu\'es}}, \bibinfo {author} {\bibfnamefont {S.}~\bibnamefont
  {Pironio}},\ and\ \bibinfo {author} {\bibfnamefont {A.}~\bibnamefont
  {Ac\'{\i}n}},\ }\bibfield  {title} {\bibinfo {title} {Bounding the set of
  quantum correlations},\ }\href
  {https://doi.org/10.1103/PhysRevLett.98.010401} {\bibfield  {journal}
  {\bibinfo  {journal} {Phys. Rev. Lett.}\ }\textbf {\bibinfo {volume} {98}},\
  \bibinfo {pages} {010401} (\bibinfo {year} {2007})}\BibitemShut {NoStop}%
\bibitem [{\citenamefont {Barthel}\ and\ \citenamefont
  {H\"ubener}(2012)}]{Barthel.Huebener2012}%
  \BibitemOpen
  \bibfield  {author} {\bibinfo {author} {\bibfnamefont {T.}~\bibnamefont
  {Barthel}}\ and\ \bibinfo {author} {\bibfnamefont {R.}~\bibnamefont
  {H\"ubener}},\ }\bibfield  {title} {\bibinfo {title} {Solving
  condensed-matter ground-state problems by semidefinite relaxations},\ }\href
  {https://doi.org/10.1103/PhysRevLett.108.200404} {\bibfield  {journal}
  {\bibinfo  {journal} {Phys. Rev. Lett.}\ }\textbf {\bibinfo {volume} {108}},\
  \bibinfo {pages} {200404} (\bibinfo {year} {2012})}\BibitemShut {NoStop}%
\bibitem [{\citenamefont {{Simmons-Duffin}}(2015)}]{SimmonsDuffin2015}%
  \BibitemOpen
  \bibfield  {author} {\bibinfo {author} {\bibfnamefont {D.}~\bibnamefont
  {{Simmons-Duffin}}},\ }\bibfield  {title} {\bibinfo {title} {{A semidefinite
  program solver for the conformal bootstrap}},\ }\href
  {https://doi.org/10.1007/JHEP06(2015)174} {\bibfield  {journal} {\bibinfo
  {journal} {J. High Energy Phys.}\ }\textbf {\bibinfo {volume} {2015}},\
  \bibinfo {eid} {174}}\BibitemShut {NoStop}%
\bibitem [{\citenamefont {Lov{\'a}sz}(1979)}]{Lovasz1979}%
  \BibitemOpen
  \bibfield  {author} {\bibinfo {author} {\bibfnamefont {L.}~\bibnamefont
  {Lov{\'a}sz}},\ }\bibfield  {title} {\bibinfo {title} {On the shannon
  capacity of a graph},\ }\href {https://doi.org/10.1109/TIT.1979.1055985}
  {\bibfield  {journal} {\bibinfo  {journal} {IEEE Trans. Inf. Theory}\
  }\textbf {\bibinfo {volume} {25}},\ \bibinfo {pages} {1} (\bibinfo {year}
  {1979})}\BibitemShut {NoStop}%
\bibitem [{\citenamefont {Lasserre}(2001)}]{Lasserre2001}%
  \BibitemOpen
  \bibfield  {author} {\bibinfo {author} {\bibfnamefont {J.~B.}\ \bibnamefont
  {Lasserre}},\ }\bibfield  {title} {\bibinfo {title} {Global optimization with
  polynomials and the problem of moments},\ }\href
  {https://doi.org/10.1137/S1052623400366802} {\bibfield  {journal} {\bibinfo
  {journal} {SIAM J. Optim.}\ }\textbf {\bibinfo {volume} {11}},\ \bibinfo
  {pages} {796} (\bibinfo {year} {2001})}\BibitemShut {NoStop}%
\bibitem [{\citenamefont {Parrilo}(2000)}]{Parrilo2000}%
  \BibitemOpen
  \bibfield  {author} {\bibinfo {author} {\bibfnamefont {P.~A.}\ \bibnamefont
  {Parrilo}},\ }\emph {\bibinfo {title} {Structured semidefinite programs and
  semialgebraic geometry methods in robustness and optimization}},\ \href
  {https://doi.org/10.7907/2K6Y-CH43} {Ph.D. thesis},\ \bibinfo  {school}
  {California Institute of Technology} (\bibinfo {year} {2000})\BibitemShut
  {NoStop}%
\bibitem [{\citenamefont {{Navascu{\'e}s}}\ \emph {et~al.}(2008)\citenamefont
  {{Navascu{\'e}s}}, \citenamefont {{Pironio}},\ and\ \citenamefont
  {{Ac{\'\i}n}}}]{Navascues.etal2008}%
  \BibitemOpen
  \bibfield  {author} {\bibinfo {author} {\bibfnamefont {M.}~\bibnamefont
  {{Navascu{\'e}s}}}, \bibinfo {author} {\bibfnamefont {S.}~\bibnamefont
  {{Pironio}}},\ and\ \bibinfo {author} {\bibfnamefont {A.}~\bibnamefont
  {{Ac{\'\i}n}}},\ }\bibfield  {title} {\bibinfo {title} {A convergent
  hierarchy of semidefinite programs characterizing the set of quantum
  correlations},\ }\href {https://doi.org/10.1088/1367-2630/10/7/073013}
  {\bibfield  {journal} {\bibinfo  {journal} {New J. Phys.}\ }\textbf {\bibinfo
  {volume} {10}},\ \bibinfo {eid} {073013} (\bibinfo {year}
  {2008})}\BibitemShut {NoStop}%
\bibitem [{\citenamefont {Navascu\'es}\ and\ \citenamefont
  {V\'ertesi}(2015)}]{Navascues.Vertesi2015}%
  \BibitemOpen
  \bibfield  {author} {\bibinfo {author} {\bibfnamefont {M.}~\bibnamefont
  {Navascu\'es}}\ and\ \bibinfo {author} {\bibfnamefont {T.}~\bibnamefont
  {V\'ertesi}},\ }\bibfield  {title} {\bibinfo {title} {Bounding the set of
  finite dimensional quantum correlations},\ }\href
  {https://doi.org/10.1103/PhysRevLett.115.020501} {\bibfield  {journal}
  {\bibinfo  {journal} {Phys. Rev. Lett.}\ }\textbf {\bibinfo {volume} {115}},\
  \bibinfo {pages} {020501} (\bibinfo {year} {2015})}\BibitemShut {NoStop}%
\bibitem [{\citenamefont {G\"uhne}\ \emph {et~al.}(2021)\citenamefont
  {G\"uhne}, \citenamefont {Mao},\ and\ \citenamefont {Yu}}]{Guehne.etal2021}%
  \BibitemOpen
  \bibfield  {author} {\bibinfo {author} {\bibfnamefont {O.}~\bibnamefont
  {G\"uhne}}, \bibinfo {author} {\bibfnamefont {Y.}~\bibnamefont {Mao}},\ and\
  \bibinfo {author} {\bibfnamefont {X.-D.}\ \bibnamefont {Yu}},\ }\bibfield
  {title} {\bibinfo {title} {Geometry of faithful entanglement},\ }\href
  {https://doi.org/10.1103/PhysRevLett.126.140503} {\bibfield  {journal}
  {\bibinfo  {journal} {Phys. Rev. Lett.}\ }\textbf {\bibinfo {volume} {126}},\
  \bibinfo {pages} {140503} (\bibinfo {year} {2021})}\BibitemShut {NoStop}%
\bibitem [{\citenamefont {Barahona}\ \emph {et~al.}(1988)\citenamefont
  {Barahona}, \citenamefont {Gr\"otschel}, \citenamefont {J\"unger},\ and\
  \citenamefont {Reinelt}}]{Barahona.etal1988}%
  \BibitemOpen
  \bibfield  {author} {\bibinfo {author} {\bibfnamefont {F.}~\bibnamefont
  {Barahona}}, \bibinfo {author} {\bibfnamefont {M.}~\bibnamefont
  {Gr\"otschel}}, \bibinfo {author} {\bibfnamefont {M.}~\bibnamefont
  {J\"unger}},\ and\ \bibinfo {author} {\bibfnamefont {G.}~\bibnamefont
  {Reinelt}},\ }\bibfield  {title} {\bibinfo {title} {An application of
  combinatorial optimization to statistical physics and circuit layout
  design},\ }\href {https://doi.org/10.1287/opre.36.3.493} {\bibfield
  {journal} {\bibinfo  {journal} {Oper. Res.}\ }\textbf {\bibinfo {volume}
  {36}},\ \bibinfo {pages} {493} (\bibinfo {year} {1988})}\BibitemShut
  {NoStop}%
\bibitem [{\citenamefont {Gross}\ \emph {et~al.}(2010)\citenamefont {Gross},
  \citenamefont {Liu}, \citenamefont {Flammia}, \citenamefont {Becker},\ and\
  \citenamefont {Eisert}}]{Gross.etal2010}%
  \BibitemOpen
  \bibfield  {author} {\bibinfo {author} {\bibfnamefont {D.}~\bibnamefont
  {Gross}}, \bibinfo {author} {\bibfnamefont {Y.-K.}\ \bibnamefont {Liu}},
  \bibinfo {author} {\bibfnamefont {S.~T.}\ \bibnamefont {Flammia}}, \bibinfo
  {author} {\bibfnamefont {S.}~\bibnamefont {Becker}},\ and\ \bibinfo {author}
  {\bibfnamefont {J.}~\bibnamefont {Eisert}},\ }\bibfield  {title} {\bibinfo
  {title} {Quantum state tomography via compressed sensing},\ }\href
  {https://doi.org/10.1103/PhysRevLett.105.150401} {\bibfield  {journal}
  {\bibinfo  {journal} {Phys. Rev. Lett.}\ }\textbf {\bibinfo {volume} {105}},\
  \bibinfo {pages} {150401} (\bibinfo {year} {2010})}\BibitemShut {NoStop}%
\bibitem [{\citenamefont {Markovsky}(2019)}]{Markovsky2019}%
  \BibitemOpen
  \bibfield  {author} {\bibinfo {author} {\bibfnamefont {I.}~\bibnamefont
  {Markovsky}},\ }\href@noop {} {\emph {\bibinfo {title} {Low-Rank
  Approximation: Algorithms, Implementation, Applications}}}\ (\bibinfo
  {publisher} {Springer},\ \bibinfo {year} {2019})\BibitemShut {NoStop}%
\bibitem [{\citenamefont {Orsi}\ \emph {et~al.}(2006)\citenamefont {Orsi},
  \citenamefont {Helmke},\ and\ \citenamefont {Moore}}]{Orsi.etal2006}%
  \BibitemOpen
  \bibfield  {author} {\bibinfo {author} {\bibfnamefont {R.}~\bibnamefont
  {Orsi}}, \bibinfo {author} {\bibfnamefont {U.}~\bibnamefont {Helmke}},\ and\
  \bibinfo {author} {\bibfnamefont {J.~B.}\ \bibnamefont {Moore}},\ }\bibfield
  {title} {\bibinfo {title} {A newton-like method for solving rank constrained
  linear matrix inequalities},\ }\href
  {https://doi.org/10.1016/j.automatica.2006.05.026} {\bibfield  {journal}
  {\bibinfo  {journal} {Automatica}\ }\textbf {\bibinfo {volume} {42}},\
  \bibinfo {pages} {1875} (\bibinfo {year} {2006})}\BibitemShut {NoStop}%
\bibitem [{\citenamefont {Sun}\ and\ \citenamefont {Dai}(2017)}]{Sun.Dai2017}%
  \BibitemOpen
  \bibfield  {author} {\bibinfo {author} {\bibfnamefont {C.}~\bibnamefont
  {Sun}}\ and\ \bibinfo {author} {\bibfnamefont {R.}~\bibnamefont {Dai}},\
  }\bibfield  {title} {\bibinfo {title} {Rank-constrained optimization and its
  applications},\ }\href
  {https://doi.org/https://doi.org/10.1016/j.automatica.2017.04.039} {\bibfield
   {journal} {\bibinfo  {journal} {Automatica}\ }\textbf {\bibinfo {volume}
  {82}},\ \bibinfo {pages} {128} (\bibinfo {year} {2017})}\BibitemShut
  {NoStop}%
\bibitem [{\citenamefont {Weilenmann}\ \emph {et~al.}(2020)\citenamefont
  {Weilenmann}, \citenamefont {Dive}, \citenamefont {Trillo}, \citenamefont
  {Aguilar},\ and\ \citenamefont {Navascu\'es}}]{Weilenmann.etal2020}%
  \BibitemOpen
  \bibfield  {author} {\bibinfo {author} {\bibfnamefont {M.}~\bibnamefont
  {Weilenmann}}, \bibinfo {author} {\bibfnamefont {B.}~\bibnamefont {Dive}},
  \bibinfo {author} {\bibfnamefont {D.}~\bibnamefont {Trillo}}, \bibinfo
  {author} {\bibfnamefont {E.~A.}\ \bibnamefont {Aguilar}},\ and\ \bibinfo
  {author} {\bibfnamefont {M.}~\bibnamefont {Navascu\'es}},\ }\bibfield
  {title} {\bibinfo {title} {Entanglement detection beyond measuring
  fidelities},\ }\href {https://doi.org/10.1103/PhysRevLett.124.200502}
  {\bibfield  {journal} {\bibinfo  {journal} {Phys. Rev. Lett.}\ }\textbf
  {\bibinfo {volume} {124}},\ \bibinfo {pages} {200502} (\bibinfo {year}
  {2020})},\ \bibinfo {note} {{Erratum:
  \href{https://doi.org/10.1103/PhysRevLett.125.159903}{Phys. Rev. Lett.
  \textbf{125}, 159903(E) (2020)}}}\BibitemShut {NoStop}%
\bibitem [{\citenamefont {Alberti}\ and\ \citenamefont
  {Uhlmann}(1982)}]{Alberti.Uhlmann1982}%
  \BibitemOpen
  \bibfield  {author} {\bibinfo {author} {\bibfnamefont {P.~M.}\ \bibnamefont
  {Alberti}}\ and\ \bibinfo {author} {\bibfnamefont {A.}~\bibnamefont
  {Uhlmann}},\ }\href@noop {} {\emph {\bibinfo {title} {Stochasticity and
  partial order}}}\ (\bibinfo  {publisher} {Deutscher Verlag der
  Wissenschaften, Berlin},\ \bibinfo {year} {1982})\BibitemShut {NoStop}%
\bibitem [{\citenamefont {Lee}\ and\ \citenamefont
  {Watrous}(2020)}]{Lee.Watrous2020}%
  \BibitemOpen
  \bibfield  {author} {\bibinfo {author} {\bibfnamefont {C.~D.-Y.}\
  \bibnamefont {Lee}}\ and\ \bibinfo {author} {\bibfnamefont {J.}~\bibnamefont
  {Watrous}},\ }\bibfield  {title} {\bibinfo {title} {Detecting mixed-unitary
  quantum channels is {NP}-hard},\ }\href
  {https://doi.org/10.22331/q-2020-04-16-253} {\bibfield  {journal} {\bibinfo
  {journal} {{Quantum}}\ }\textbf {\bibinfo {volume} {4}},\ \bibinfo {pages}
  {253} (\bibinfo {year} {2020})}\BibitemShut {NoStop}%
\bibitem [{\citenamefont {Lov{\'a}sz}(2019)}]{Lovasz2019}%
  \BibitemOpen
  \bibfield  {author} {\bibinfo {author} {\bibfnamefont {L.}~\bibnamefont
  {Lov{\'a}sz}},\ }\href@noop {} {\emph {\bibinfo {title} {Graphs and
  geometry}}}\ (\bibinfo  {publisher} {American Mathematical Society},\
  \bibinfo {year} {2019})\BibitemShut {NoStop}%
\bibitem [{\citenamefont {Cabello}\ \emph {et~al.}(2014)\citenamefont
  {Cabello}, \citenamefont {Severini},\ and\ \citenamefont
  {Winter}}]{Cabello.etal2014}%
  \BibitemOpen
  \bibfield  {author} {\bibinfo {author} {\bibfnamefont {A.}~\bibnamefont
  {Cabello}}, \bibinfo {author} {\bibfnamefont {S.}~\bibnamefont {Severini}},\
  and\ \bibinfo {author} {\bibfnamefont {A.}~\bibnamefont {Winter}},\
  }\bibfield  {title} {\bibinfo {title} {Graph-theoretic approach to quantum
  correlations},\ }\href {https://doi.org/10.1103/PhysRevLett.112.040401}
  {\bibfield  {journal} {\bibinfo  {journal} {Phys. Rev. Lett.}\ }\textbf
  {\bibinfo {volume} {112}},\ \bibinfo {pages} {040401} (\bibinfo {year}
  {2014})}\BibitemShut {NoStop}%
\bibitem [{\citenamefont {Ramanathan}\ and\ \citenamefont
  {Horodecki}(2014)}]{Ramanathan.Horodecki2014}%
  \BibitemOpen
  \bibfield  {author} {\bibinfo {author} {\bibfnamefont {R.}~\bibnamefont
  {Ramanathan}}\ and\ \bibinfo {author} {\bibfnamefont {P.}~\bibnamefont
  {Horodecki}},\ }\bibfield  {title} {\bibinfo {title} {Necessary and
  sufficient condition for state-independent contextual measurement
  scenarios},\ }\href {https://doi.org/10.1103/PhysRevLett.112.040404}
  {\bibfield  {journal} {\bibinfo  {journal} {Phys. Rev. Lett.}\ }\textbf
  {\bibinfo {volume} {112}},\ \bibinfo {pages} {040404} (\bibinfo {year}
  {2014})}\BibitemShut {NoStop}%
\bibitem [{\citenamefont {Goemans}\ and\ \citenamefont
  {Williamson}(1995)}]{Goemans.Williamson1995}%
  \BibitemOpen
  \bibfield  {author} {\bibinfo {author} {\bibfnamefont {M.~X.}\ \bibnamefont
  {Goemans}}\ and\ \bibinfo {author} {\bibfnamefont {D.~P.}\ \bibnamefont
  {Williamson}},\ }\bibfield  {title} {\bibinfo {title} {Improved approximation
  algorithms for maximum cut and satisfiability problems using semidefinite
  programming},\ }\href {https://doi.org/10.1145/227683.227684} {\bibfield
  {journal} {\bibinfo  {journal} {JACM}\ }\textbf {\bibinfo {volume} {42}},\
  \bibinfo {pages} {1115} (\bibinfo {year} {1995})}\BibitemShut {NoStop}%
\bibitem [{\citenamefont {Boros}\ and\ \citenamefont
  {Hammer}(2002)}]{Boros.Hammer2002}%
  \BibitemOpen
  \bibfield  {author} {\bibinfo {author} {\bibfnamefont {E.}~\bibnamefont
  {Boros}}\ and\ \bibinfo {author} {\bibfnamefont {P.~L.}\ \bibnamefont
  {Hammer}},\ }\bibfield  {title} {\bibinfo {title} {Pseudo-boolean
  optimization},\ }\href {https://doi.org/10.1016/S0166-218X(01)00341-9}
  {\bibfield  {journal} {\bibinfo  {journal} {Discrete Appl. Math.}\ }\textbf
  {\bibinfo {volume} {123}},\ \bibinfo {pages} {155} (\bibinfo {year}
  {2002})}\BibitemShut {NoStop}%
\bibitem [{\citenamefont {Kirkpatrick}\ \emph {et~al.}(1983)\citenamefont
  {Kirkpatrick}, \citenamefont {Gelatt},\ and\ \citenamefont
  {Vecchi}}]{Kirkpatrick.etal1983}%
  \BibitemOpen
  \bibfield  {author} {\bibinfo {author} {\bibfnamefont {S.}~\bibnamefont
  {Kirkpatrick}}, \bibinfo {author} {\bibfnamefont {C.~D.}\ \bibnamefont
  {Gelatt}},\ and\ \bibinfo {author} {\bibfnamefont {M.~P.}\ \bibnamefont
  {Vecchi}},\ }\bibfield  {title} {\bibinfo {title} {Optimization by simulated
  annealing},\ }\href {https://doi.org/10.1126/science.220.4598.671} {\bibfield
   {journal} {\bibinfo  {journal} {Science}\ }\textbf {\bibinfo {volume}
  {220}},\ \bibinfo {pages} {671} (\bibinfo {year} {1983})}\BibitemShut
  {NoStop}%
\bibitem [{\citenamefont {Dorogovtsev}\ \emph {et~al.}(2008)\citenamefont
  {Dorogovtsev}, \citenamefont {Goltsev},\ and\ \citenamefont
  {Mendes}}]{Dorogovtsev.etal2008}%
  \BibitemOpen
  \bibfield  {author} {\bibinfo {author} {\bibfnamefont {S.~N.}\ \bibnamefont
  {Dorogovtsev}}, \bibinfo {author} {\bibfnamefont {A.~V.}\ \bibnamefont
  {Goltsev}},\ and\ \bibinfo {author} {\bibfnamefont {J.~F.~F.}\ \bibnamefont
  {Mendes}},\ }\bibfield  {title} {\bibinfo {title} {Critical phenomena in
  complex networks},\ }\href {https://doi.org/10.1103/RevModPhys.80.1275}
  {\bibfield  {journal} {\bibinfo  {journal} {Rev. Mod. Phys.}\ }\textbf
  {\bibinfo {volume} {80}},\ \bibinfo {pages} {1275} (\bibinfo {year}
  {2008})}\BibitemShut {NoStop}%
\bibitem [{\citenamefont {Lucas}(2014)}]{Lucas2014}%
  \BibitemOpen
  \bibfield  {author} {\bibinfo {author} {\bibfnamefont {A.}~\bibnamefont
  {Lucas}},\ }\bibfield  {title} {\bibinfo {title} {Ising formulations of many
  {NP} problems},\ }\href {https://doi.org/10.3389/fphy.2014.00005} {\bibfield
  {journal} {\bibinfo  {journal} {Front. Phys.}\ }\textbf {\bibinfo {volume}
  {2}},\ \bibinfo {pages} {5} (\bibinfo {year} {2014})}\BibitemShut {NoStop}%
\bibitem [{\citenamefont {{Boixo}}\ \emph {et~al.}(2014)\citenamefont
  {{Boixo}}, \citenamefont {{R{\o}nnow}}, \citenamefont {{Isakov}},
  \citenamefont {{Wang}}, \citenamefont {{Wecker}}, \citenamefont {{Lidar}},
  \citenamefont {{Martinis}},\ and\ \citenamefont {{Troyer}}}]{Boixo.etal2014}%
  \BibitemOpen
  \bibfield  {author} {\bibinfo {author} {\bibfnamefont {S.}~\bibnamefont
  {{Boixo}}}, \bibinfo {author} {\bibfnamefont {T.~F.}\ \bibnamefont
  {{R{\o}nnow}}}, \bibinfo {author} {\bibfnamefont {S.~V.}\ \bibnamefont
  {{Isakov}}}, \bibinfo {author} {\bibfnamefont {Z.}~\bibnamefont {{Wang}}},
  \bibinfo {author} {\bibfnamefont {D.}~\bibnamefont {{Wecker}}}, \bibinfo
  {author} {\bibfnamefont {D.~A.}\ \bibnamefont {{Lidar}}}, \bibinfo {author}
  {\bibfnamefont {J.~M.}\ \bibnamefont {{Martinis}}},\ and\ \bibinfo {author}
  {\bibfnamefont {M.}~\bibnamefont {{Troyer}}},\ }\bibfield  {title} {\bibinfo
  {title} {Evidence for quantum annealing with more than one hundred qubits},\
  }\href {https://doi.org/10.1038/nphys2900} {\bibfield  {journal} {\bibinfo
  {journal} {Nat. Phys.}\ }\textbf {\bibinfo {volume} {10}},\ \bibinfo {pages}
  {218} (\bibinfo {year} {2014})}\BibitemShut {NoStop}%
\bibitem [{\citenamefont {Farhi}\ \emph {et~al.}()\citenamefont {Farhi},
  \citenamefont {Goldstone},\ and\ \citenamefont {Gutmann}}]{Farhi.etal2014}%
  \BibitemOpen
  \bibfield  {author} {\bibinfo {author} {\bibfnamefont {E.}~\bibnamefont
  {Farhi}}, \bibinfo {author} {\bibfnamefont {J.}~\bibnamefont {Goldstone}},\
  and\ \bibinfo {author} {\bibfnamefont {S.}~\bibnamefont {Gutmann}},\
  }\href@noop {} {\bibinfo {title} {A quantum approximate optimization
  algorithm}},\ \Eprint {https://arxiv.org/abs/1411.4028} {arXiv:1411.4028}
  \BibitemShut {NoStop}%
\bibitem [{\citenamefont {Farhi}\ and\ \citenamefont
  {Harrow}(2019)}]{Farhi.Harrow2019}%
  \BibitemOpen
  \bibfield  {author} {\bibinfo {author} {\bibfnamefont {E.}~\bibnamefont
  {Farhi}}\ and\ \bibinfo {author} {\bibfnamefont {A.~W.}\ \bibnamefont
  {Harrow}},\ }\href@noop {} {\bibinfo {title} {Quantum supremacy through the
  quantum approximate optimization algorithm}} (\bibinfo {year} {2019}),\
  \Eprint {https://arxiv.org/abs/1602.07674} {arXiv:1602.07674} \BibitemShut
  {NoStop}%
\bibitem [{\citenamefont {Boyd}\ and\ \citenamefont
  {Vandenberghe}(2004)}]{Boyd.Vandenberghe2004}%
  \BibitemOpen
  \bibfield  {author} {\bibinfo {author} {\bibfnamefont {S.}~\bibnamefont
  {Boyd}}\ and\ \bibinfo {author} {\bibfnamefont {L.}~\bibnamefont
  {Vandenberghe}},\ }\href@noop {} {\emph {\bibinfo {title} {Convex
  optimization}}}\ (\bibinfo  {publisher} {Cambridge University Press, New
  York},\ \bibinfo {year} {2004})\BibitemShut {NoStop}%
\bibitem [{\citenamefont {Vallentin}(2009)}]{Vallentin2009}%
  \BibitemOpen
  \bibfield  {author} {\bibinfo {author} {\bibfnamefont {F.}~\bibnamefont
  {Vallentin}},\ }\bibfield  {title} {\bibinfo {title} {Symmetry in
  semidefinite programs},\ }\href {https://doi.org/10.1016/j.laa.2008.07.025}
  {\bibfield  {journal} {\bibinfo  {journal} {Linear Algebra Appl.}\ }\textbf
  {\bibinfo {volume} {430}},\ \bibinfo {pages} {360} (\bibinfo {year}
  {2009})}\BibitemShut {NoStop}%
\bibitem [{\citenamefont {{Rosset}}()}]{Rosset2018}%
  \BibitemOpen
  \bibfield  {author} {\bibinfo {author} {\bibfnamefont {D.}~\bibnamefont
  {{Rosset}}},\ }\href@noop {} {\bibinfo {title} {{SymDPoly}: symmetry-adapted
  moment relaxations for noncommutative polynomial optimization}},\ \Eprint
  {https://arxiv.org/abs/1808.09598} {arXiv:1808.09598} \BibitemShut {NoStop}%
\bibitem [{\citenamefont {Aguilar}\ \emph {et~al.}(2018)\citenamefont
  {Aguilar}, \citenamefont {Borka\l{}a}, \citenamefont {Mironowicz},\ and\
  \citenamefont {Paw\l{}owski}}]{Aguilar.etal2018}%
  \BibitemOpen
  \bibfield  {author} {\bibinfo {author} {\bibfnamefont {E.~A.}\ \bibnamefont
  {Aguilar}}, \bibinfo {author} {\bibfnamefont {J.~J.}\ \bibnamefont
  {Borka\l{}a}}, \bibinfo {author} {\bibfnamefont {P.}~\bibnamefont
  {Mironowicz}},\ and\ \bibinfo {author} {\bibfnamefont {M.}~\bibnamefont
  {Paw\l{}owski}},\ }\bibfield  {title} {\bibinfo {title} {Connections between
  mutually unbiased bases and quantum random access codes},\ }\href
  {https://doi.org/10.1103/PhysRevLett.121.050501} {\bibfield  {journal}
  {\bibinfo  {journal} {Phys. Rev. Lett.}\ }\textbf {\bibinfo {volume} {121}},\
  \bibinfo {pages} {050501} (\bibinfo {year} {2018})}\BibitemShut {NoStop}%
\bibitem [{\citenamefont {Nielsen}\ and\ \citenamefont
  {Chuang}(2000)}]{Nielsen.Chuang2000}%
  \BibitemOpen
  \bibfield  {author} {\bibinfo {author} {\bibfnamefont {M.~A.}\ \bibnamefont
  {Nielsen}}\ and\ \bibinfo {author} {\bibfnamefont {I.~L.}\ \bibnamefont
  {Chuang}},\ }\href@noop {} {\emph {\bibinfo {title} {Quantum Computation and
  Quantum Information}}}\ (\bibinfo  {publisher} {Cambridge University Press,
  Cambridge, UK},\ \bibinfo {year} {2000})\BibitemShut {NoStop}%
\bibitem [{\citenamefont {Horodecki}\ \emph {et~al.}(2009)\citenamefont
  {Horodecki}, \citenamefont {Horodecki}, \citenamefont {Horodecki},\ and\
  \citenamefont {Horodecki}}]{Horodecki.etal2009}%
  \BibitemOpen
  \bibfield  {author} {\bibinfo {author} {\bibfnamefont {R.}~\bibnamefont
  {Horodecki}}, \bibinfo {author} {\bibfnamefont {P.}~\bibnamefont
  {Horodecki}}, \bibinfo {author} {\bibfnamefont {M.}~\bibnamefont
  {Horodecki}},\ and\ \bibinfo {author} {\bibfnamefont {K.}~\bibnamefont
  {Horodecki}},\ }\bibfield  {title} {\bibinfo {title} {Quantum entanglement},\
  }\href {https://doi.org/10.1103/RevModPhys.81.865} {\bibfield  {journal}
  {\bibinfo  {journal} {Rev. Mod. Phys.}\ }\textbf {\bibinfo {volume} {81}},\
  \bibinfo {pages} {865} (\bibinfo {year} {2009})}\BibitemShut {NoStop}%
\bibitem [{\citenamefont {G{\"u}hne}\ and\ \citenamefont
  {T{\'o}th}(2009)}]{Guehne.Toth2009}%
  \BibitemOpen
  \bibfield  {author} {\bibinfo {author} {\bibfnamefont {O.}~\bibnamefont
  {G{\"u}hne}}\ and\ \bibinfo {author} {\bibfnamefont {G.}~\bibnamefont
  {T{\'o}th}},\ }\bibfield  {title} {\bibinfo {title} {Entanglement
  detection},\ }\href {https://doi.org/10.1016/j.physrep.2009.02.004}
  {\bibfield  {journal} {\bibinfo  {journal} {Phys. Rep.}\ }\textbf {\bibinfo
  {volume} {474}},\ \bibinfo {pages} {1} (\bibinfo {year} {2009})}\BibitemShut
  {NoStop}%
\bibitem [{\citenamefont {T\'oth}\ and\ \citenamefont
  {G\"uhne}(2009)}]{Toth.Guehne2009}%
  \BibitemOpen
  \bibfield  {author} {\bibinfo {author} {\bibfnamefont {G.}~\bibnamefont
  {T\'oth}}\ and\ \bibinfo {author} {\bibfnamefont {O.}~\bibnamefont
  {G\"uhne}},\ }\bibfield  {title} {\bibinfo {title} {Entanglement and
  permutational symmetry},\ }\href
  {https://doi.org/10.1103/PhysRevLett.102.170503} {\bibfield  {journal}
  {\bibinfo  {journal} {Phys. Rev. Lett.}\ }\textbf {\bibinfo {volume} {102}},\
  \bibinfo {pages} {170503} (\bibinfo {year} {2009})}\BibitemShut {NoStop}%
\bibitem [{\citenamefont {Gurvits}(2003)}]{Gurvits2003}%
  \BibitemOpen
  \bibfield  {author} {\bibinfo {author} {\bibfnamefont {L.}~\bibnamefont
  {Gurvits}},\ }\bibfield  {title} {\bibinfo {title} {Classical deterministic
  complexity of edmonds' problem and quantum entanglement},\ }in\ \href
  {https://doi.org/10.1145/780542.780545} {\emph {\bibinfo {booktitle}
  {Proceedings of the thirty-fifth annual ACM symposium on Theory of
  computing}}},\ \bibinfo {series and number} {STOC '03}\ (\bibinfo
  {publisher} {ACM},\ \bibinfo {address} {New York, NY, USA},\ \bibinfo {year}
  {2003})\ pp.\ \bibinfo {pages} {10--19}\BibitemShut {NoStop}%
\bibitem [{\citenamefont {Peres}(1996)}]{Peres1996}%
  \BibitemOpen
  \bibfield  {author} {\bibinfo {author} {\bibfnamefont {A.}~\bibnamefont
  {Peres}},\ }\bibfield  {title} {\bibinfo {title} {Separability criterion for
  density matrices},\ }\href {https://doi.org/10.1103/PhysRevLett.77.1413}
  {\bibfield  {journal} {\bibinfo  {journal} {Phys. Rev. Lett.}\ }\textbf
  {\bibinfo {volume} {77}},\ \bibinfo {pages} {1413} (\bibinfo {year}
  {1996})}\BibitemShut {NoStop}%
\bibitem [{\citenamefont {Horodecki}\ \emph {et~al.}(1996)\citenamefont
  {Horodecki}, \citenamefont {Horodecki},\ and\ \citenamefont
  {Horodecki}}]{Horodecki.etal1996}%
  \BibitemOpen
  \bibfield  {author} {\bibinfo {author} {\bibfnamefont {M.}~\bibnamefont
  {Horodecki}}, \bibinfo {author} {\bibfnamefont {P.}~\bibnamefont
  {Horodecki}},\ and\ \bibinfo {author} {\bibfnamefont {R.}~\bibnamefont
  {Horodecki}},\ }\bibfield  {title} {\bibinfo {title} {Separability of mixed
  states: {N}ecessary and sufficient conditions},\ }\href
  {https://doi.org/10.1016/S0375-9601(96)00706-2} {\bibfield  {journal}
  {\bibinfo  {journal} {Phys. Lett. A}\ }\textbf {\bibinfo {volume} {223}},\
  \bibinfo {pages} {1} (\bibinfo {year} {1996})}\BibitemShut {NoStop}%
\bibitem [{\citenamefont {Werner}(1989)}]{Werner1989b}%
  \BibitemOpen
  \bibfield  {author} {\bibinfo {author} {\bibfnamefont {R.~F.}\ \bibnamefont
  {Werner}},\ }\bibfield  {title} {\bibinfo {title} {An application of
  {{Bell}}'s inequalities to a quantum state extension problem},\ }\href
  {https://doi.org/10.1007/BF00399761} {\bibfield  {journal} {\bibinfo
  {journal} {Lett. Math. Phys.}\ }\textbf {\bibinfo {volume} {17}},\ \bibinfo
  {pages} {359} (\bibinfo {year} {1989})}\BibitemShut {NoStop}%
\bibitem [{\citenamefont {{Doherty}}\ and\ \citenamefont
  {{Wehner}}()}]{Doherty.Wehner2012}%
  \BibitemOpen
  \bibfield  {author} {\bibinfo {author} {\bibfnamefont {A.~C.}\ \bibnamefont
  {{Doherty}}}\ and\ \bibinfo {author} {\bibfnamefont {S.}~\bibnamefont
  {{Wehner}}},\ }\href@noop {} {\bibinfo {title} {Convergence of {SDP}
  hierarchies for polynomial optimization on the hypersphere}},\ \Eprint
  {https://arxiv.org/abs/1210.5048} {arXiv:1210.5048} \BibitemShut {NoStop}%
\bibitem [{\citenamefont {Doherty}\ \emph {et~al.}(2004)\citenamefont
  {Doherty}, \citenamefont {Parrilo},\ and\ \citenamefont
  {Spedalieri}}]{Doherty.etal2004}%
  \BibitemOpen
  \bibfield  {author} {\bibinfo {author} {\bibfnamefont {A.~C.}\ \bibnamefont
  {Doherty}}, \bibinfo {author} {\bibfnamefont {P.~A.}\ \bibnamefont
  {Parrilo}},\ and\ \bibinfo {author} {\bibfnamefont {F.~M.}\ \bibnamefont
  {Spedalieri}},\ }\bibfield  {title} {\bibinfo {title} {Complete family of
  separability criteria},\ }\href {https://doi.org/10.1103/PhysRevA.69.022308}
  {\bibfield  {journal} {\bibinfo  {journal} {Phys. Rev. A}\ }\textbf {\bibinfo
  {volume} {69}},\ \bibinfo {pages} {022308} (\bibinfo {year}
  {2004})}\BibitemShut {NoStop}%
\bibitem [{Note1()}]{Note1}%
  \BibitemOpen
  \bibinfo {note} {An explicit counterexample is the (unnormalized) state $\Phi
  _{AB}=\protect \mathds {1}_{AB}+V_{AB}$.}\BibitemShut {Stop}%
\bibitem [{\citenamefont {Goodman}\ and\ \citenamefont
  {Wallach}(2009)}]{Goodman.Wallach2009}%
  \BibitemOpen
  \bibfield  {author} {\bibinfo {author} {\bibfnamefont {R.}~\bibnamefont
  {Goodman}}\ and\ \bibinfo {author} {\bibfnamefont {N.~R.}\ \bibnamefont
  {Wallach}},\ }\href@noop {} {\emph {\bibinfo {title} {Symmetry,
  representations, and invariants}}},\ Vol.\ \bibinfo {volume} {255}\ (\bibinfo
   {publisher} {Springer},\ \bibinfo {year} {2009})\BibitemShut {NoStop}%
\bibitem [{\citenamefont {Watrous}(2018)}]{Watrous2018}%
  \BibitemOpen
  \bibfield  {author} {\bibinfo {author} {\bibfnamefont {J.}~\bibnamefont
  {Watrous}},\ }\href@noop {} {\emph {\bibinfo {title} {The Theory of Quantum
  Information}}}\ (\bibinfo  {publisher} {Cambridge University Press,
  Cambridge, UK},\ \bibinfo {year} {2018})\BibitemShut {NoStop}%
\bibitem [{\citenamefont {Ray}\ \emph {et~al.}(2021)\citenamefont {Ray},
  \citenamefont {Boddu}, \citenamefont {Bharti}, \citenamefont {Kwek},\ and\
  \citenamefont {Cabello}}]{Ray.etal2021}%
  \BibitemOpen
  \bibfield  {author} {\bibinfo {author} {\bibfnamefont {M.}~\bibnamefont
  {Ray}}, \bibinfo {author} {\bibfnamefont {N.~G.}\ \bibnamefont {Boddu}},
  \bibinfo {author} {\bibfnamefont {K.}~\bibnamefont {Bharti}}, \bibinfo
  {author} {\bibfnamefont {L.-C.}\ \bibnamefont {Kwek}},\ and\ \bibinfo
  {author} {\bibfnamefont {A.}~\bibnamefont {Cabello}},\ }\bibfield  {title}
  {\bibinfo {title} {Graph-theoretic approach to dimension witnessing},\ }\href
  {https://doi.org/10.1088/1367-2630/abcacd} {\bibfield  {journal} {\bibinfo
  {journal} {New J. Phys.}\ }\textbf {\bibinfo {volume} {23}},\ \bibinfo
  {pages} {033006} (\bibinfo {year} {2021})}\BibitemShut {NoStop}%
\bibitem [{Note2()}]{Note2}%
  \BibitemOpen
  \bibinfo {note} {The numerical error can be shown to be smaller than
  $10^{-100}$ using the standard primal and dual problem of the Lov\'asz
  $\vartheta $-function's SDP characterization \cite {Lovasz1979} and the
  arbitrary-precision SDP solver SDPA-GMP \cite {Nakata2010}}\BibitemShut
  {NoStop}%
\bibitem [{\citenamefont {Hamerly}\ \emph {et~al.}(2019)\citenamefont {Hamerly}
  \emph {et~al.}}]{Hamerly.etal2019}%
  \BibitemOpen
  \bibfield  {author} {\bibinfo {author} {\bibfnamefont {R.}~\bibnamefont
  {Hamerly}} \emph {et~al.},\ }\bibfield  {title} {\bibinfo {title}
  {Experimental investigation of performance differences between coherent ising
  machines and a quantum annealer},\ }\href
  {https://doi.org/10.1126/sciadv.aau0823} {\bibfield  {journal} {\bibinfo
  {journal} {Sci. Adv.}\ }\textbf {\bibinfo {volume} {5}},\ \bibinfo {pages}
  {eaau0823} (\bibinfo {year} {2019})}\BibitemShut {NoStop}%
\bibitem [{\citenamefont {{Crooks}}()}]{Crooks2018}%
  \BibitemOpen
  \bibfield  {author} {\bibinfo {author} {\bibfnamefont {G.~E.}\ \bibnamefont
  {{Crooks}}},\ }\href@noop {} {\bibinfo {title} {Performance of the quantum
  approximate optimization algorithm on the maximum cut problem}},\ \Eprint
  {https://arxiv.org/abs/1811.08419} {arXiv:1811.08419} \BibitemShut {NoStop}%
\bibitem [{Note3()}]{Note3}%
  \BibitemOpen
  \bibinfo {note} {To obtain $\xi _2$, we used the SCS solver (unparallel
  version) on a single cluster node. The solver took $1.33\times 10^5$ and
  $3.67\times 10^5$ seconds for solving the corresponding SDPs of the
  $64$-vertex and $72$-vertex graphs, respectively.}\BibitemShut {Stop}%
\bibitem [{\citenamefont {Nesterov}(1998)}]{Nesterov1998}%
  \BibitemOpen
  \bibfield  {author} {\bibinfo {author} {\bibfnamefont {Y.}~\bibnamefont
  {Nesterov}},\ }\bibfield  {title} {\bibinfo {title} {Semidefinite relaxation
  and nonconvex quadratic optimization},\ }\href
  {https://doi.org/10.1080/10556789808805690} {\bibfield  {journal} {\bibinfo
  {journal} {Optim. Method. Softw.}\ }\textbf {\bibinfo {volume} {9}},\
  \bibinfo {pages} {141} (\bibinfo {year} {1998})}\BibitemShut {NoStop}%
\bibitem [{\citenamefont {Lov\'asz}\ and\ \citenamefont
  {Schrijver}(1991)}]{Lovasz.Schrijver1991}%
  \BibitemOpen
  \bibfield  {author} {\bibinfo {author} {\bibfnamefont {L.}~\bibnamefont
  {Lov\'asz}}\ and\ \bibinfo {author} {\bibfnamefont {A.}~\bibnamefont
  {Schrijver}},\ }\bibfield  {title} {\bibinfo {title} {Cones of matrices and
  set-functions and 0-1 optimization},\ }\href
  {https://doi.org/10.1137/0801013} {\bibfield  {journal} {\bibinfo  {journal}
  {SIAM J. Optim.}\ }\textbf {\bibinfo {volume} {1}},\ \bibinfo {pages} {166}
  (\bibinfo {year} {1991})}\BibitemShut {NoStop}%
\bibitem [{\citenamefont {Fleming}\ \emph {et~al.}(2019)\citenamefont
  {Fleming}, \citenamefont {Kothari},\ and\ \citenamefont
  {Pitassi}}]{Fleming.etal2019}%
  \BibitemOpen
  \bibfield  {author} {\bibinfo {author} {\bibfnamefont {N.}~\bibnamefont
  {Fleming}}, \bibinfo {author} {\bibfnamefont {P.}~\bibnamefont {Kothari}},\
  and\ \bibinfo {author} {\bibfnamefont {T.}~\bibnamefont {Pitassi}},\
  }\bibfield  {title} {\bibinfo {title} {Semialgebraic proofs and efficient
  algorithm design},\ }\href {https://doi.org/10.1561/0400000086} {\bibfield
  {journal} {\bibinfo  {journal} {Found. Trends Theor. Comput. Sci.}\ }\textbf
  {\bibinfo {volume} {14}},\ \bibinfo {pages} {1} (\bibinfo {year}
  {2019})}\BibitemShut {NoStop}%
\bibitem [{\citenamefont {Erdogdu}\ \emph {et~al.}()\citenamefont {Erdogdu},
  \citenamefont {Deshpande},\ and\ \citenamefont
  {Montanari}}]{Erdogdu.etal2017}%
  \BibitemOpen
  \bibfield  {author} {\bibinfo {author} {\bibfnamefont {M.~A.}\ \bibnamefont
  {Erdogdu}}, \bibinfo {author} {\bibfnamefont {Y.}~\bibnamefont {Deshpande}},\
  and\ \bibinfo {author} {\bibfnamefont {A.}~\bibnamefont {Montanari}},\
  }\href@noop {} {\bibinfo {title} {Inference in graphical models via
  semidefinite programming hierarchies}},\ \Eprint
  {https://arxiv.org/abs/1709.06525} {arXiv:1709.06525} \BibitemShut {NoStop}%
\bibitem [{\citenamefont {Bandeira}\ and\ \citenamefont
  {Kunisky}()}]{Bandeira.Kunisky2019}%
  \BibitemOpen
  \bibfield  {author} {\bibinfo {author} {\bibfnamefont {A.~S.}\ \bibnamefont
  {Bandeira}}\ and\ \bibinfo {author} {\bibfnamefont {D.}~\bibnamefont
  {Kunisky}},\ }\href@noop {} {\bibinfo {title} {A gramian description of the
  degree 4 generalized elliptope}},\ \Eprint {https://arxiv.org/abs/1812.11583}
  {arXiv:1812.11583} \BibitemShut {NoStop}%
\bibitem [{\citenamefont {Gatermann}\ and\ \citenamefont
  {Parrilo}(2004)}]{Gatermann.Parrilo2004}%
  \BibitemOpen
  \bibfield  {author} {\bibinfo {author} {\bibfnamefont {K.}~\bibnamefont
  {Gatermann}}\ and\ \bibinfo {author} {\bibfnamefont {P.~A.}\ \bibnamefont
  {Parrilo}},\ }\bibfield  {title} {\bibinfo {title} {Symmetry groups,
  semidefinite programs, and sums of squares},\ }\href
  {https://doi.org/10.1016/j.jpaa.2003.12.011} {\bibfield  {journal} {\bibinfo
  {journal} {J. Pure Appl. Algebra}\ }\textbf {\bibinfo {volume} {192}},\
  \bibinfo {pages} {95} (\bibinfo {year} {2004})}\BibitemShut {NoStop}%
\bibitem [{\citenamefont {Yu}\ \emph {et~al.}(2021)\citenamefont {Yu},
  \citenamefont {Simnacher}, \citenamefont {Wyderka}, \citenamefont {Nguyen},\
  and\ \citenamefont {G{\"u}hne}}]{Yu.etal2021}%
  \BibitemOpen
  \bibfield  {author} {\bibinfo {author} {\bibfnamefont {X.-D.}\ \bibnamefont
  {Yu}}, \bibinfo {author} {\bibfnamefont {T.}~\bibnamefont {Simnacher}},
  \bibinfo {author} {\bibfnamefont {N.}~\bibnamefont {Wyderka}}, \bibinfo
  {author} {\bibfnamefont {H.~C.}\ \bibnamefont {Nguyen}},\ and\ \bibinfo
  {author} {\bibfnamefont {O.}~\bibnamefont {G{\"u}hne}},\ }\bibfield  {title}
  {\bibinfo {title} {A complete hierarchy for the pure state marginal problem
  in quantum mechanics},\ }\href {https://doi.org/10.1038/s41467-020-20799-5}
  {\bibfield  {journal} {\bibinfo  {journal} {Nat. Commun.}\ }\textbf {\bibinfo
  {volume} {12}},\ \bibinfo {pages} {1012} (\bibinfo {year}
  {2021})}\BibitemShut {NoStop}%
\bibitem [{\citenamefont {{Christandl}}\ \emph {et~al.}(2007)\citenamefont
  {{Christandl}}, \citenamefont {{K{\"o}nig}}, \citenamefont {{Mitchison}},\
  and\ \citenamefont {{Renner}}}]{Christandl.etal2007}%
  \BibitemOpen
  \bibfield  {author} {\bibinfo {author} {\bibfnamefont {M.}~\bibnamefont
  {{Christandl}}}, \bibinfo {author} {\bibfnamefont {R.}~\bibnamefont
  {{K{\"o}nig}}}, \bibinfo {author} {\bibfnamefont {G.}~\bibnamefont
  {{Mitchison}}},\ and\ \bibinfo {author} {\bibfnamefont {R.}~\bibnamefont
  {{Renner}}},\ }\bibfield  {title} {\bibinfo {title} {One-and-a-half quantum
  {de Finetti} theorems},\ }\href {https://doi.org/10.1007/s00220-007-0189-3}
  {\bibfield  {journal} {\bibinfo  {journal} {Commun. Math. Phys.}\ }\textbf
  {\bibinfo {volume} {273}},\ \bibinfo {pages} {473} (\bibinfo {year}
  {2007})}\BibitemShut {NoStop}%
\bibitem [{\citenamefont {{Caves}}\ \emph {et~al.}(2002)\citenamefont
  {{Caves}}, \citenamefont {{Fuchs}},\ and\ \citenamefont
  {{Schack}}}]{Caves.etal2002}%
  \BibitemOpen
  \bibfield  {author} {\bibinfo {author} {\bibfnamefont {C.~M.}\ \bibnamefont
  {{Caves}}}, \bibinfo {author} {\bibfnamefont {C.~A.}\ \bibnamefont
  {{Fuchs}}},\ and\ \bibinfo {author} {\bibfnamefont {R.}~\bibnamefont
  {{Schack}}},\ }\bibfield  {title} {\bibinfo {title} {Unknown quantum states:
  The quantum {de Finetti} representation},\ }\href
  {https://doi.org/10.1063/1.1494475} {\bibfield  {journal} {\bibinfo
  {journal} {J. Math. Phys.}\ }\textbf {\bibinfo {volume} {43}},\ \bibinfo
  {pages} {4537} (\bibinfo {year} {2002})}\BibitemShut {NoStop}%
\bibitem [{\citenamefont {Bogachev}(2007)}]{Bogachev2007}%
  \BibitemOpen
  \bibfield  {author} {\bibinfo {author} {\bibfnamefont {V.~I.}\ \bibnamefont
  {Bogachev}},\ }\href@noop {} {\emph {\bibinfo {title} {Measure theory}}},\
  Vol.\ \bibinfo {volume} {1 \& 2}\ (\bibinfo  {publisher} {Springer Science \&
  Business Media},\ \bibinfo {year} {2007})\BibitemShut {NoStop}%
\bibitem [{\citenamefont {Yamashita}\ \emph {et~al.}(2010)\citenamefont
  {Yamashita}, \citenamefont {Fujisawa}, \citenamefont {Nakata}, \citenamefont
  {Nakata}, \citenamefont {Fukuda}, \citenamefont {Kobayashi},\ and\
  \citenamefont {Goto}}]{SDPA7}%
  \BibitemOpen
  \bibfield  {author} {\bibinfo {author} {\bibfnamefont {M.}~\bibnamefont
  {Yamashita}}, \bibinfo {author} {\bibfnamefont {K.}~\bibnamefont {Fujisawa}},
  \bibinfo {author} {\bibfnamefont {K.}~\bibnamefont {Nakata}}, \bibinfo
  {author} {\bibfnamefont {M.}~\bibnamefont {Nakata}}, \bibinfo {author}
  {\bibfnamefont {M.}~\bibnamefont {Fukuda}}, \bibinfo {author} {\bibfnamefont
  {K.}~\bibnamefont {Kobayashi}},\ and\ \bibinfo {author} {\bibfnamefont
  {K.}~\bibnamefont {Goto}},\ }\href@noop {} {\emph {\bibinfo {title} {A
  high-performance software package for semidefinite programs: {SDPA} 7}}},\
  \bibinfo {type} {Tech. Rep.}\ (\bibinfo {year} {2010})\BibitemShut {NoStop}%
\bibitem [{\citenamefont {Diamond}\ and\ \citenamefont
  {Boyd}(2016)}]{Diamond.Boyd2016}%
  \BibitemOpen
  \bibfield  {author} {\bibinfo {author} {\bibfnamefont {S.}~\bibnamefont
  {Diamond}}\ and\ \bibinfo {author} {\bibfnamefont {S.}~\bibnamefont {Boyd}},\
  }\bibfield  {title} {\bibinfo {title} {{CVXPY}: A {Python}-embedded modeling
  language for convex optimization},\ }\href@noop {} {\bibfield  {journal}
  {\bibinfo  {journal} {J. Mach. Learn. Res.}\ }\textbf {\bibinfo {volume}
  {17}},\ \bibinfo {pages} {1} (\bibinfo {year} {2016})}\BibitemShut {NoStop}%
\bibitem [{\citenamefont {Agrawal}\ \emph {et~al.}(2018)\citenamefont
  {Agrawal}, \citenamefont {Verschueren}, \citenamefont {Diamond},\ and\
  \citenamefont {Boyd}}]{Agrawal.etal2018}%
  \BibitemOpen
  \bibfield  {author} {\bibinfo {author} {\bibfnamefont {A.}~\bibnamefont
  {Agrawal}}, \bibinfo {author} {\bibfnamefont {R.}~\bibnamefont
  {Verschueren}}, \bibinfo {author} {\bibfnamefont {S.}~\bibnamefont
  {Diamond}},\ and\ \bibinfo {author} {\bibfnamefont {S.}~\bibnamefont
  {Boyd}},\ }\bibfield  {title} {\bibinfo {title} {A rewriting system for
  convex optimization problems},\ }\href
  {https://doi.org/10.1080/23307706.2017.1397554} {\bibfield  {journal}
  {\bibinfo  {journal} {J. Control Decis.}\ }\textbf {\bibinfo {volume} {5}},\
  \bibinfo {pages} {42} (\bibinfo {year} {2018})}\BibitemShut {NoStop}%
\bibitem [{\citenamefont {O'Donoghue}\ \emph {et~al.}(2016)\citenamefont
  {O'Donoghue}, \citenamefont {Chu}, \citenamefont {Parikh},\ and\
  \citenamefont {Boyd}}]{ODonoghue.etal2016}%
  \BibitemOpen
  \bibfield  {author} {\bibinfo {author} {\bibfnamefont {B.}~\bibnamefont
  {O'Donoghue}}, \bibinfo {author} {\bibfnamefont {E.}~\bibnamefont {Chu}},
  \bibinfo {author} {\bibfnamefont {N.}~\bibnamefont {Parikh}},\ and\ \bibinfo
  {author} {\bibfnamefont {S.}~\bibnamefont {Boyd}},\ }\bibfield  {title}
  {\bibinfo {title} {Conic optimization via operator splitting and homogeneous
  self-dual embedding},\ }\href {https://doi.org/10.1007/s10957-016-0892-3}
  {\bibfield  {journal} {\bibinfo  {journal} {J. Optim. Theory Appl.}\ }\textbf
  {\bibinfo {volume} {169}},\ \bibinfo {pages} {1042} (\bibinfo {year}
  {2016})}\BibitemShut {NoStop}%
\bibitem [{\citenamefont {{MOSEK ApS}}(2020)}]{Mosek}%
  \BibitemOpen
  \bibfield  {author} {\bibinfo {author} {\bibnamefont {{MOSEK ApS}}},\ }\href
  {https://docs.mosek.com/9.2/pythonapi/index.html} {\emph {\bibinfo {title}
  {MOSEK Optimizer API for Python 9.2}}} (\bibinfo {year} {2020})\BibitemShut
  {NoStop}%
\bibitem [{\citenamefont {Chen}\ \emph {et~al.}(2019)\citenamefont {Chen},
  \citenamefont {Chu}, \citenamefont {Qian},\ and\ \citenamefont
  {Shen}}]{Chen.etal2019}%
  \BibitemOpen
  \bibfield  {author} {\bibinfo {author} {\bibfnamefont {L.}~\bibnamefont
  {Chen}}, \bibinfo {author} {\bibfnamefont {D.}~\bibnamefont {Chu}}, \bibinfo
  {author} {\bibfnamefont {L.}~\bibnamefont {Qian}},\ and\ \bibinfo {author}
  {\bibfnamefont {Y.}~\bibnamefont {Shen}},\ }\bibfield  {title} {\bibinfo
  {title} {Separability of completely symmetric states in a multipartite
  system},\ }\href {https://doi.org/10.1103/PhysRevA.99.032312} {\bibfield
  {journal} {\bibinfo  {journal} {Phys. Rev. A}\ }\textbf {\bibinfo {volume}
  {99}},\ \bibinfo {pages} {032312} (\bibinfo {year} {2019})}\BibitemShut
  {NoStop}%
\bibitem [{\citenamefont {Navascu\'es}\ \emph {et~al.}(2009)\citenamefont
  {Navascu\'es}, \citenamefont {Owari},\ and\ \citenamefont
  {Plenio}}]{Navascues.etal2009b}%
  \BibitemOpen
  \bibfield  {author} {\bibinfo {author} {\bibfnamefont {M.}~\bibnamefont
  {Navascu\'es}}, \bibinfo {author} {\bibfnamefont {M.}~\bibnamefont {Owari}},\
  and\ \bibinfo {author} {\bibfnamefont {M.~B.}\ \bibnamefont {Plenio}},\
  }\bibfield  {title} {\bibinfo {title} {Power of symmetric extensions for
  entanglement detection},\ }\href {https://doi.org/10.1103/PhysRevA.80.052306}
  {\bibfield  {journal} {\bibinfo  {journal} {Phys. Rev. A}\ }\textbf {\bibinfo
  {volume} {80}},\ \bibinfo {pages} {052306} (\bibinfo {year}
  {2009})}\BibitemShut {NoStop}%
\bibitem [{\citenamefont {Berta}\ \emph {et~al.}()\citenamefont {Berta},
  \citenamefont {Borderi}, \citenamefont {Fawzi},\ and\ \citenamefont
  {Scholz}}]{Berta.etal2018}%
  \BibitemOpen
  \bibfield  {author} {\bibinfo {author} {\bibfnamefont {M.}~\bibnamefont
  {Berta}}, \bibinfo {author} {\bibfnamefont {F.}~\bibnamefont {Borderi}},
  \bibinfo {author} {\bibfnamefont {O.}~\bibnamefont {Fawzi}},\ and\ \bibinfo
  {author} {\bibfnamefont {V.}~\bibnamefont {Scholz}},\ }\href@noop {}
  {\bibinfo {title} {Semidefinite programming hierarchies for quantum error
  correction}},\ \Eprint {https://arxiv.org/abs/1810.12197} {arXiv:1810.12197}
  \BibitemShut {NoStop}%
\bibitem [{\citenamefont {Pedersen}(1989)}]{Pedersen1989}%
  \BibitemOpen
  \bibfield  {author} {\bibinfo {author} {\bibfnamefont {G.~K.}\ \bibnamefont
  {Pedersen}},\ }\href@noop {} {\emph {\bibinfo {title} {Analysis now}}}\
  (\bibinfo  {publisher} {Springer-Verlag, New York},\ \bibinfo {year}
  {1989})\BibitemShut {NoStop}%
\bibitem [{\citenamefont {Nakata}(2010)}]{Nakata2010}%
  \BibitemOpen
  \bibfield  {author} {\bibinfo {author} {\bibfnamefont {M.}~\bibnamefont
  {Nakata}},\ }\bibfield  {title} {\bibinfo {title} {A numerical evaluation of
  highly accurate multiple-precision arithmetic version of semidefinite
  programming solver: {SDPA-GMP}, -{QD} and -{DD}},\ }in\ \href
  {https://doi.org/10.1109/CACSD.2010.5612693} {\emph {\bibinfo {booktitle}
  {2010 IEEE International Symposium on Computer-Aided Control System
  Design}}}\ (\bibinfo {year} {2010})\ pp.\ \bibinfo {pages}
  {29--34}\BibitemShut {NoStop}%
\end{thebibliography}%

\end{document}